\documentclass[11pt]{article}
\usepackage[a4paper]{geometry}
\usepackage{abstract}
\usepackage[utf8]{inputenc}
\usepackage{array}
\usepackage{float}
\usepackage{booktabs}
\usepackage{mathtools}
\usepackage{multirow}
\usepackage{amsmath}
\usepackage{adjustbox}
\usepackage{afterpage}

\usepackage{capt-of}
\usepackage{amsthm}
\usepackage{amssymb}
\usepackage{graphicx}
\usepackage[authoryear]{natbib}
\usepackage[unicode=true,
 bookmarks=false,
 breaklinks=false,pdfborder={0 0 1},backref=section,colorlinks=false]
 {hyperref}
\usepackage{amsfonts}  

\usepackage{tikz}
\usetikzlibrary{arrows.meta,positioning,fit,calc}
\usepackage{rotating}   

\usepackage{pdflscape}
\usepackage{siunitx}
\newcommand{\best}[1]{{\bfseries\num{#1}}}

\makeatletter

\newtheorem{remark}{Remark}

\usepackage{bm}
\usepackage{authblk}
\usepackage[inline]{enumitem}
\usepackage{subcaption}
\usepackage{xcolor}
\usepackage{cleveref}
\newtheorem{theorem}{Theorem}
\newtheorem{proposition}{Proposition}
\newtheorem{assumption}{Assumption}

\title{Breaking the Dimensional Barrier: Dynamic Portfolio Choice with Parameter Uncertainty via Pontryagin Projection}

\author[1]{Jeonggyu Huh}
\author[2]{Hyeng Keun Koo}
\affil[1]{\small Department of Mathematics, Sungkyunkwan University, Suwon, Republic of Korea}
\affil[2]{\small Department of Financial Engineering, Ajou University, Suwon, Republic of Korea}

\makeatother

\begin{document}

\maketitle

\begin{abstract}
We study continuous-time CRRA portfolio choice in diffusion markets with estimated and hence uncertain coefficients. Nature draws a latent parameter $\theta\sim q$ at time $0$ and keeps it fixed; the investor never observes $\theta$ and must commit to a single \emph{$\theta$-blind} policy maximizing an ex--ante objective, treating $q$ as a decision-time input.

We propose a simulation-only two-stage solver. \textbf{Stage~1 (DPO)} performs BPTT-based stochastic gradient ascent through an Euler simulator while sampling $\theta$ only inside the simulator. \textbf{Stage~2 (Pontryagin projection)} aggregates costate blocks across $\theta\sim q$ and enforces the \emph{$q$-aggregated} stationarity condition within the deployable class; the resulting correction can be amortized via interactive distillation. We refer to the full Stage~1+Stage~2 pipeline as \textbf{PG--DPO}.

We prove a uniform conditional BPTT--PMP correspondence and a residual-based policy-gap bound with explicit discretization and Monte Carlo error terms. Experiments on high-dimensional Gaussian drift-uncertainty and factor-driven benchmarks show that projection stabilizes learning and accurately recovers analytic decision-time references, while a model-free PPO baseline remains far from the targets.
\end{abstract}

\vspace{0.5cm}

\section{Introduction}
\label{sec:intro}

Dynamic portfolio choice is typically studied under the assumption that investment opportunities are known, yielding Merton-type closed-form rules and their intertemporal extensions
\citep{markowitz1952portfolio,samuelson1969lifetime,merton1969lifetime,merton1971optimum,karatzas1998methods,campbell2002strategic}.
In practice, expected returns and other inputs are estimated from finite samples, are unstable across specifications, and can be highly uncertain at the time a decision is made
\citep{goyalwelch2008,campbellthompson2008,lettauvanNieuwerburgh2008shifts,rapach2010forecast}.

A classical response is \emph{learning and belief-state control}: investors update posteriors over time and hedge against future revisions in investment opportunities
\citep{kandelstambaugh1996predictability,barberis2000investor,xia2001learning,brandt2005portfolio}.
While conceptually appealing, belief-state formulations lead to high-dimensional (sometimes infinite-dimensional) state spaces and are difficult to implement in realistic multi-asset settings
\citep{lakner1995utility,bensoussan1985optimal,bjork2010partial,pham2009continuous}.

This paper takes a complementary \emph{decision-time} perspective centered on deployability:
\begin{quote}
\emph{Given an externally supplied description of parameter uncertainty at the decision time, what constitutes an optimal \textbf{deployable} portfolio rule when the true parameter is latent and cannot be conditioned on during trading?}
\end{quote}
We treat estimation risk as an exogenous \emph{decision-time uncertainty object}---a fixed law $q(d\theta)$ over market parameters produced by an external pipeline (e.g., bootstrap, subsampling, or asymptotic approximations)
\citep{efron1979bootstrap,efron1994bootstrap,breiman1996bagging,vandervaart1998asymptotic}.
Nature draws $\theta\sim q$ at time $0$ and keeps it fixed over the horizon, but the investor never observes $\theta$ and must commit to a single \emph{$\theta$-blind} policy.

This information structure changes the relevant first-order condition. Full-information policies such as the Merton rule are defined conditional on the realized parameter and are infeasible when $\theta$ is latent
\citep{merton1969lifetime,merton1971optimum,campbell2002strategic}.
More fundamentally, optimality must be enforced \emph{within the deployable class}: admissible perturbations are $\theta$-blind, so stationarity requires the \emph{$q$-expectation} of the Hamiltonian gradient to vanish along the controlled state process. In portfolio problems, this $q$-aggregated Pontryagin condition reduces to a statewise linear system whose solution defines a canonical \emph{deployable projected rule}. Averaging full-information optima does not, in general, produce such a rule; optimality does not commute with averaging
\citep{pastor2000portfolio,avramov2002stock,cremers2002stock,brown1978bayes}.

These observations lead to three contributions. First, we derive the $q$-aggregated Pontryagin first-order condition for the fixed-$q$, $\theta$-blind problem and show how it induces a deployable projection map in portfolio settings
\citep{pontryagin1962the,yong1999stochastic,fleming2006controlled,pham2009continuous}.
Second, we propose a simulation-only two-stage solver that avoids dynamic programming, which is infeasible in the high-dimensional uncertain regimes of interest
\citep{bellman1961adaptive,kushner2001numerical,han2018solving,beck2019machine}.
Stage~1 (\textbf{DPO}) trains a $\theta$-blind policy by stochastic gradient ascent using exact discrete-time gradients computed by backpropagation through time through an Euler simulator, while sampling $\theta$ only inside the simulator. Stage~2 applies a \textbf{Pontryagin projection} by aggregating costate objects across $\theta\sim q$ and projecting the warm policy onto the deployable $q$-aggregated stationarity condition, yielding a structured correction that can be amortized via interactive distillation
\citep{werbos1990backpropagation,williams1989learning,pardouxpeng1990adapted,ma1999forward,hinton2015distilling,jin2020pontryagin}.
We refer to the full Stage~1+Stage~2 pipeline as \textbf{PG--DPO}. Third, we provide a theory that matches this computational design: we prove a conditional BPTT--PMP correspondence that is uniform over compact parameter sets and includes the costate blocks required by the projection step, and we establish a residual-based policy-gap bound with explicit discretization and Monte Carlo error terms. In Gaussian benchmarks, the projected rule recovers familiar shrinkage and horizon-dependent hedging effects connected to classical return-predictability insights
\citep{kandelstambaugh1996predictability,campbell1999consumption,kim1996dynamic,lynch2001portfolio}.

Empirically, in high-dimensional settings the Stage~2 projection stabilizes learning and accurately recovers analytic decision-time references, while a model-free PPO baseline remains far from the targets under the same deployability restriction
\citep{schulman2017proximal,WZZ2020,JZ2022b}.
Overall, the paper reframes portfolio choice under parameter uncertainty as a problem of \emph{commitment under imperfect information} at the decision time, rather than explicit belief-state learning.

Section~\ref{sec:merton-uncertainty} formulates the fixed-$q$ ex--ante problem, contrasts $\theta$-conditional and $q$-aggregated PMP conditions, and presents Gaussian decision-time reference models.
Section~\ref{sec:pgdpo-uncertainty} develops Stage~1 DPO and Stage~2 Pontryagin projection (PG--DPO) for latent $\theta$ and proves the policy-gap bound.
Sections~\ref{sec:hd-geometry} and~\ref{sec:hedging-recovery} report numerical experiments, and the appendix collects proofs and implementation details.


\section{Dynamic Portfolio Choice in Estimated Diffusion Markets with Latent Parameter Uncertainty}
\label{sec:merton-uncertainty}

In this section we formalize a continuous-time CRRA portfolio problem in an estimated diffusion market with latent parameter uncertainty.
Rather than fixing an inference architecture, we take as input an exogenous decision-time law $q(d\theta)$ over $\Theta$ and work with a latent,
time-constant parameter $\theta\sim q$ (see Section~\ref{subsec:model-objective} for the full model and objective).
We adopt a \emph{$\theta$-blind commitment} viewpoint: the investor commits at the decision time to a single Markov feedback rule
$\pi_t=\bar\pi(t,X_t,Y_t)$ and does not update $q$ during trading.
This restriction shifts the relevant optimality notion from infeasible $\theta$-conditional (full-information) conditions to a
\emph{$q$-aggregated} Pontryagin stationarity condition within the $\theta$-blind admissible class (Section~\ref{subsec:pmp-latent}).
Closed-form Gaussian decision-time reference models used for validation are collected in Section~\ref{subsec:gaussian-drift}.

\subsection{Model and ex--ante objective in estimated diffusion markets}
\label{subsec:model-objective}

We interpret time $0$ as the \emph{decision time} and study portfolio choice on a fixed horizon $[0,T]$.
The uncertainty law $q(d\theta)$ is an exogenous input available at the decision time and is treated as fixed over $[0,T]$
(no online belief updates during trading).

\paragraph{Estimated diffusion market family (conditional on a latent parameter).}
Drift and volatility are not assumed known. Instead, we consider a general multi-asset, multi-factor diffusion family indexed by
$\theta\in\Theta\subset\mathbb{R}^k$, where $\theta$ represents the (possibly high-dimensional) parameter produced by an estimation procedure.
Conditional on $\theta$, the $d$ risky assets and an $m$-dimensional factor process $Y_t$ evolve as
\begin{align}
  \frac{dS_t}{S_t}
  &= r\,\mathbf{1}\,dt
      + b\big(Y_t,\theta\big)\,dt
      + \sigma\big(Y_t,\theta\big)\, dW_t,
  \qquad S_0 \in (0,\infty)^d,
  \label{eq:general-model-S}\\[0.3em]
  dY_t
  &= a\big(Y_t,\theta\big)\,dt
      + \beta\big(Y_t,\theta\big)\,dW_t^Y,
  \qquad Y_0 = y\in\mathbb{R}^m,
  \label{eq:general-model-Y}
\end{align}
where $W$ and $W^Y$ are Brownian motions (possibly of different dimension) that may be instantaneously correlated.
For bookkeeping, let $W$ be $d_W$-dimensional, $W^Y$ be $d_Y$-dimensional, and write
\begin{equation}
  d\langle W, W^Y\rangle_t = \rho\,dt,
  \qquad \rho\in\mathbb{R}^{d_W\times d_Y}.
  \label{eq:WWY-corr}
\end{equation}
We write the instantaneous covariance and return--factor cross-covariance as
\begin{equation}
  \Sigma(y,\theta) := \sigma(y,\theta)\sigma(y,\theta)^\top,
  \qquad
  \Sigma_{SY}(y,\theta) := \sigma(y,\theta)\,\rho\,\beta(y,\theta)^\top,
  \label{eq:sigma-sy-def}
\end{equation}
so that $\Sigma(y,\theta)\in\mathbb{R}^{d\times d}$ and $\Sigma_{SY}(y,\theta)\in\mathbb{R}^{d\times m}$.

\paragraph{Uncertainty law $q(d\theta)$ and information structure.}
We summarize the uncertainty of the parameter $\theta$, estimated from finite samples, by a probability distribution
\begin{equation}
  q(d\theta).
  \label{eq:theta-law}
\end{equation}
We treat $q$ as inference-agnostic; it may represent resampling-based distributions (e.g., bootstrap \citep{efron1979bootstrap,efron1994bootstrap}, bagging \citep{breiman1996bagging}) or asymptotic approximations \citep{vandervaart1998asymptotic}.
For our purposes, $q$ is a fixed decision-time input describing plausible market parameters over $[0,T]$.

\begin{remark}[Latent parameter, observability, and admissible controls]
\label{rem:latent-theta}
All processes are defined on a filtered probability space
$(\Omega,\mathcal{F},\mathbb{F}=(\mathcal{F}_t)_{0\le t\le T},\mathbb{P})$
satisfying the usual conditions and carrying $(W,W^Y,\theta)$.

We interpret $\theta$ as a latent (unobserved) market parameter: at the (shifted) decision time $0$,
Nature draws an $\mathcal{F}_0$-measurable random variable $\theta\sim q$ (independent of the Brownian drivers)
and keeps it fixed over $[0,T]$. The investor knows $q$ but does not observe the realized $\theta$,
so deployable portfolio rules cannot take $\theta$ as an input.

We consider the observable market filtration (a subfiltration of $\mathbb{F}$)
\begin{equation}
\mathcal{F}_t^{\mathrm{obs}}
:=\sigma\{(S_s,Y_s):0\le s\le t\},
\qquad 0\le t\le T,
\label{eq:obs-filtration}
\end{equation}
where $\sigma\{\cdot\}$ denotes the $\sigma$-field generated by the observed asset and factor paths (with the usual augmentation).
Admissible portfolio processes are required to be progressively measurable with respect to $(\mathcal{F}_t^{\mathrm{obs}})$.

Throughout the paper we \emph{restrict} attention to the Markov feedback subclass
\begin{equation}
  \mathcal{A}^{\mathrm{fb}}
  :=
  \Big\{
    \pi\in\mathcal{A}^{\mathrm{obs}}:\ \exists\,\bar\pi:[0,T]\times(0,\infty)\times\mathbb{R}^m\to\mathbb{R}^d
    \ \text{s.t.}\ \pi_t=\bar\pi(t,X_t,Y_t)
  \Big\},
  \label{eq:markov-feedback-class}
\end{equation}
where $\mathcal{A}^{\mathrm{obs}}$ is defined below. This restriction reflects a fixed-$q$ commitment model: the investor forms $q$
at the decision time and does not perform online filtering/belief-state updates during $[0,T]$.

Whenever we display $\theta$-conditional (full-information) controls or sensitivity objects, they are computed under frozen-$\theta$
simulations and are used only for offline diagnostics; the deployed policy class and the learned policy remain $\theta$-blind.

\emph{Notation.} Throughout Section~\ref{sec:merton-uncertainty}, $\rho$ denotes the Brownian correlation matrix in
\eqref{eq:WWY-corr} (and \eqref{eq:ou-corr}); $T$ denotes the terminal time/horizon; and $K$ denotes the OU mean-reversion matrix in
\eqref{eq:theta-ou-belief}. We do not reuse $\rho$, $T$, or $K$ for unrelated quantities elsewhere.
\end{remark}

\paragraph{Wealth dynamics and admissibility (given $\theta$).}
For any fixed $\theta$, the corresponding wealth dynamics under a portfolio process $\pi_t(\omega)\in\mathbb{R}^d$
adapted to $\mathcal{F}_t^{\mathrm{obs}}$ are
\begin{equation}
  \frac{dX_t^\pi}{X_t^\pi}
  = \Big(
      r + \pi_t^\top b(Y_t,\theta)
    \Big) dt
    + \pi_t^\top \sigma(Y_t,\theta)\, dW_t,
  \label{eq:wealth-uncertain-theta}
\end{equation}
and we denote by $\mathcal{A}^{\mathrm{obs}}$ the set of progressively measurable portfolio processes adapted to
$(\mathcal{F}_t^{\mathrm{obs}})$ for which \eqref{eq:wealth-uncertain-theta} admits a (strictly) positive wealth solution.
In the Markovian feedback case $\pi\in\mathcal{A}^{\mathrm{fb}}$ one may think of $\pi_t=\bar\pi(t,X_t,Y_t)$.

\paragraph{Ex--ante objective under latent $\theta$ (and simulator viewpoint).}
The investor evaluates policies under an \emph{ex--ante} objective that averages over both diffusion noise for fixed $\theta$
and the parametric uncertainty encoded by \eqref{eq:theta-law}:
\begin{equation}
  J(\pi)
  :=
  \mathbb{E}_{\theta\sim q}
  \bigg[
    \mathbb{E}\big[ U(X_T^\pi) \,\big|\, \theta \big]
  \bigg]
  =
  \int_{\Theta}
    \mathbb{E}\big[ U(X_T^\pi) \,\big|\, \theta \big]
  \,q(d\theta).
  \label{eq:ex-ante-objective}
\end{equation}
The corresponding optimization problem (under our feedback restriction) is
\begin{equation}
  \sup_{\pi\in\mathcal{A}^{\mathrm{fb}}} J(\pi).
  \label{eq:ex-ante-optimization-new}
\end{equation}
Whenever it exists, we denote by
\[
  \pi^{\star,\mathrm{blind}}
  \in \arg\max_{\pi\in\mathcal{A}^{\mathrm{fb}}} J(\pi)
\]
an optimal $\theta$-blind feedback for the fixed-$q$ commitment problem \eqref{eq:ex-ante-optimization-new}.
For each fixed $\theta$, we also write $\pi^{\star,\theta}$ for the (infeasible) $\theta$-conditional \emph{full-information}
optimal control that would be available if $\theta$ were observed.

The $\theta$-blind constraint makes \eqref{eq:ex-ante-optimization-new} strictly harder than solving a separate control problem for each fixed
$\theta$, since the latter yields a $\theta$-indexed full-information family. Ex--ante averaging in \eqref{eq:ex-ante-objective}
can also create gradient cancellation across heterogeneous parameter draws when one attempts to learn a single global policy end-to-end.
While an $\mathcal{F}_t^{\mathrm{obs}}$-adapted policy could, in principle, filter $\theta$ online and solve a belief-state control problem
(see, e.g., \citet{bensoussan1985optimal,pham2017dynamic}), we do \emph{not} pursue that formulation here.

Approximating the outer expectation in \eqref{eq:ex-ante-objective} amounts to sampling $\theta\sim q$ \emph{inside the simulator}
(once per trajectory or once per update), running \eqref{eq:general-model-S}--\eqref{eq:general-model-Y} under that frozen draw,
and updating a $\theta$-blind feedback policy to perform well \emph{on average} over such draws. This is the setting targeted by the
simulation-based DPO and PG--DPO methods developed in Section~\ref{sec:pgdpo-uncertainty}.

\paragraph{Decision-time commitment and (optional) plug-in replanning.}
Throughout, we treat $q(d\theta)$ as an exogenous \emph{decision-time} input and keep it fixed over the trading horizon $[0,T]$:
within a single run, we do not update $q$ using newly observed returns or additional data, and we do not hedge against future changes of the
uncertainty description. If the decision is revisited at a later calendar time $t_0\in(0,T)$, we interpret this as starting a \emph{new}
decision-time problem on the remaining horizon with an externally supplied law $q_{t_0}$.

Such a law may coincide with a static input (Section~\ref{subsubsec:gaussian-drift-static}),
may arise as a \emph{model-implied predictive} (prior) propagation of an initial uncertainty through factor dynamics without conditioning on
new observations (Section~\ref{subsubsec:gaussian-drift-ou-belief}), or may be produced by an external filtering/estimation routine that outputs
a time-indexed uncertainty description (Appendix~\ref{app:kalman-bucy}). Our model and algorithms are conditional on whichever $q$ is supplied at
the decision time and always freeze that input over the ensuing horizon.

\subsection{Pontryagin optimality under latent parameters: full-information vs.\ aggregated conditions}
\label{subsec:pmp-latent}

This subsection records the Hamiltonian structure underlying our projection step and clarifies what ``Pontryagin first-order conditions''
mean when the market parameter $\theta$ is latent and admissible controls are $\theta$-blind.
Throughout, $q(d\theta)$ denotes the \emph{decision-time} input law (Section~\ref{subsec:model-objective}) and is treated as fixed over the trading
horizon $[0,T]$. All expectations in this subsection are taken conditional on the decision-time information (suppressed in notation).

In particular, we distinguish between
(i) \emph{$\theta$-conditional} (full-information) criticality conditions that would apply if $\theta$ were observable (and are therefore infeasible
under latent $\theta$), and (ii) \emph{$q$-aggregated} criticality conditions that characterize stationarity \emph{within the $\theta$-blind admissible
class} for the fixed-$q$ ex--ante objective.
Our discussion follows standard stochastic control/PMP arguments for diffusion control \citep[e.g.][]{yong1999stochastic,fleming2006controlled,pham2009continuous}.
We also comment on the relationship to partial-information (belief-state) PMP, but we do not develop that formulation here.

\paragraph{A $\theta$-conditional (full-information) Hamiltonian and first-order condition (infeasible under latent $\theta$).}
Fix $\theta\in\Theta$ and suppose, for the moment, that $\theta$ were observable to the controller.
In Markovian settings with sufficient smoothness, the $\theta$-conditional value function $V^{\star,\theta}(t,x,y)$ satisfies an HJB equation whose
\emph{control Hamiltonian} (the part depending on $\pi$) can be written explicitly using \eqref{eq:sigma-sy-def}.
For clarity, we define the \emph{control Hamiltonian map} as a function of generic ``sensitivity'' inputs
$(\psi_x,\psi_{xx},\psi_{xy})$:
\begin{equation}
  \mathcal{H}_\theta^{\mathrm{ctrl}}(t,x,y,\pi;\,\psi_x,\psi_{xx},\psi_{xy})
:=
x\,\pi^\top b(y,\theta)\,\psi_x
+ \frac12 x^2\,\pi^\top \Sigma(y,\theta)\,\pi\,\psi_{xx}
+ x\,\pi^\top \Sigma_{SY}(y,\theta)\,\psi_{xy}.
  \label{eq:ctrl-hamiltonian}
\end{equation}
In the full-information HJB interpretation, we evaluate \eqref{eq:ctrl-hamiltonian} at
$(\psi_x,\psi_{xx},\psi_{xy})=(V_x,V_{xx},V_{xy})$.
In the Pontryagin interpretation used below, the same affine map is evaluated at the corresponding $\theta$-conditional adjoint/sensitivity objects
(denoted $(p^\theta,p_x^\theta,p_y^\theta)$ along a trajectory); in smooth Markov regimes these objects coincide with the indicated value derivatives.

The pointwise first-order condition for an interior optimizer is therefore
\begin{equation}
  \partial_\pi \mathcal{H}_\theta^{\mathrm{ctrl}}
  =
  x\,V_x^{\star,\theta}\,b(y,\theta)
  + x^2\,V_{xx}^{\star,\theta}\,\Sigma(y,\theta)\,\pi
  + x\,\Sigma_{SY}(y,\theta)\,V_{xy}^{\star,\theta}
  \;=\;0,
  \label{eq:foc-fixed-theta}
\end{equation}
where $V_x^{\star,\theta},V_{xx}^{\star,\theta},V_{xy}^{\star,\theta}$ are evaluated at $(t,x,y)$.
Assuming $\Sigma(y,\theta)$ is invertible and $V_{xx}^{\star,\theta}<0$, this yields the closed-form $\theta$-conditional full-information portfolio rule
\begin{equation}
  \pi^{\star,\theta}(t,x,y)
  =
  -\,\frac{1}{x\,V_{xx}^{\star,\theta}(t,x,y)}\,
    \Sigma(y,\theta)^{-1}
    \Big(
      V_x^{\star,\theta}(t,x,y)\,b(y,\theta)
      + \Sigma_{SY}(y,\theta)\,V_{xy}^{\star,\theta}(t,x,y)
    \Big).
  \label{eq:pi-star-fixed-theta}
\end{equation}
This $\theta$-indexed rule is \emph{not deployable} under latent parameters; we record it only as a full-information benchmark and diagnostic reference.
In our setting, deployable policies never take the realized $\theta$ as an input; $\theta$ is accessed only through sampling inside the simulator when
approximating $q$-expectations.

\paragraph{$q$-aggregated Pontryagin condition for the $\theta$-blind ex--ante problem (Markov feedback).}
We now return to the actual setting: $\theta$ is latent, policies are $\theta$-blind, and we restrict attention to the Markov feedback class
$\mathcal{A}^{\mathrm{fb}}$ (Remark~\ref{rem:latent-theta}).
Under this restriction we neither perform online filtering of $\theta$ nor replace $q$ by a time-varying posterior distribution.
Accordingly, the relevant Pontryagin condition is not the $\theta$-conditional criticality \eqref{eq:foc-fixed-theta} enforced pointwise in $\theta$,
but rather a necessary condition for optimality \emph{within the $\theta$-blind admissible class} for the fixed-$q$ objective \eqref{eq:ex-ante-objective}.

To see why ex--ante aggregation enters the first-order condition, take any $\theta$-blind admissible perturbation $h=\{h_t\}_{t\in[0,T]}$
that is progressively measurable with respect to the observation filtration $(\mathcal{F}_t^{\mathrm{obs}})$ and square-integrable, and define
$\pi^\varepsilon:=\pi+\varepsilon h$ for small $\varepsilon$.
For each fixed $\theta$, define $J^\theta(\pi):=\mathbb{E}[U(X_T^\pi)\mid \theta]$. The stochastic maximum principle yields the first-variation identity
\begin{equation}
  \left.\frac{d}{d\varepsilon} J^\theta(\pi^\varepsilon)\right|_{\varepsilon=0}
  =
  \mathbb{E}\!\left[\int_0^T
    \partial_\pi \mathcal{H}^{\mathrm{ctrl}}_\theta
      \big(t,X_t,Y_t,\pi_t;\,p_t^\theta,p_{x,t}^\theta,p_{y,t}^\theta\big)^\top h_t\,dt
  \,\Big|\,\theta\right],
  \label{eq:var-identity-fixed-theta}
\end{equation}
where $\big(p_t^\theta,p_{x,t}^\theta,p_{y,t}^\theta\big)$ denotes the $\theta$-conditional Pontryagin sensitivity/adjoint objects associated with the
\emph{fixed} policy $\pi$ in the frozen-$\theta$ market.
(Equivalently, \eqref{eq:var-identity-fixed-theta} evaluates the same affine map $\partial_\pi\mathcal{H}_\theta^{\mathrm{ctrl}}$ defined by
\eqref{eq:ctrl-hamiltonian}, with $(V_x,V_{xx},V_{xy})$ replaced by the corresponding Pontryagin objects along the controlled trajectory.)

\paragraph{Frozen-$\theta$ adjoint BSDE and the blocks used in projection.}
To make the $\theta$-conditional Pontryagin objects in \eqref{eq:var-identity-fixed-theta} concrete, we briefly record the backward (stochastic)
equation they satisfy for a \emph{fixed} policy in a frozen-$\theta$ market.
Fix a $\theta$-blind feedback policy $\pi\in\mathcal{A}^{\mathrm{fb}}$ and condition on a frozen parameter $\theta$.
Let
\(
\mathcal{G}_t^\theta:=\sigma\big(\theta,\{W_s,W_s^Y:0\le s\le t\}\big)
\)
denote the simulator filtration (with the usual augmentation).
Under standard hypotheses ensuring well-posedness of the stochastic maximum principle (summarized in Assumption~\ref{ass:bsde-regime}),
there exists an adapted triple $(p^\theta,z^\theta,\tilde z^\theta)$ solving the $\theta$-conditional adjoint BSDE
\begin{equation}
\label{eq:adjoint-bsde}
  p_t^\theta
  =
  U'\!\big(X_T^{\pi,\theta}\big)
  +\int_t^T f_\theta\!\big(s,X_s^{\pi,\theta},Y_s^\theta,\pi_s;\,p_s^\theta,z_s^\theta,\tilde z_s^\theta\big)\,ds
  -\int_t^T z_s^\theta\cdot dW_s
  -\int_t^T \tilde z_s^\theta\cdot dW_s^Y,
\end{equation}
for a driver $f_\theta$ determined by the frozen-$\theta$ dynamics and the fixed policy $\pi$
\citep[e.g.][]{yong1999stochastic,pham2009continuous}.
In a smooth Markov regime one may identify $p_t^\theta = V_x^\theta(t,X_t^{\pi,\theta},Y_t^\theta)$ along the controlled state and interpret the
blocks entering the portfolio Hamiltonian gradient \eqref{eq:ctrl-hamiltonian} as
\begin{equation}
\label{eq:adjoint-blocks-def}
  p_{x,t}^\theta := V_{xx}^\theta(t,X_t^{\pi,\theta},Y_t^\theta),
  \qquad
  p_{y,t}^\theta := V_{xy}^\theta(t,X_t^{\pi,\theta},Y_t^\theta).
\end{equation}
Equivalently, one may work directly with the second-order stochastic maximum principle, where $(p_{x,t}^\theta,p_{y,t}^\theta)$ coincide with the
relevant blocks of the second adjoint.

\begin{assumption}[Frozen-$\theta$ Markov/BSDE regime (used in Theorems~\ref{thm:q-agg-foc-theta-blind} and~\ref{thm:bptt-pmp-uncertainty})]
\label{ass:bsde-regime}
Fix a compact parameter set $\Theta_0\subset\Theta$ and a working domain $\mathcal{D}\subset[0,T]\times(0,\infty)\times\mathbb{R}^m$.
For every $\theta\in\Theta_0$ and every $\theta$-blind feedback policy $\pi$ considered in the paper, assume:
\begin{enumerate}[label=(A\arabic*),leftmargin=3.2em]
\item \textbf{Forward well-posedness and moments.}
The frozen-$\theta$ forward system \eqref{eq:wealth-uncertain-theta}--\eqref{eq:general-model-Y} admits a unique strong solution under $\pi$, and for some $p\ge2$,
\[
\sup_{\theta\in\Theta_0}\mathbb{E}\Big[\sup_{t\in[0,T]}\big(|X_t^{\pi,\theta}|^p+\|Y_t^\theta\|^p\big)\Big]<\infty.
\]
\item \textbf{Positivity and utility integrability.}
Wealth stays strictly positive a.s., and $U$ is $C^3$ on $(0,\infty)$ with
\[
\sup_{\theta\in\Theta_0}\mathbb{E}\Big[|U'(X_T^{\pi,\theta})|^2+|U''(X_T^{\pi,\theta})|^2\Big]<\infty.
\]
(For CRRA, this entails suitable negative-moment bounds for $X_T^{\pi,\theta}$.)
\item \textbf{Adjoint BSDE and smooth blocks.}
For each $\theta\in\Theta_0$, the adjoint BSDE \eqref{eq:adjoint-bsde} is well posed in $L^2$ and admits a Markov decoupling field $u^\theta$ on $\mathcal{D}$
of class $C^{1,2}$ in $(t,x,y)$ such that
\(p_t^\theta=u^\theta(t,X_t^{\pi,\theta},Y_t^\theta)\),
\(p_{x,t}^\theta=\partial_x u^\theta(t,X_t^{\pi,\theta},Y_t^\theta)\),
and \(p_{y,t}^\theta=\partial_y u^\theta(t,X_t^{\pi,\theta},Y_t^\theta)\),
with square-integrable bounds.
\item \textbf{Nondegenerate driving noise for conditional projections.}
The block covariance matrix
\(
\Gamma:=\begin{psmallmatrix} I_{d_W} & \rho\\ \rho^\top & I_{d_Y}\end{psmallmatrix}
\)
is positive definite (equivalently, the conditional $L^2$ projections used below are taken onto a nondegenerate basis of the driving increments).
\item \textbf{Uniformity on compacts.}
All Lipschitz/growth/nondegeneracy/smoothness constants in (A1)--(A4) can be chosen uniformly for $\theta\in\Theta_0$.
\end{enumerate}
\end{assumption}

\begin{remark}[What we assume vs.\ what we prove]
\label{rem:assume-vs-prove}
Assumption~\ref{ass:bsde-regime} collects the analytic conditions required to \emph{define} the Pontryagin/BSDE objects and to justify the interchanges
of limits and conditional expectations used later (notably in the BPTT--BSDE correspondence of Theorem~\ref{thm:bptt-pmp-uncertainty}).
Establishing such smooth decoupling fields and uniform bounds for general high-dimensional coefficients is nontrivial; we do not attempt to do so here.
Instead, our contributions are conditional: \emph{given} a regime where the frozen-$\theta$ SMP/BSDE objects are well behaved (uniformly on compacts),
we show that BPTT recovers them in the small-step limit and that the resulting projection yields quantitative residual-based policy-gap control.
\end{remark}

Because both $\pi$ and $h$ are $\theta$-blind, taking the outer expectation over $\theta\sim q$ and using Fubini's theorem gives
\begin{equation}
  \left.\frac{d}{d\varepsilon} J(\pi^\varepsilon)\right|_{\varepsilon=0}
  =
  \mathbb{E}\!\left[\int_0^T
    \mathbb{E}_{\theta\sim q}\!\Big[
      \partial_\pi \mathcal{H}^{\mathrm{ctrl}}_\theta
      \big(t,X_t,Y_t,\pi_t;\,p_t^\theta,p_{x,t}^\theta,p_{y,t}^\theta\big)
    \Big]^\top h_t\,dt
  \right].
  \label{eq:var-identity-aggregated}
\end{equation}
Hence, for an interior $\theta$-blind optimum $\pi^{\star,\mathrm{blind}}$, the first variation must vanish for all such perturbations $h$,
which implies the aggregated first-order condition
\begin{equation}
  \mathbb{E}_{\theta\sim q}\!\Big[
    \partial_\pi \mathcal{H}^{\mathrm{ctrl}}_\theta
    \big(t,X_t,Y_t,\pi_t;\,p_t^\theta,p_{x,t}^\theta,p_{y,t}^\theta\big)
  \Big]
  = 0,
  \qquad \text{a.s.\ for a.e.\ }t\in[0,T].
  \label{eq:agg-foc-process}
\end{equation}
Equation \eqref{eq:agg-foc-process} is the correct necessary condition for the ex--ante problem under the $\theta$-blind constraint.
In particular, it is generally distinct from imposing \eqref{eq:foc-fixed-theta} for each $\theta$ separately, because $\theta$-conditional
criticality cannot be enforced by a single deployable $\theta$-blind policy.

To operationalize \eqref{eq:agg-foc-process} in the Markov feedback class, fix a feedback policy $\pi\in\mathcal{A}^{\mathrm{fb}}$ and,
for each frozen $\theta$, consider the corresponding $\theta$-conditional Pontryagin sensitivity objects
$\big(p_t^\theta, p_{x,t}^\theta, p_{y,t}^\theta\big)$ along the induced state process.
In smooth Markov regimes these coincide with spatial derivatives of a decoupling field and, in particular, reduce to $(V_x,V_{xx},V_{xy})$
in the full-information setting; in our algorithms we estimate them pathwise by automatic differentiation (see Section~\ref{sec:pgdpo-uncertainty}).

For the portfolio Hamiltonian \eqref{eq:ctrl-hamiltonian}, $\partial_\pi\mathcal{H}_\theta^{\mathrm{ctrl}}$ is affine in $\pi$. This motivates
defining the $\theta$-conditional ``projection inputs''
\begin{align}
  A_t^\theta(t,x,y)
  &:= x\,p_{x,t}^\theta(t,x,y)\,\Sigma(y,\theta)\in\mathbb{R}^{d\times d},
  \label{eq:Ahat-theta-sec2}\\
  g_t^\theta(t,x,y)
  &:= p_t^\theta(t,x,y)\,b(y,\theta)
      + \Sigma_{SY}(y,\theta)\,p_{y,t}^\theta(t,x,y)\in\mathbb{R}^{d},
  \label{eq:Ghat-theta-sec2}
\end{align}
and their $q$-aggregated counterparts
\begin{equation}
  A_t(t,x,y):=\mathbb{E}_{\theta\sim q}\!\big[A_t^\theta(t,x,y)\big],
  \qquad
  g_t(t,x,y):=\mathbb{E}_{\theta\sim q}\!\big[g_t^\theta(t,x,y)\big].
  \label{eq:AG-aggregated-def-sec2}
\end{equation}
These objects summarize how the latent parameter affects the first-order stationarity condition through the
$\theta$-conditional sensitivities.

\begin{theorem}[$q$-aggregated first-order condition under latent $\theta$ (deployable $\theta$-blind stationarity)]
\label{thm:q-agg-foc-theta-blind}
Consider the fixed-$q$ ex--ante objective \eqref{eq:ex-ante-objective} over the $\theta$-blind Markov feedback class $\mathcal{A}^{\mathrm{fb}}$.
Assume Assumption~\ref{ass:bsde-regime} holds (in particular, the first-variation identity \eqref{eq:var-identity-fixed-theta} is valid within $\mathcal{A}^{\mathrm{fb}}$ and the associated $\theta$-conditional Pontryagin/BSDE objects exist).
If $\pi^{\star,\mathrm{blind}}$ is a locally optimal interior policy in $\mathcal{A}^{\mathrm{fb}}$, then \eqref{eq:agg-foc-process} holds.
Moreover, in the portfolio setting \eqref{eq:ctrl-hamiltonian}, the aggregated stationarity is equivalent to the statewise linear system
\begin{equation}
  A_t(t,x,y)\,\pi^{\star,\mathrm{blind}}(t,x,y) = -\,g_t(t,x,y),
  \qquad (t,x,y)\in[0,T]\times(0,\infty)\times\mathbb{R}^m,
  \label{eq:q-agg-linear-system-theta-blind}
\end{equation}
(where $A_t,g_t$ are defined by \eqref{eq:AG-aggregated-def-sec2} using the $\theta$-conditional Pontryagin objects generated by $\pi^{\star,\mathrm{blind}}$).
Whenever $A_t(t,x,y)$ is invertible on the working domain, \eqref{eq:q-agg-linear-system-theta-blind} is equivalently expressed as the projected feedback rule
\begin{equation}
  \pi^{\mathrm{agg}}(t,x,y)
  = -\,A_t(t,x,y)^{-1}\,g_t(t,x,y).
  \label{eq:pi-agg-formal}
\end{equation}
\end{theorem}

\begin{proof}[Proof sketch]
Fix an interior locally optimal $\theta$-blind feedback
$\pi^{\star,\mathrm{blind}}\in\mathcal{A}^{\mathrm{fb}}$.
For any $\theta$-blind admissible perturbation $h=\{h_t\}_{t\in[0,T]}$
(progressively measurable w.r.t.\ $(\mathcal{F}_t^{\mathrm{obs}})$ and square-integrable),
set $\pi^\varepsilon:=\pi^{\star,\mathrm{blind}}+\varepsilon h$.

\begin{enumerate}
  \item \textbf{First variation and $q$-aggregation.}
  For each fixed $\theta$, the $\theta$-conditional first-variation identity
  \eqref{eq:var-identity-fixed-theta} is standard for diffusion control under a frozen parameter
  \citep[e.g.][]{yong1999stochastic,fleming2006controlled,pham2009continuous}.
  Taking the outer expectation over $\theta\sim q$ and using Fubini/Tonelli under the standing
  integrability assumptions yields \eqref{eq:var-identity-aggregated}.
  \emph{Remark.} This step implements the fixed-$q$ decision-time commitment model: $\theta$ is sampled only inside the simulator
  to form the ex--ante objective $J(\cdot)$, while admissible perturbations remain $\theta$-blind. We do not formulate the belief-state/partial-information PMP,
  where stationarity would be expressed in terms of time-varying posteriors $q_t$ (cf.\ Remark~\ref{rem:belief-state-vs-agg}).

  \item \textbf{Vanishing first variation $\Rightarrow$ aggregated stationarity.}
  Since $\pi^{\star,\mathrm{blind}}$ is an (interior) local maximizer,
  the first variation must vanish:
  \[
    \left.\frac{d}{d\varepsilon}J(\pi^\varepsilon)\right|_{\varepsilon=0}=0
    \qquad\text{for all such $\theta$-blind perturbations $h$.}
  \]
  Using standard spike-variation/density arguments within the $\theta$-blind admissible class
  (i.e., choosing $h_t$ supported on arbitrarily small time intervals with arbitrary directions),
  we obtain the aggregated Pontryagin condition \eqref{eq:agg-foc-process}, a.s.\ for a.e.\ $t\in[0,T]$.

  \item \textbf{Affine-in-control Hamiltonian $\Rightarrow$ linear system and projected form.}
  For the quadratic portfolio Hamiltonian \eqref{eq:ctrl-hamiltonian}, $\partial_\pi\mathcal{H}^{\mathrm{ctrl}}_\theta$ is affine in $\pi$.
  Substituting the explicit expression for $\partial_\pi\mathcal{H}^{\mathrm{ctrl}}_\theta$ into \eqref{eq:agg-foc-process} and introducing
  the projection inputs \eqref{eq:Ahat-theta-sec2}--\eqref{eq:AG-aggregated-def-sec2} (and using $x>0$ to divide out a common factor) yields
  the statewise linear system \eqref{eq:q-agg-linear-system-theta-blind}.
  Whenever $A_t(t,x,y)$ is invertible on the working domain, this is equivalently expressed as the projected rule
  $\pi^{\mathrm{agg}}(t,x,y)=-A_t(t,x,y)^{-1}g_t(t,x,y)$, i.e.\ \eqref{eq:pi-agg-formal}.
\end{enumerate}
\end{proof}

Note that $\pi^{\mathrm{agg}}$ is generally \emph{not} equal to the naive average $\mathbb{E}_{\theta\sim q}[\pi^{\star,\theta}(t,x,y)]$ of
$\theta$-conditional full-information controls, reflecting the noncommutativity between averaging over $\theta$ and solving a first-order condition.
In particular, even if one could compute $\pi^{\star,\theta}$ for each $\theta$, averaging these infeasible oracles does not, in general, enforce the
deployable $q$-aggregated stationarity \eqref{eq:agg-foc-process}.

\begin{remark}[Relation to belief-state/learning formulations]
\label{rem:belief-state-vs-agg}
If one allows history-dependent policies that explicitly infer $\theta$ from observed returns, a principled partial-information formulation introduces
a time-varying posterior/belief state $q_t(\cdot)=\mathbb{P}(\theta\in\cdot\mid\mathcal{F}_t^{\mathrm{obs}})$.
In such belief-state problems, the corresponding PMP/Hamiltonian criticality condition is expressed in terms of conditional expectations under $q_t$
(or, equivalently, conditional on $\mathcal{F}_t^{\mathrm{obs}}$); see, e.g., \citet{haussmann1987maximum,li1995general,baghery2007maximum}.
We do not pursue that learning/belief-state route here.
Our algorithms and theory target the fixed-$q$, $q$-aggregated projection \eqref{eq:pi-agg-formal} under the $\theta$-blind Markov feedback restriction.
\end{remark}

\subsection{Gaussian decision-time references}
\label{subsec:gaussian-drift}

This subsection collects Gaussian benchmarks that yield closed-form \emph{decision-time reference} allocations.
Fix a decision time $t\in[0,T)$ and let an external routine provide an $\mathcal{F}_t^{\mathrm{obs}}$-measurable \emph{input law} $q_t$ for the uncertain premium parameter/state.
We treat $q_t$ as a decision-time input (not a belief process controlled/updated within a run): it is held fixed over the remaining horizon $[t,T]$.
The resulting portfolios are used only as validation targets and controlled diagnostics (not as characterizations of the unrestricted optimum).
If the decision is revisited at a later time $t'$, the reference is recomputed by resolving the decision-time problem with the new input $q_{t'}$ (plug-in replanning).

We present two references.
Section~\ref{subsubsec:gaussian-drift-static} considers \emph{static} drift uncertainty and provides the benchmark used in the high-dimensional
scaling/geometry experiments of Section~\ref{sec:hd-geometry}.
Section~\ref{subsubsec:gaussian-drift-ou-belief} considers a mean-reverting (OU) premium factor with Gaussian decision-time uncertainty for the current factor level;
OU propagation and time-averaging induce a horizon-dependent effective Gaussian premium law, yielding a tractable reference used in the hedging-demand recovery study
of Section~\ref{sec:hedging-recovery}.
All derivations are deferred to Appendix~\ref{app:gaussian-references-derivations}.

\subsubsection{Static Gaussian drift uncertainty}
\label{subsubsec:gaussian-drift-static}

\paragraph{Benchmark.}
Let $d$ risky assets satisfy, for $s\in[t,T]$,
\begin{equation}
  \frac{dS_s}{S_s}
  = r\,\mathbf{1}\,ds + \theta\,ds + \Sigma^{1/2} dW_s,
  \qquad S_t\in(0,\infty)^d,
  \label{eq:gaussian-model}
\end{equation}
where $\Sigma\in\mathbb{R}^{d\times d}$ is symmetric positive definite and $\theta\in\mathbb{R}^d$ is a latent constant excess-return vector.
At decision time $t$, uncertainty about $\theta$ is summarized by the input law
\begin{equation}
  \theta \sim q_t(d\theta).
  \label{eq:theta-law-gaussian-toy}
\end{equation}
A $\theta$-blind portfolio fraction process $\pi_s\in\mathbb{R}^d$ generates wealth on $[t,T]$:
\begin{equation}
  \frac{dX_s^\pi}{X_s^\pi}
  = \big(r + \pi_s^\top \theta\big)\,ds
    + \pi_s^\top \Sigma^{1/2}\,dW_s,
  \qquad X_t^\pi=x>0.
  \label{eq:wealth-gaussian}
\end{equation}
We evaluate the decision-time objective by mixing over $q_t$:
\begin{equation}
  J_t(\pi)
  :=\mathbb{E}_{\theta\sim q_t}\Big[\mathbb{E}\big[U(X_T^\pi)\mid \theta,\,\mathcal{F}_t^{\mathrm{obs}}\big]\Big].
  \label{eq:exante-J-gaussian}
\end{equation}

\paragraph{Decision-time constant-portfolio reference.}
To obtain a transparent closed-form target, we restrict to portfolio fractions constant over $[t,T]$:
\begin{equation}
  \pi_s \equiv \pi \in \mathbb{R}^d, \qquad s\in[t,T],
  \label{eq:constant-portfolio-restriction}
\end{equation}
and take CRRA utility $U(x)=x^{1-\gamma}/(1-\gamma)$ with $\gamma>1$.

If the decision-time input is Gaussian,
\begin{equation}
  q_t = \mathcal{N}(m_t,P_t),
  \qquad P_t\succeq 0,
  \label{eq:q-gaussian}
\end{equation}
then the constant-portfolio reference is
\begin{equation}
  \pi_{q_t,\gamma}^{\mathrm{const}}(t,\tau)
  =
  \big(\gamma\Sigma + (\gamma-1)\tau\,P_t\big)^{-1} m_t,
  \qquad \tau:=T-t,\quad (\gamma>1).
  \label{eq:pi_q_const_gaussian_ref}
\end{equation}
This exhibits the familiar \emph{Gaussian shrinkage}: decision-time uncertainty $P_t$ dampens exposures, increasingly so as the remaining horizon $\tau$ grows.

\subsubsection{Mean-reverting Gaussian premium and an induced horizon-dependent reference}
\label{subsubsec:gaussian-drift-ou-belief}

\paragraph{OU premium factor.}
Let $Y_s\in\mathbb{R}^m$ follow
\begin{equation}
  dY_s
  =
  K(\bar y - Y_s)\,ds + \Xi\,dW_s^Y,
  \qquad
  K\in\mathbb{R}^{m\times m}\ \text{Hurwitz and invertible},
  \label{eq:theta-ou-belief}
\end{equation}
and let risky excess returns satisfy
\begin{equation}
  dR_s
  :=
  \frac{dS_s}{S_s}-r\mathbf{1}\,ds
  =
  B Y_s\,ds + \Sigma^{1/2}\,dW_s,
  \label{eq:obs-returns-belief}
\end{equation}
with instantaneous correlation
\begin{equation}
  d\langle W, W^Y\rangle_s = \rho\,ds,
  \qquad
  \rho\in\mathbb{R}^{d\times m}.
  \label{eq:ou-corr}
\end{equation}

\paragraph{Decision-time input law and induced effective premium law.}
At decision time $t$, we take as input a Gaussian law for the current factor level,
\begin{equation}
  Y_t \sim q_t(dy) = \mathcal{N}(m_t,P_t),
  \qquad P_t\succeq 0.
  \label{eq:qt-factor-law}
\end{equation}
Define the remaining-horizon integrated factor and horizon-averaged effective premium:
\begin{equation}
  I_{t,T} := \int_t^T Y_s\,ds \in \mathbb{R}^m,
  \qquad
  \bar\theta_{t,\tau} := \frac{1}{\tau}\,B\,I_{t,T}\in\mathbb{R}^d,
  \qquad \tau:=T-t.
  \label{eq:integrated-premium-tau}
\end{equation}
Under the linear--Gaussian OU model, holding the input law \eqref{eq:qt-factor-law} fixed over $[t,T]$ induces a horizon-dependent Gaussian law
\begin{equation}
  \bar\theta_{t,\tau} \sim \mathcal{N}\!\big(m_{\bar\theta}(t,\tau),P_{\bar\theta}(t,\tau)\big),
  \label{eq:theta-bar-law-tau}
\end{equation}
where explicit formulas for $m_{\bar\theta}(t,\tau)$ and $P_{\bar\theta}(t,\tau)$ are given in Appendix~\ref{app:gaussian-references-derivations}.

\paragraph{Decision-time constant-portfolio reference (CRRA, $\gamma>1$).}
Restricting to constant portfolio fractions $\pi_s\equiv\pi$ on $[t,T]$ as above,
the Gaussian decision-time reference under CRRA ($\gamma>1$) is characterized by the linear system
\begin{equation}
  \Big(\gamma \tau\Sigma + (\gamma-1)\big(B\,C_I(t,\tau)\,B^\top + M_{\mathrm{cross}}(\tau)\big)\Big)\,
  \pi_{q_t,\gamma}^{\mathrm{const}}(t,\tau)
  = B\,m_I(t,\tau),
  \qquad (\gamma>1),
  \label{eq:pi-q-const-ou}
\end{equation}
where $m_I(t,\tau):=\mathbb{E}[I_{t,T}]$, $C_I(t,\tau):=\mathrm{Cov}(I_{t,T})$, and the cross term $M_{\mathrm{cross}}(\tau)$ captures the effect of
instantaneous return--factor correlation $\rho\neq0$ (hedging-demand channel). Closed-form expressions for $m_I(t,\tau)$, $C_I(t,\tau)$, and
$M_{\mathrm{cross}}(\tau)$ are collected in Appendix~\ref{app:gaussian-references-derivations}.

When $\rho=0$, we have $M_{\mathrm{cross}}(\tau)=0$ and the reference reduces to the independence-case shrinkage form obtained by plugging the effective Gaussian law
\eqref{eq:theta-bar-law-tau} into the static Gaussian formula \eqref{eq:pi_q_const_gaussian_ref}.

\section{Pontryagin--Guided Policy Optimization under Latent Parameter Uncertainty}
\label{sec:pgdpo-uncertainty}

We solve the fixed-$q$ ex--ante portfolio problem of Section~\ref{sec:merton-uncertainty} under a latent, time-constant parameter $\theta\sim q$.
The deployed controller is $\theta$-blind: it depends only on observable states $(t,X_t,Y_t)$ and never conditions on the realized $\theta$.
Uncertainty enters only through sampling $\theta$ \emph{inside the simulator}.

We use a deterministic policy network $\pi_{\varphi}:\mathcal{D}\to\mathbb{R}^d$, parameterized by $\varphi$, where
$\mathcal{D}\subset [0,T]\times(0,\infty)\times\mathbb{R}^m$ denotes a working domain inside the state space.

\medskip
\noindent\textbf{Algorithmic overview and terminology.}
Our full solver is \emph{Pontryagin--guided direct policy optimization} (PG--DPO) and consists of two stages.
\textbf{Stage~1 (DPO)} performs stochastic gradient ascent on the ex--ante objective
$J(\varphi)=\mathbb{E}\big[U\big(X_T^{\pi_\varphi,\theta}\big)\big]$ by differentiating through an Euler simulator via BPTT.
\textbf{Stage~2 (Pontryagin guidance / projection)} aggregates the BPTT adjoint (costate) blocks across $\theta\sim q$ and enforces the $q$-aggregated Pontryagin
first-order condition, yielding a deployable projected control.
We refer to Stage~1 alone as \emph{DPO} and to the two-stage procedure (Stage~1+2) as \emph{PG--DPO}.

Section~\ref{subsec:pgdpo-bptt-pmp} establishes the conditional BPTT--PMP correspondence and specifies the adjoint blocks needed for the projection stage.
Section~\ref{subsec:ppgdpo-uncertainty} develops the $q$-aggregated projection and a residual-based policy-gap bound.
Practical coupling mechanisms (residual-form projection and interactive distillation) are deferred to Appendix~\ref{app:coupling}.

\subsection{DPO via stochastic gradient ascent and a conditional BPTT--PMP correspondence}
\label{subsec:pgdpo-bptt-pmp}

\paragraph{Setup and objectives (frozen $\theta$, deployable $\theta$-blind feedback).}
A latent parameter $\theta\in\Theta$ is sampled from a fixed law $q(d\theta)$ \emph{inside the simulator} and kept frozen along each simulated trajectory.
A deployable portfolio policy is a $\theta$-blind Markov feedback rule represented by a neural network
\begin{equation}
  \pi_\varphi:\ \mathcal{D}\to\mathbb{R}^d,
  \qquad \mathcal{D}\subset [0,T]\times(0,\infty)\times\mathbb{R}^m,
  \qquad (t,x,y)\mapsto \pi_\varphi(t,x,y),
  \qquad \varphi\in\mathbb{R}^p,
  \label{eq:markov-feedback-policy}
\end{equation}
which does \emph{not} take $\theta$ as an input.
Conditional on a fixed frozen $\theta$, the controlled state $(X_t^{\pi_\varphi,\theta},Y_t^\theta)_{t\in[0,T]}$ evolves as
\begin{align}
  \frac{dX_t^{\pi_\varphi,\theta}}{X_t^{\pi_\varphi,\theta}}
  &= \Big(
      r + \pi_t^\top b\big(Y_t^\theta,\theta\big)
    \Big)\,dt
      + \pi_t^\top \sigma\big(Y_t^\theta,\theta\big)\,dW_t,
  \qquad X_0=x>0,
  \label{eq:wealth-theta-section2}\\[0.3em]
  dY_t^{\theta}
  &= a\big(Y_t^{\theta},\theta\big)\,dt
      + \beta\big(Y_t^{\theta},\theta\big)\,dW_t^Y,
  \qquad Y_0=y\in\mathbb{R}^m,
  \label{eq:factor-theta-section2}
\end{align}
where $(W,W^Y)$ may be instantaneously correlated as in Section~\ref{sec:merton-uncertainty}.
For each fixed $\theta$ we evaluate $\pi_\varphi$ by the conditional objective
\begin{equation}
  J^\theta(\varphi)
  :=
  \mathbb{E}\big[ U\big(X_T^{\pi_\varphi,\theta}\big) \,\big|\, \theta \big],
  \label{eq:conditional-objective-theta}
\end{equation}
where the expectation is over Brownian paths in \eqref{eq:wealth-theta-section2}--\eqref{eq:factor-theta-section2}.
The fixed-$q$ ex--ante objective is
\begin{equation}
  J(\varphi)
  :=
  \mathbb{E}_{\theta\sim q}\big[J^\theta(\varphi)\big]
  =
  \mathbb{E}\Big[ U\big(X_T^{\pi_\varphi,\theta}\big)\Big],
  \label{eq:exante-objective-section3}
\end{equation}
where the last expectation is joint over $\theta\sim q$ and $(W,W^Y)$.
Thus $\sup_\varphi J(\varphi)$ is a stochastic optimization problem in which $\theta$ is sampled inside the simulator while the policy remains $\theta$-blind.

\paragraph{Discretization, sampling over $\theta$, and the baseline DPO update.}
We discretize the horizon using an Euler scheme.
In training we may sample initial states $z_0^{(i)}=(t_0^{(i)},x_0^{(i)},y_0^{(i)})$ from a user-chosen distribution $\nu$ on
$[0,T)\times(0,\infty)\times\mathbb{R}^m$.
Given $t_0^{(i)}$, we discretize $[t_0^{(i)},T]$ into $N$ steps with step size
\[
  \Delta t^{(i)} := \frac{T-t_0^{(i)}}{N},
  \qquad
  t_k^{(i)} := t_0^{(i)} + k\Delta t^{(i)},\quad k=0,\dots,N.
\]
For episode $i$, we denote the discrete state by $(X_k^{(i)},Y_k^{(i)})_{k=0,\dots,N}$ and write
$\theta^{(i)}$ for the frozen parameter used to generate that simulated environment.
Given $\pi_\varphi$ and Brownian increments, the mapping
\[
  \big(
    x_0^{(i)},y_0^{(i)},\theta^{(i)},
    \{\Delta W_k^{(i)},\Delta W_k^{Y,(i)}\}_{k=0}^{N-1},
    \varphi
  \big)
  \longmapsto
  U\big(X_N^{(i)}\big)
\]
is a finite computational graph; automatic differentiation thus computes exact \emph{discrete-time} gradients
$\nabla_\varphi U(X_N^{(i)})$ for the chosen discretization.

A typical \emph{DPO} update samples a mini-batch $\{z_0^{(i)}\}_{i=1}^M\sim\nu$, samples latent parameters
either independently per episode ($\theta^{(i)}\sim q$) or shared within the batch (one $\theta\sim q$ reused across $i$),
simulates Euler rollouts, and performs stochastic gradient ascent.
Suppressing episode indices and writing $\theta$ for the frozen draw used in the rollout, the Euler recursion is
\begin{align*}
  Y_{k+1}
  &= Y_k + a\big(Y_k,\theta\big)\Delta t
     + \beta\big(Y_k,\theta\big)\Delta W_k^{Y},\\[0.3em]
  X_{k+1}
  &= X_k +
     X_k\Big(
       r + \pi_\varphi\big(t_k,X_k,Y_k\big)^\top
           b\big(Y_k,\theta\big)
     \Big)\Delta t
     +
     X_k\,
     \pi_\varphi\big(t_k,X_k,Y_k\big)^\top
     \sigma\big(Y_k,\theta\big)\Delta W_k,
\end{align*}
starting from $(t_0,X_0,Y_0)=(t_0^{(i)},x_0^{(i)},y_0^{(i)})$.

\begin{remark}[Positivity of wealth in the discrete simulator]
\label{rem:positive-euler}
In continuous time, the wealth process in \eqref{eq:wealth-theta-section2} stays strictly positive for any admissible portfolio.
A naive additive Euler step may nevertheless generate negative $X_k$ for coarse steps or aggressive controls, making CRRA utilities (and the higher-order
derivatives used in BPTT) ill-defined.
In implementations we therefore use a positivity-preserving discretization (e.g.\ exponential Euler on log-wealth) or apply a mild truncation/clipping so that
$X_k\ge \underline x>0$ on the computation graph.
The correspondence results below should be read for such a positivity-preserving simulator; equivalently, they apply on the event
$\{\min_{0\le k\le N}X_k>0\}$.
\end{remark}

The episode reward is
\begin{equation}
  J^{(i)}(\varphi)
  :=
  U\big(X_N^{(i)}\big),
  \label{eq:episode-reward}
\end{equation}
and BPTT computes $\nabla_\varphi J^{(i)}(\varphi)$.
The policy parameters are then updated (e.g.\ by Adam) as
\begin{equation}
  \varphi \leftarrow \varphi
    + \alpha \,\frac{1}{M}\sum_{i=1}^M \nabla_\varphi J^{(i)}(\varphi).
  \label{eq:dpo-update}
\end{equation}
Under standard interchange conditions (dominated convergence / uniform integrability for the discretized simulator),
this is an unbiased stochastic gradient for the \emph{discretized} ex--ante objective; as $\Delta t\to 0$ it approximates the continuous-time objective.

\paragraph{BPTT adjoints and their interpretation: first adjoint and second-adjoint blocks.}
BPTT returns not only $\nabla_\varphi J^{(i)}(\varphi)$ but also adjoint variables with respect to intermediate state variables.
For a single episode (suppressing $i$ and $\theta$ in notation), define the pathwise wealth adjoint
\begin{equation}
  p_k
  :=
  \frac{\partial U(X_N)}{\partial X_k},
  \qquad k=0,\dots,N,
  \label{eq:pathwise-costate-def}
\end{equation}
which is the discrete-time analogue of the \emph{wealth-component first adjoint}.
For the Pontryagin projection in Section~\ref{subsec:ppgdpo-uncertainty}, we also require curvature-type objects.
In smooth Markov regimes these correspond to $(V_{xx},V_{xy})$; in the stochastic maximum principle they are naturally
interpreted as \emph{blocks of the second adjoint}.
Accordingly, we track the following discrete-time blocks (computed by higher-order automatic differentiation):
\begin{equation}
  p_{x,k}
  :=
  \frac{\partial p_k}{\partial X_k}
  \;=\;
  \frac{\partial^2 U(X_N)}{\partial X_k^2},
  \qquad
  p_{y,k}
  :=
  \frac{\partial p_k}{\partial Y_k}
  \;=\;
  \frac{\partial^2 U(X_N)}{\partial Y_k\,\partial X_k},
  \qquad k=0,\dots,N.
  \label{eq:pathwise-costate-derivatives-def}
\end{equation}
(Here $p_{y,k}\in\mathbb{R}^m$ is the wealth--factor cross block.)
These are the discrete-time objects that enter the affine-in-control Hamiltonian gradient and hence the projection step.

\medskip
\noindent\textbf{Pathwise vs.\ SMP-adapted adjoints.}
The quantities in \eqref{eq:pathwise-costate-def}--\eqref{eq:pathwise-costate-derivatives-def} are \emph{pathwise} (sample-wise) derivatives on the simulated computation graph and, for $k<N$, they need not be adapted to time-$t_k$ information (they can depend on future increments).
To connect them to the stochastic maximum principle (which yields adapted adjoint processes and associated martingale terms),
one performs the standard conditional $L^2$ projection onto the span of the driving increments on each step, yielding a discrete-time BSDE representation.
\emph{Concretely}, for a fixed frozen $\theta$ and the simulator filtration
\[
  \mathcal{G}_{t_k}^\theta := \sigma\!\big(\theta,\{W_s,W_s^Y:0\le s\le t_k\}\big)
  \quad \text{(with the usual augmentation)},
\]
we define the \emph{adapted (projected) discrete adjoints}
\[
  p_k^{\Delta t,\theta} := \mathbb{E}[p_k \mid \mathcal{G}_{t_k}^\theta],\qquad
  p_{x,k}^{\Delta t,\theta} := \mathbb{E}[p_{x,k} \mid \mathcal{G}_{t_k}^\theta],\qquad
  p_{y,k}^{\Delta t,\theta} := \mathbb{E}[p_{y,k} \mid \mathcal{G}_{t_k}^\theta].
\]
Theorem~\ref{thm:bptt-pmp-uncertainty} concerns the convergence of these adapted discrete objects (together with the associated one-step BSDE martingale coefficients) to the continuous-time Pontryagin adjoints.

\begin{theorem}[BPTT--BSDE/PMP correspondence (conditional on $\theta$, including second-adjoint blocks; uniform on compacts)]
\label{thm:bptt-pmp-uncertainty}
Fix a compact $\Theta_0\subset\Theta$ and a $\theta$-blind policy network $\pi_\varphi$.
Assume Assumption~\ref{ass:bsde-regime} holds for $\pi=\pi_\varphi$ and every $\theta\in\Theta_0$, and that the discrete simulator used for BPTT
preserves positivity as in Remark~\ref{rem:positive-euler}.

For each $\theta\in\Theta_0$, let $(p_t^\theta,z_t^\theta,\tilde z_t^\theta)$ denote the wealth-component adjoint BSDE solution \eqref{eq:adjoint-bsde}
under $\pi_\varphi$, and let $(p_{x,t}^\theta,p_{y,t}^\theta)$ denote the blocks \eqref{eq:adjoint-blocks-def} (equivalently, the relevant blocks of the
second adjoint along the controlled trajectory).

Let $(X_k,Y_k)_{k=0}^N$ be the Euler simulator on a grid $(t_k)_{k=0}^N$ with step size $\Delta t$, and let $(p_k,p_{x,k},p_{y,k})$ be the
(pathwise, generally non-adapted) BPTT adjoint blocks \eqref{eq:pathwise-costate-def}--\eqref{eq:pathwise-costate-derivatives-def}.
Define their adapted projections by conditional expectations under the simulator filtration $\mathcal{G}_{t_k}^\theta$:
\[
  p_k^{\Delta t,\theta}:=\mathbb{E}[p_k\mid \mathcal{G}_{t_k}^\theta],\qquad
  p_{x,k}^{\Delta t,\theta}:=\mathbb{E}[p_{x,k}\mid \mathcal{G}_{t_k}^\theta],\qquad
  p_{y,k}^{\Delta t,\theta}:=\mathbb{E}[p_{y,k}\mid \mathcal{G}_{t_k}^\theta].
\]
Let $(Z_k^\theta,\tilde Z_k^\theta)$ be the one-step conditional $L^2$ projection coefficients of $p_{k+1}$ onto
$\mathrm{span}\{ \Delta W_k,\Delta W_k^Y\}$ given $\mathcal{G}_{t_k}^\theta$ (so that the orthogonal residual is $L^2$-uncorrelated with the
increments).
Form the piecewise-constant interpolations on $[t_k,t_{k+1})$:
\[
  p_t^{\Delta t,\theta}:=p_k^{\Delta t,\theta},\quad
  p_{x,t}^{\Delta t,\theta}:=p_{x,k}^{\Delta t,\theta},\quad
  p_{y,t}^{\Delta t,\theta}:=p_{y,k}^{\Delta t,\theta},\quad
  z_t^{\Delta t,\theta}:=Z_k^\theta,\quad
  \tilde z_t^{\Delta t,\theta}:=\tilde Z_k^\theta.
\]
Then, as $\Delta t\to0$,
\begin{align*}
\sup_{\theta\in\Theta_0}\sup_{t\in[0,T]}
\mathbb{E}\Big[\,|p_t^{\Delta t,\theta}-p_t^\theta|^2
+|p_{x,t}^{\Delta t,\theta}-p_{x,t}^\theta|^2
+\|p_{y,t}^{\Delta t,\theta}-p_{y,t}^\theta\|^2\,\Big]
&\to 0,\\
\sup_{\theta\in\Theta_0}
\mathbb{E}\!\left[\int_0^T \Big(\|z_t^{\Delta t,\theta}-z_t^\theta\|^2
+\|\tilde z_t^{\Delta t,\theta}-\tilde z_t^\theta\|^2\Big)\,dt\right]
&\to 0.
\end{align*}
In particular, the convergence bounds can be taken uniform over $\theta$ in compact subsets of $\Theta$.
\end{theorem}

\begin{proof}[Proof sketch]
Fix $\theta\in\Theta$ and work conditionally on $\theta$ under the augmented simulator filtration
$\mathcal{G}_t^\theta:=\sigma(\theta,\{W_s,W_s^Y:0\le s\le t\})$ (usual augmentation).
The argument follows the deterministic-parameter BPTT--BSDE/PMP correspondence of \citet{huh2025breaking},
applied conditionally on $\theta$; uniformity on compact $\Theta_0\subset\Theta$ is obtained by assuming the required
Lipschitz/growth/nondegeneracy/smoothness constants are uniform over $\theta\in \Theta_0$.

\begin{enumerate}
  \item \textbf{Pathwise BPTT adjoints on the Euler computation graph.}
  For fixed $\theta$, simulate the Euler discretization of \eqref{eq:wealth-theta-section2}--\eqref{eq:factor-theta-section2}.
  BPTT computes the pathwise derivatives of the terminal utility along the full discrete computation graph, producing
  the pathwise adjoint blocks $(p_k,p_{x,k},p_{y,k})$ as in \eqref{eq:pathwise-costate-def}--\eqref{eq:pathwise-costate-derivatives-def}.
  These quantities are generally $\mathcal{G}_T^\theta$-measurable (they depend on future increments) and hence are not adapted.

  \item \textbf{One-step conditional $L^2$ projection $\Rightarrow$ adapted discrete BSDE variables.}
  On each step, project $p_{k+1}$ onto $\mathrm{span}\{1,\Delta W_k,\Delta W_k^Y\}$ in conditional $L^2$ given $\mathcal{G}_{t_k}^\theta$:
  \[
    p_{k+1}
    =
    \mathbb{E}[p_{k+1}\mid \mathcal{G}_{t_k}^\theta]
    + Z_k^\theta\,\Delta W_k
    + \tilde Z_k^\theta\,\Delta W_k^Y
    + \varepsilon_{k+1}^\theta,
  \]
  where $\varepsilon_{k+1}^\theta$ is orthogonal (in conditional $L^2$) to the span.
  Define the adapted discrete adjoints by conditional expectations,
  \[
    p_k^{\Delta t,\theta}:=\mathbb{E}[p_k\mid \mathcal{G}_{t_k}^\theta],\qquad
    p_{x,k}^{\Delta t,\theta}:=\mathbb{E}[p_{x,k}\mid \mathcal{G}_{t_k}^\theta],\qquad
    p_{y,k}^{\Delta t,\theta}:=\mathbb{E}[p_{y,k}\mid \mathcal{G}_{t_k}^\theta].
  \]
  By construction, $(p_k^{\Delta t,\theta},Z_k^\theta,\tilde Z_k^\theta)$ satisfy a discrete-time BSDE representation.

  \item \textbf{Identification of the discrete driver; second-adjoint blocks and projection residual.}
  Substituting the projection in Step 2 into the one-step BPTT recursion and taking conditional expectations yields that the discrete BSDE driver
  coincides with the Euler discretization of the $\theta$-conditional adjoint BSDE under the fixed policy $\pi_\varphi$, up to one-step remainder terms.
  In particular, in a smooth Markov regime (spelled out in Appendix~\ref{app:bptt-pmp-proof}), the continuous-time adjoint admits a decoupling field
  $p_t^\theta=u^\theta(t,X_t^\theta,Y_t^\theta)$ and the sensitivity blocks needed in the portfolio Hamiltonian gradient coincide along trajectories with
  the spatial derivatives $\partial_x u^\theta$ and $\partial_y u^\theta$, yielding the interpretation of $(p_{x,t}^\theta,p_{y,t}^\theta)$.
  The projection residual $\varepsilon_{k+1}^\theta$ corresponds to higher-order chaos components and is controlled by the same one-step Taylor/chaos estimates
  as in \citet{huh2025breaking}, contributing only higher-order discretization error (uniformly on compacts).

  \item \textbf{Convergence as $\Delta t\to 0$ and uniformity on compacts.}
  Standard Euler stability for the forward SDE together with convergence of the associated discrete BSDE scheme yields, as $\Delta t\to 0$,
  mean-square convergence of the adapted discrete adjoint blocks $(p^{\Delta t,\theta},p_x^{\Delta t,\theta},p_y^{\Delta t,\theta})$ to
  $(p^\theta,p_{x}^\theta,p_{y}^\theta)$ along the controlled trajectory, and convergence of the martingale coefficients
  $(Z^{\Delta t,\theta},\tilde Z^{\Delta t,\theta})$ to $(Z^\theta,\tilde Z^\theta)$ in the natural $L^2([0,T]\times\Omega)$ sense.
  If the coefficient/policy regularity and nondegeneracy constants are uniform for $\theta\in \Theta_0$, then the convergence bounds hold with constants independent
  of $\theta\in \Theta_0$.
\end{enumerate}

A complete proof (including explicit regularity assumptions, the one-step remainder control, and the uniform-on-compacts bookkeeping) is provided in
Appendix~\ref{app:bptt-pmp-proof}.
\end{proof}

Across $\theta\sim q$, these Pontryagin objects form a $\theta$-indexed family.
Baseline \emph{DPO} trains against the ex--ante objective \eqref{eq:exante-objective-section3} by repeatedly sampling $\theta\sim q$ inside the simulator,
while the deployable policy remains $\theta$-blind.

\subsection{Stage~2 of PG--DPO under latent $\theta$: $q$-aggregated projection, diagnostics, and a residual-based policy-gap bound}
\label{subsec:ppgdpo-uncertainty}

\noindent\textbf{Stage~2 (Pontryagin guidance) is a projection step.}
Given a warm-up deployable $\theta$-blind feedback policy
$\pi^{\mathrm{warm}}=\pi_{\varphi^{\mathrm{warm}}}$ obtained from Stage~1 (\emph{DPO}),
we estimate $\theta$-conditional Pontryagin adjoint objects under frozen $\theta\sim q$ by BPTT/Monte Carlo,
aggregate them across $\theta$, and construct a deployable $\theta$-blind policy by projecting onto the $q$-aggregated stationarity condition
derived in Section~\ref{subsec:pmp-latent}.
Together with Stage~1, this projection completes \emph{PG--DPO}.

Equivalently, Stage~2 can be viewed as a plug-in approximation of one application of the population projection map
$\mathcal{T}(\pi):=-A_\pi^{-1}g_\pi^{\mathrm{mix}}$ (defined below) evaluated at $\pi=\pi^{\mathrm{warm}}$ on the working domain.
The key point is that the aggregated first-order condition is \emph{affine} in the portfolio control; hence it induces a statewise linear system
and, on a suitable working domain, a concrete projection map from (estimated) adjoint objects to a portfolio rule.

Throughout, the adjoint objects used in projection consist of the \emph{wealth-component first adjoint} together with the \emph{relevant blocks of the second adjoint}
(cf.\ Section~\ref{subsec:pgdpo-bptt-pmp}): in smooth Markov regimes these correspond to $(V_x,V_{xx},V_{xy})$ evaluated along the controlled state.

\medskip
\noindent\textbf{Working domain and norms.}
Fix a measurable working state domain $\mathcal{D}\subset [0,T]\times(0,\infty)\times\mathbb{R}^m$ (e.g.\ a training/evaluation band)
and a reference measure $\mu$ on $\mathcal{D}$ (e.g.\ an empirical state distribution induced by rollouts).
For $h:\mathcal{D}\to\mathbb{R}^n$ we write
\[
  \|h\|_{L^2(\mu)} := \Big(\int_\mathcal{D} \|h(z)\|^2\,\mu(dz)\Big)^{1/2},
  \qquad z=(t,x,y),
\]
and for $\theta$-indexed families (used when tracking frozen-$\theta$ quantities in analysis/inspection),
\begin{equation}
  \|f\|_{L^2(q\otimes\mu)}
  :=
  \bigg(
    \int_{\Theta}\int_{\mathcal{D}} \|f^\theta(z)\|^2\,\mu(dz)\,q(d\theta)
  \bigg)^{1/2}.
  \label{eq:policy-gap-norm-exante}
\end{equation}

\paragraph{Mixed-moment $q$-aggregation under a warm-up policy.}
By Theorem~\ref{thm:q-agg-foc-theta-blind}, any locally optimal interior \emph{deployable} $\theta$-blind policy
$\pi^{\star,\mathrm{blind}}$ for the fixed-$q$ ex--ante problem satisfies the $q$-aggregated stationarity condition
\eqref{eq:agg-foc-process}.
In the portfolio Hamiltonian \eqref{eq:ctrl-hamiltonian}, this stationarity is equivalent to a statewise linear system and hence to the projected form
\eqref{eq:pi-agg-formal} on the working domain (under invertibility of the aggregated curvature term).
Stage~2 implements a practical approximation of this projection by estimating the relevant aggregated Pontryagin objects under the fixed warm-up policy
$\pi^{\mathrm{warm}}=\pi_{\varphi^{\mathrm{warm}}}$.

Fix a query state $z=(t,x,y)\in \mathcal{D}$ and a frozen parameter $\theta$.
We simulate trajectories under $\pi^{\mathrm{warm}}$ and compute discrete-time adjoint blocks by autodiff/BPTT on the Euler simulator;
averaging over $M_{\mathrm{MC}}$ trajectories yields Monte Carlo estimates
\begin{equation}
  \widehat p_t^\theta(z),\qquad
  \widehat p_{x,t}^\theta(z),\qquad
  \widehat p_{y,t}^\theta(z),
  \label{eq:mc-costate-estimates}
\end{equation}
where $\widehat p_t^\theta$ corresponds to the wealth-component first adjoint and
$(\widehat p_{x,t}^\theta,\widehat p_{y,t}^\theta)$ correspond to the blocks of the second adjoint relevant for the Hamiltonian gradient.
In the analysis, these are interpreted as estimators of the corresponding $\theta$-conditional Pontryagin objects
(cf.\ the conditional BPTT--PMP correspondence in Section~\ref{subsec:pgdpo-bptt-pmp}).

Using these, define the $\theta$-conditional estimated projection inputs
\begin{align}
  \widehat A_t^\theta(t,x,y)
  &:= x\,\widehat p_{x,t}^\theta(t,x,y)\,\Sigma(y,\theta)\in\mathbb{R}^{d\times d},
  \label{eq:Ahat-theta}\\
  \widehat g_t^\theta(t,x,y)
  &:= \widehat p_t^\theta(t,x,y)\, b(y,\theta)
      + \Sigma_{SY}(y,\theta)\,\widehat p_{y,t}^\theta(t,x,y)\in\mathbb{R}^{d},
  \label{eq:Ghat-theta}
\end{align}
where $\Sigma(\cdot,\theta)$ and $\Sigma_{SY}(\cdot,\theta)$ are the instantaneous covariance objects appearing in the $\theta$-conditional Hamiltonian
(as in Section~\ref{sec:merton-uncertainty}).

Aggregating across $\theta\sim q$ (approximated in practice by sampling $M_\theta$ frozen parameters) gives the \emph{mixed-moment} estimators
\begin{align}
  \widehat A_t(t,x,y)
  &:= \mathbb{E}_{\theta\sim q}\Big[\widehat A_t^\theta(t,x,y)\Big],
  \label{eq:agg-A}\\
  \widehat g_t^{\mathrm{mix}}(t,x,y)
  &:= \mathbb{E}_{\theta\sim q}\Big[\widehat g_t^\theta(t,x,y)\Big].
  \label{eq:agg-G}
\end{align}
Whenever $\widehat A_t(t,x,y)$ is invertible on $\mathcal{D}$, we obtain the mixed-moment projected policy
\begin{equation}
  \widehat\pi^{\mathrm{agg,mix}}(t,x,y)
  :=
  -\,\widehat A_t(t,x,y)^{-1}\,\widehat g_t^{\mathrm{mix}}(t,x,y).
  \label{eq:pi-agg-mix}
\end{equation}

\medskip
\noindent\textbf{Stabilized inversion and projection diagnostics (implementation).}
In practice, finite-sample estimates can make $\widehat A_t$ ill-conditioned or even singular on parts of $\mathcal{D}$.
Our default implementation therefore uses standard stabilizers (Appendix~\ref{app:impl-stabilizers}), including:
(i) ridge regularization $\widehat A_{t,\lambda}:=\widehat A_t+\lambda I_d$ and $\widehat\pi^{\mathrm{agg,mix}}_\lambda:=-\widehat A_{t,\lambda}^{-1}\widehat g_t^{\mathrm{mix}}$,
(ii) condition-number or eigenvalue diagnostics to skip/clip projection on unreliable states (falling back to $\pi^{\mathrm{warm}}$),
and (iii) residual-form projection (Section~\ref{subsubsec:cv-projection}) to reduce variance.
The analysis below is stated for the population projection $\mathcal{T}(\pi):=-A_\pi^{-1}g_\pi^{\mathrm{mix}}$ and for small estimator perturbations,
which is precisely the regime enforced by these stabilizers.

\paragraph{Residual diagnostic and a slab-wise small-gain policy-gap bound.}
To connect \eqref{eq:pi-agg-mix} to a locally optimal deployable $\theta$-blind policy, we measure how well the warm-up policy satisfies the
\emph{population} mixed-moment aggregated stationarity.
Let $(A_\pi,g_\pi^{\mathrm{mix}})$ denote the mixed-moment $q$-aggregated projection inputs induced by a policy $\pi$
(i.e.\ the objects in \eqref{eq:AG-aggregated-def-sec2} evaluated using the $\theta$-conditional Pontryagin adjoint blocks generated by $\pi$).
Define the warm-up aggregated stationarity residual on $\mathcal{D}$ by
\begin{equation}
  r_{\mathrm{FOC,mix}}^{\mathrm{warm}}(t,x,y)
  :=
  A_{\pi^{\mathrm{warm}}}(t,x,y)\,\pi^{\mathrm{warm}}(t,x,y)
  + g_{\pi^{\mathrm{warm}}}^{\mathrm{mix}}(t,x,y),
  \qquad
  \varepsilon_{\mathrm{warm}}^{\mathrm{mix}}
  :=
  \big\|r_{\mathrm{FOC,mix}}^{\mathrm{warm}}\big\|_{L^2(\mu)}.
  \label{eq:residual-warm-mix}
\end{equation}
Note that $r_{\mathrm{FOC,mix}}^{\mathrm{warm}}=0$ is equivalent to $\pi^{\mathrm{warm}}$ being a fixed point of $\mathcal{T}$ on $\mathcal{D}$.
More generally, the population projection admits the residual form
\begin{equation}
  \mathcal{T}(\pi^{\mathrm{warm}})
  \;=\;
  \pi^{\mathrm{warm}}
  \;-\;
  A_{\pi^{\mathrm{warm}}}^{-1}\,r_{\mathrm{FOC,mix}}^{\mathrm{warm}}.
  \label{eq:residual-form-population}
\end{equation}
In practice we monitor the empirical analogue
\(
  \widehat r_{\mathrm{FOC,mix}}^{\mathrm{warm}}
  :=
  \widehat A_t\,\pi^{\mathrm{warm}}+\widehat g_t^{\mathrm{mix}}
\),
and use it as a diagnostic to flag states where projection is unreliable (cf.\ Appendix~\ref{app:impl-stabilizers}).
(Equivalently, $\widehat\pi^{\mathrm{agg,mix}}=\pi^{\mathrm{warm}}-\widehat A_t^{-1}\widehat r_{\mathrm{FOC,mix}}^{\mathrm{warm}}$,
or its stabilized variant, which motivates the residual-form implementation.)

\medskip
\noindent\textbf{BPTT/Monte Carlo error term (one explicit choice).}
To make the discretization/Monte Carlo dependence in \eqref{eq:policy-gap-bound-residual} explicit, one may take for example
\begin{equation}
  \delta_{\mathrm{BPTT}}(\Delta t,M_{\mathrm{MC}},M_\theta)
  :=
  \big\|\widehat A_t-A_{\pi^{\mathrm{warm}}}\big\|_{L^2(\mu)}
  +
  \big\|\widehat g_t^{\mathrm{mix}}-g_{\pi^{\mathrm{warm}}}^{\mathrm{mix}}\big\|_{L^2(\mu)}.
  \label{eq:delta-bptt-explicit}
\end{equation}
Under standard Lipschitz/bounded-moment assumptions (Euler strong error and Monte Carlo sampling error),
one typically has a scaling of the form
\(
  \delta_{\mathrm{BPTT}}(\Delta t,M_{\mathrm{MC}},M_\theta)
  \lesssim
  \Delta t^{1/2}+M_{\mathrm{MC}}^{-1/2}+M_\theta^{-1/2}
\),
up to constants depending on the working domain and stability bounds.

\medskip
\noindent
A technical point is that a \emph{global} small-gain condition of the form $C_1<1$ can be overly restrictive.
Following the slab-wise philosophy in our prior PGDPO analysis (e.g.\ \citet[Appendix B]{huh2025breaking}),
we default to a \emph{time-slab} decomposition of the working domain and close the warm-up gap on each short slab.
Concretely, assume $\mathcal{D}$ carries a time coordinate and fix a partition $0=t_0<t_1<\cdots<t_L=T$ with slab lengths
$\tau_k:=t_k-t_{k-1}$. For notational convenience, write $\mathcal{D}\subset[0,T]\times\mathcal{S}$ for some spatial band
$\mathcal{S}\subset(0,\infty)\times\mathbb{R}^m$, and define
\[
  \mathcal{D}_k:=\mathcal{D}\cap([t_{k-1},t_k]\times\mathcal{S}),
  \qquad
  \mu_k:=\mu|_{\mathcal{D}_k},
  \qquad
  \|f\|_k:=\|f\|_{L^2(\mu_k)}.
\]
We write $\mathcal{T}(\pi):=-A_\pi^{-1}g_\pi^{\mathrm{mix}}$ for the (population) $q$-aggregated projection map.
Theorem~\ref{thm:policy-gap-residual} below shows that, under a mild \emph{slab-wise} local stability regime
(i.e.\ a short-time contraction of $\mathcal{T}$ on each $\mathcal{D}_k$), small residual implies that the projected policy is close
(in $L^2(\mu)$) to a locally optimal deployable $\theta$-blind policy, up to discretization/Monte Carlo error.
The proof combines a projection-map stability bound (Appendix~\ref{app:proj-map-stability}) with a slab-wise closure
(Appendix~\ref{app:slab-small-gain}), in the same spirit as the slab analyses used in \citet{huh2025breaking}.

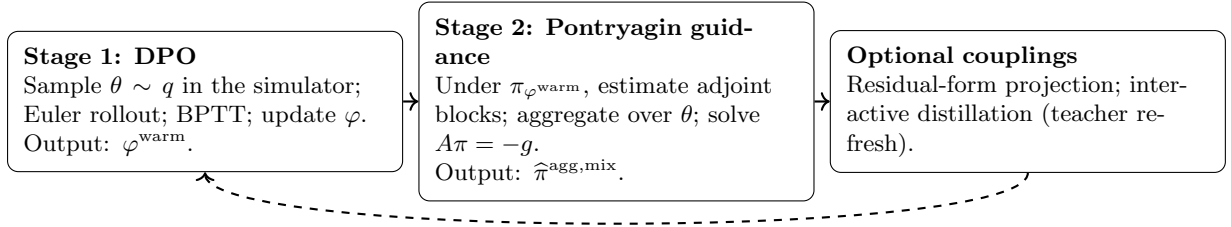
\begin{figure}[t]
\centering
\begin{tikzpicture}[
  font=\footnotesize,
  box/.style={draw, rounded corners, align=left, inner sep=6pt, text width=0.30\textwidth},
  arrow/.style={->, thick},
  darrow/.style={->, thick, dashed}
]
\node[box, anchor=west] (s1) at (0,0) {%
  \textbf{Stage~1: DPO}\\
  Sample $\theta\sim q$ in the simulator; Euler rollout; BPTT; update $\varphi$.\\
  Output: $\varphi^{\mathrm{warm}}$.
};
\node[box, anchor=west] (s2) at (0.34\textwidth,0) {%
  \textbf{Stage~2: Pontryagin guidance}\\
  Under $\pi_{\varphi^{\mathrm{warm}}}$, estimate adjoint blocks; aggregate over $\theta$; solve $A\pi=-g$.\\
  Output: $\widehat\pi^{\mathrm{agg,mix}}$.
};
\node[box, anchor=west] (s3) at (0.68\textwidth,0) {%
  \textbf{Optional couplings}\\
  Residual-form projection; interactive distillation (teacher refresh).
};

\draw[arrow] (s1.east) -- (s2.west);
\draw[arrow] (s2.east) -- (s3.west);
\draw[darrow] (s3.south) .. controls +(0,-0.9) and +(0,-0.9) .. (s1.south);

\end{tikzpicture}
\caption{Two-stage PG--DPO pipeline of Section~\ref{sec:pgdpo-uncertainty}: Stage~1 (DPO) trains a $\theta$-blind policy by BPTT, and Stage~2 applies a $q$-aggregated Pontryagin projection. A detailed schematic (including coupling mechanisms and implementation notes) is provided in Appendix~\ref{app:coupling}, Fig.~\ref{fig:sec3-pipeline-sideways}.}
\label{fig:sec3-pipeline}
\end{figure}

\begin{theorem}[Residual-based ex--ante $\theta$-blind policy-gap bound for PG--DPO
(mixed-moment, deployable, slab-wise local)]
\label{thm:policy-gap-residual}
Let $\pi^{\mathrm{warm}}$ be the Stage~1 (DPO) warm-up policy and let $(\widehat A_t,\widehat g_t^{\mathrm{mix}})$ be the Stage~2 mixed-moment estimators
computed under $\pi^{\mathrm{warm}}$.
Assume Proposition~\ref{prop:proj-map-stability} (Appendix~\ref{app:proj-map-stability}) applies on the working domain $\mathcal{D}$ with
$A=A_{\pi^{\mathrm{warm}}}$ and $g=g_{\pi^{\mathrm{warm}}}^{\mathrm{mix}}$; equivalently, there exist constants $\kappa,M_g>0$ such that
\[
  \|A_{\pi^{\mathrm{warm}}}^{-1}\|_{L^\infty(\mathcal{D})}\le \kappa,\qquad
  \|g_{\pi^{\mathrm{warm}}}^{\mathrm{mix}}\|_{L^\infty(\mathcal{D})}\le M_g,\qquad
  \|\widehat A_t-A_{\pi^{\mathrm{warm}}}\|_{L^\infty(\mathcal{D})}\le (2\kappa)^{-1},
\]
where the last condition is the \emph{stable inversion} regime enforced in practice by ridge/diagnostics/skip.

Let $\pi^{\star,\mathrm{blind}}$ be a locally optimal interior deployable $\theta$-blind policy for the fixed-$q$ ex--ante problem.
Assume there exists a neighborhood $\mathcal{U}$ of $\pi^{\star,\mathrm{blind}}$ in $L^2(\mu)$ such that:
\begin{enumerate}[label=(B\arabic*),leftmargin=3.2em]
\item $\pi^{\mathrm{warm}}\in\mathcal{U}$ and, for all $\pi\in\mathcal{U}$,
\(
  \|A_\pi^{-1}\|_{L^\infty(\mathcal{D})}\le \kappa
\)
and
\(
  \|g_\pi^{\mathrm{mix}}\|_{L^\infty(\mathcal{D})}\le M_g;
\)
\item the \emph{slab-wise Lipschitz gain} of Appendix~\ref{app:slab-small-gain} holds: there exist constants $\bar L_A,\bar L_g>0$ such that for every slab
$\mathcal{D}_k$ and all $\pi_1,\pi_2\in\mathcal{U}$,
\begin{equation}
\label{eq:slab-Lip-AG-main}
  \|A_{\pi_1}-A_{\pi_2}\|_k \le \bar L_A\,\tau_k^{1/2}\,\|\pi_1-\pi_2\|_k,
  \qquad
  \|g_{\pi_1}^{\mathrm{mix}}-g_{\pi_2}^{\mathrm{mix}}\|_k \le \bar L_g\,\tau_k^{1/2}\,\|\pi_1-\pi_2\|_k.
\end{equation}
\end{enumerate}
Define
\begin{equation}
\label{eq:rho-star-def-main}
  \varrho(\tau):=\big(\kappa \bar L_g+\kappa^2 M_g \bar L_A\big)\tau^{1/2},
  \qquad
  \varrho_*:=\max_{1\le k\le L}\varrho(\tau_k),
\end{equation}
and assume the slab partition is chosen so that $\varrho_*<1$.

Let $\widehat\pi^{\mathrm{agg,mix}}$ be the mixed-moment projected policy \eqref{eq:pi-agg-mix} computed from BPTT/Monte Carlo estimates under
$\pi^{\mathrm{warm}}$ (with stabilized inversion as needed), and let $\varepsilon_{\mathrm{warm}}^{\mathrm{mix}}$ be the population residual
\eqref{eq:residual-warm-mix}. Then there exists $C_2>0$ such that
\begin{equation}
  \big\|
    \widehat\pi^{\mathrm{agg,mix}} - \pi^{\star,\mathrm{blind}}
  \big\|_{L^2(\mu)}
  \;\le\;
  \frac{\varrho_*\kappa}{1-\varrho_*}\,
  \varepsilon_{\mathrm{warm}}^{\mathrm{mix}}
  \;+\;
  C_2\,
  \delta_{\mathrm{BPTT}}(\Delta t,M_{\mathrm{MC}},M_\theta),
  \label{eq:policy-gap-bound-residual}
\end{equation}
where one admissible explicit choice of $\delta_{\mathrm{BPTT}}$ is given in \eqref{eq:delta-bptt-explicit}.
Moreover, under the stable inversion regime above one may take, for example,
\(
C_2:=2\kappa+2\kappa^2M_g.
\)
\end{theorem}

\begin{remark}[Population vs.\ empirical residual; scope of the bound]
\label{rem:pop-vs-emp-residual}
The residual $\varepsilon_{\mathrm{warm}}^{\mathrm{mix}}$ in \eqref{eq:policy-gap-bound-residual} is defined using population mixed-moment objects
$(A_{\pi^{\mathrm{warm}}},g_{\pi^{\mathrm{warm}}}^{\mathrm{mix}})$.
In practice we monitor the empirical analogue
$\widehat r_{\mathrm{FOC,mix}}^{\mathrm{warm}}=\widehat A_t\,\pi^{\mathrm{warm}}+\widehat g_t^{\mathrm{mix}}$
and enforce the stable inversion regime by ridge/diagnostics/skip (Appendix~\ref{app:impl-stabilizers}).
Relating $\varepsilon_{\mathrm{warm}}^{\mathrm{mix}}$ to $\|\widehat r_{\mathrm{FOC,mix}}^{\mathrm{warm}}\|_{L^2(\mu)}$ amounts to controlling the estimator error
in $(\widehat A_t,\widehat g_t^{\mathrm{mix}})$ and the size of $\pi^{\mathrm{warm}}$ on $\mathcal{D}$; we do not pursue nonasymptotic concentration bounds here.
\end{remark}

\begin{proof}[Proof sketch]
Define the population mixed-moment projection map on $\mathcal{D}$ by
\[
  \mathcal{T}(\pi)\;:=\;-A_\pi^{-1}g_\pi^{\mathrm{mix}}.
\]
Let the population projected policy be
\[
  \pi^{\mathrm{proj}} \;:=\; \mathcal{T}(\pi^{\mathrm{warm}}) \;=\; -A_{\pi^{\mathrm{warm}}}^{-1}g_{\pi^{\mathrm{warm}}}^{\mathrm{mix}}.
\]

\medskip
\noindent\textbf{Notation (for the proof only).}
To avoid clutter in this proof sketch, we write
\[
  \widehat A := \widehat A_t,
  \qquad
  \widehat g := \widehat g_t^{\mathrm{mix}},
\]
as functions on the working domain $\mathcal{D}$ (with the dependence on $(t,x,y)$ suppressed), so that
$\widehat\pi^{\mathrm{agg,mix}} = -\widehat A^{-1}\widehat g$ is exactly \eqref{eq:pi-agg-mix}.

\medskip
Write the stage-2 (stabilized) estimator-based projected policy as
\[
  \widehat\pi^{\mathrm{agg,mix}} \;:=\; -\widehat A^{-1}\widehat g,
\]
where $(\widehat A,\widehat g)$ are the BPTT/Monte-Carlo estimates computed under $\pi^{\mathrm{warm}}$
(with ridge/diagnostics/skip as needed to make the inverse well-defined on $\mathcal{D}$).

\begin{enumerate}
  \item \textbf{Estimator-to-population stability (Appendix~\ref{app:proj-map-stability}).}
  By the stated uniform stability/invertibility assumptions (in particular, $\widehat A$ is an $L^\infty(\mathcal{D})$-small perturbation of
  $A_{\pi^{\mathrm{warm}}}$ under the adopted stabilizers), Proposition~\ref{prop:proj-map-stability} applies with
  $A=A_{\pi^{\mathrm{warm}}}$ and $g=g_{\pi^{\mathrm{warm}}}^{\mathrm{mix}}$, yielding
  \[
    \big\|\widehat\pi^{\mathrm{agg,mix}}-\pi^{\mathrm{proj}}\big\|_{L^2(\mu)}
    \;\le\; C_2\,\delta_{\mathrm{BPTT}}(\Delta t,M_{\mathrm{MC}},M_\theta),
  \]
  for some constant $C_2$ depending only on the uniform inverse bound $\kappa$ and uniform $L^\infty$ bounds on the population $g$ terms
  (one explicit admissible $\delta_{\mathrm{BPTT}}$ is given in \eqref{eq:delta-bptt-explicit}).

  \item \textbf{Slab-wise contraction of the population projection map (Appendix~\ref{app:slab-small-gain}).}
  On each slab $\mathcal{D}_k$, the slab-wise Lipschitz gain assumption \eqref{eq:slab-Lip-AG-main} together with the uniform bounds on $\mathcal{U}$
  implies (via Proposition~\ref{prop:slab-lipschitz-gain}) that for all $\pi_1,\pi_2\in\mathcal{U}$,
  \[
    \|\mathcal{T}(\pi_1)-\mathcal{T}(\pi_2)\|_k \;\le\; \varrho(\tau_k)\,\|\pi_1-\pi_2\|_k
    \;\le\; \varrho_*\,\|\pi_1-\pi_2\|_k.
  \]

  \item \textbf{Fixed-point identity for the local interior optimum.}
  Since $\pi^{\star,\mathrm{blind}}$ is a locally optimal interior deployable $\theta$-blind policy, it satisfies the aggregated stationarity
  condition on $\mathcal{D}$, equivalently the fixed-point relation
  \[
    \pi^{\star,\mathrm{blind}} = \mathcal{T}(\pi^{\star,\mathrm{blind}}).
  \]
  Therefore, for each slab $\mathcal{D}_k$,
  \[
    \|\pi^{\mathrm{proj}}-\pi^{\star,\mathrm{blind}}\|_k
    \;=\;\|\mathcal{T}(\pi^{\mathrm{warm}})-\mathcal{T}(\pi^{\star,\mathrm{blind}})\|_k
    \;\le\;\varrho_*\,\|\pi^{\mathrm{warm}}-\pi^{\star,\mathrm{blind}}\|_k.
  \]

  \item \textbf{Residual identity and small-gain closure.}
  Let $r_{\mathrm{FOC,mix}}^{\mathrm{warm}}:=A_{\pi^{\mathrm{warm}}}\pi^{\mathrm{warm}}+g_{\pi^{\mathrm{warm}}}^{\mathrm{mix}}$ and
  $\varepsilon_{\mathrm{warm},k}^{\mathrm{mix}}:=\|r_{\mathrm{FOC,mix}}^{\mathrm{warm}}\|_k$.
  By construction of $\pi^{\mathrm{proj}}$,
  \[
    \pi^{\mathrm{warm}}-\pi^{\mathrm{proj}}
    \;=\; A_{\pi^{\mathrm{warm}}}^{-1}\,r_{\mathrm{FOC,mix}}^{\mathrm{warm}},
    \qquad
    \|\pi^{\mathrm{warm}}-\pi^{\mathrm{proj}}\|_k
    \;\le\; \kappa\,\varepsilon_{\mathrm{warm},k}^{\mathrm{mix}}.
  \]
  Using the triangle inequality and the contraction estimate above,
  \[
    \|\pi^{\mathrm{warm}}-\pi^{\star,\mathrm{blind}}\|_k
    \;\le\;
    \|\pi^{\mathrm{warm}}-\pi^{\mathrm{proj}}\|_k
    +
    \|\pi^{\mathrm{proj}}-\pi^{\star,\mathrm{blind}}\|_k
    \;\le\;
    \kappa\,\varepsilon_{\mathrm{warm},k}^{\mathrm{mix}}
    +
    \varrho_*\|\pi^{\mathrm{warm}}-\pi^{\star,\mathrm{blind}}\|_k.
  \]
  Since $\varrho_*<1$, rearranging gives
  \[
    \|\pi^{\mathrm{warm}}-\pi^{\star,\mathrm{blind}}\|_k
    \;\le\; \frac{\kappa}{1-\varrho_*}\,\varepsilon_{\mathrm{warm},k}^{\mathrm{mix}},
    \qquad
    \|\pi^{\mathrm{proj}}-\pi^{\star,\mathrm{blind}}\|_k
    \;\le\; \frac{\varrho_*\kappa}{1-\varrho_*}\,\varepsilon_{\mathrm{warm},k}^{\mathrm{mix}}.
  \]
  Summing over slabs using $\|f\|_{L^2(\mu)}^2=\sum_k\|f\|_k^2$ yields
  \[
    \|\pi^{\mathrm{proj}}-\pi^{\star,\mathrm{blind}}\|_{L^2(\mu)}
    \;\le\; \frac{\varrho_*\kappa}{1-\varrho_*}\,\varepsilon_{\mathrm{warm}}^{\mathrm{mix}}.
  \]

  \item \textbf{Conclusion.}
  Finally, apply the triangle inequality:
  \[
    \big\|\widehat\pi^{\mathrm{agg,mix}}-\pi^{\star,\mathrm{blind}}\big\|_{L^2(\mu)}
    \;\le\;
    \big\|\widehat\pi^{\mathrm{agg,mix}}-\pi^{\mathrm{proj}}\big\|_{L^2(\mu)}
    +
    \big\|\pi^{\mathrm{proj}}-\pi^{\star,\mathrm{blind}}\big\|_{L^2(\mu)},
  \]
  and combine the bounds from Steps 1 and 4 to obtain \eqref{eq:policy-gap-bound-residual}.
\end{enumerate}

A complete proof (including the precise stable-inversion conditions implied by the stage-2 stabilizers and an explicit admissible choice of
$\delta_{\mathrm{BPTT}}$) is provided in Appendix~\ref{app:policy-gap-proof}.
\end{proof}

\paragraph{Coupling and amortization.}
Figure~\ref{fig:sec3-pipeline} summarizes the overall two-stage PG--DPO design.
Stage~2 is used purely as a warm-started \emph{post-processing} map: it reuses the current Stage~1 (DPO) policy,
estimates the adjoint blocks under frozen $\theta\sim q$, and outputs a single deployable $\theta$-blind projected rule.
For numerical stability and amortized deployment, we additionally consider (i) a residual/control-variate form of the projection
and (ii) interactive distillation of the projected rule into the policy network.
These practical coupling mechanisms are described in Appendix~\ref{app:coupling}.

\section{Breaking the Dimensional Barrier under Drift Uncertainty}
\label{sec:hd-geometry}

This section instantiates the decision-time \emph{static} Gaussian drift-uncertainty benchmark from
Section~\ref{subsubsec:gaussian-drift-static} and uses its closed-form constant-portfolio $q$-reference
as an analytic target.
Nature draws a latent drift $\theta\sim q$ at $t=0$ and keeps it constant over $[0,T]$, while the investor
cannot observe $\theta$ and must deploy a single $\theta$-blind policy under an ex--ante CRRA objective.
Because the benchmark admits a transparent decision-time reference, we can measure accuracy directly via a
decision-time Euclidean RMSE to the analytic reference, rather than relying only on realized utility.

Our goal is to test whether the \emph{two-stage PG--DPO pipeline} remains stable as the number of assets grows.
We generate APT-style covariance structures and sweep dimensions $d\in\{5,10,50,100\}$ under both
\emph{aligned} uncertainty ($P=s\Sigma$) and a \emph{misaligned} geometry that rotates uncertainty away from market risk directions.
We report results for (i) Stage~1 (\textbf{DPO}) warm-up policies, (ii) the Stage~2 \textbf{Pontryagin projection} (i.e., the
post-processed output completing \textbf{PG--DPO}), and (when applicable) amortized variants via interactive distillation,
under Monte Carlo budgets that scale linearly with $d$.

\subsection{Benchmark market and evaluation protocol}
\label{subsec:exp-setup}

This subsection fixes the benchmark and evaluation protocol used in Section~\ref{sec:hd-geometry}.
Our goal is to provide controlled evidence that the proposed two-stage pipeline remains
\emph{computationally stable and accurate} as the number of assets $d$ grows under \emph{decision-time}
parameter uncertainty.
The aligned vs.\ misaligned uncertainty geometries serve as two representative stress-test regimes;
the main message is scalability under uncertainty rather than any specific choice of $P$.

\medskip
\noindent\textbf{$\theta$-blind deployability (and what uses $\theta$).}
Throughout Section~\ref{sec:hd-geometry}, \emph{all reported policies are deployable and $\theta$-blind}:
the control is a function of observable state only (here, the observable decision-time state is $(t,x)$ with $t=0$ and a fixed horizon $T$),
and \emph{never takes the realized latent premium $\theta$ as an input}.
The latent $\theta\sim q$ is sampled \emph{only inside the simulator} to generate trajectories and to form
Monte Carlo averages that approximate $q$-expectations (in Stage~1 for stochastic gradients, and in Stage~2 for aggregated adjoint/mixed-moment inputs).
Any $\theta$-indexed objects (when referenced elsewhere) are used only for \emph{offline diagnostics} and are
not part of the deployable decision rule.

\medskip
\noindent\textbf{Static decision-time uncertainty benchmark.}
We adopt the static Gaussian drift-uncertainty market of Section~\ref{subsubsec:gaussian-drift-static},
i.e., \eqref{eq:gaussian-model} with \eqref{eq:q-gaussian}. Equivalently, we simulate
\[
  \frac{dS_t}{S_t}
  = r\,\mathbf{1}\,dt + \theta\,dt + \Sigma^{1/2} dW_t,
  \qquad
  \theta \sim \mathcal{N}(m,P),
\]
where the latent premium $\theta$ is drawn once at time $0$ and kept fixed over $[0,T]$.
The deployable policy is $\theta$-blind and interacts with $q$ only through sampling $\theta$ inside the simulator.

\medskip
\noindent\textbf{APT-style factor construction of $(m,\Sigma)$.}
We construct the mean premium and covariance via a low-dimensional factor representation in the spirit of
arbitrage pricing theory (APT) \cite{ross1976apt}.
Let $W^f$ be a $k_\Sigma$-dimensional Brownian motion (factor shocks) and $W^\varepsilon$ a $d$-dimensional
Brownian motion (idiosyncratic shocks), independent of $W^f$.
We write excess returns as
\begin{equation}
  dR_t
  :=
  \frac{dS_t}{S_t}-r\mathbf{1}\,dt
  =
  \theta\,dt
  +
  B\,\Sigma_f^{1/2}\,dW_t^f
  +
  \mathrm{diag}(\sigma_\varepsilon)\,dW_t^\varepsilon,
  \label{eq:exp_return_factor}
\end{equation}
with $B\in\mathbb{R}^{d\times k_\Sigma}$, $\Sigma_f\succ 0$, and $\sigma_\varepsilon\in(0,\infty)^d$.
This implies
\begin{equation}
  \Sigma
  =
  B\,\Sigma_f\,B^\top + \mathrm{diag}(\sigma_\varepsilon^2)
  \;=\; FF^\top+\mathrm{diag}(\sigma_\varepsilon^2),
  \label{eq:exp_sigma_factor}
\end{equation}
where $F:=B\,\mathrm{chol}(\Sigma_f)$.
We generate the mean premium in an APT-like form by drawing a factor price vector
$\lambda_m\in\mathbb{R}^{k_\Sigma}$ and setting
\begin{equation}
  m := B\,\lambda_m.
  \label{eq:exp_m_apt}
\end{equation}

\medskip
\noindent\textbf{One-shot generation and fairness across methods.}
For each dimension $d$, we generate a single market instance $(B,\Sigma_f,\sigma_\varepsilon,\lambda_m)$ once (using a fixed random seed)
and \emph{hold it fixed across all algorithmic comparisons and MC-budget variants}.
Within a fixed $d$, we change only the uncertainty covariance $P$ (aligned vs.\ misaligned and the scale $s$ below).
This isolates algorithmic effects from instance-to-instance randomness and makes the scaling comparisons controlled.

\medskip
\noindent\textbf{Uncertainty regimes (aligned vs.\ misaligned).}
We consider two geometries for the drift-uncertainty covariance $P$, controlled by a scalar magnitude $s>0$.

\emph{Aligned:} uncertainty shares market risk directions,
\begin{equation}
  P = s\,\Sigma,
  \qquad s>0.
  \label{eq:exp_aligned_P}
\end{equation}

\emph{Misaligned:} uncertainty factors are rotated away from the market factor space,
\begin{equation}
  P
  =
  \widetilde B\,\widetilde\Sigma_f\,\widetilde B^\top
  +
  s\,\mathrm{diag}(\sigma_\varepsilon^2),
  \label{eq:exp_misaligned_P}
\end{equation}
where $\widetilde B$ is generated independently of $B$ (and can be orthogonalized against the span of $B$ to enforce
large principal angles). The factor term is rescaled to match the aligned factor magnitude under the same $s$; concretely,
writing $\Sigma_{\mathrm{fac}}:=B\Sigma_f B^\top$ and $\widetilde\Sigma_{\mathrm{fac}}:=\widetilde B\,\widetilde\Sigma_f\,\widetilde B^\top$,
we rescale $\widetilde\Sigma_{\mathrm{fac}}$ so that
$\mathrm{tr}(\widetilde\Sigma_{\mathrm{fac}})=\mathrm{tr}(s\,\Sigma_{\mathrm{fac}})$
while keeping the diagonal term fixed as $s\,\mathrm{diag}(\sigma_\varepsilon^2)$.
This geometry increases heterogeneity across $\theta\sim q$ and makes mixed-moment estimation and subsequent linear-algebra
steps more fragile, providing a stringent scalability test.

\medskip
\noindent\textbf{Experiment grid and simulation budgets.}
We vary the number of assets over $d\in\{5,10,50,100\}$ and sweep three uncertainty magnitudes
$s\in\{10^{-3},10^{-2},10^{-1}\}$ for both aligned and misaligned geometries.
To keep Monte Carlo noise comparable across dimensions, we use linear-in-$d$ trajectory budgets:
a \emph{base} regime with $M_{\mathrm{MC}}=100\cdot d$ rollouts and a \emph{high} regime with $M_{\mathrm{MC}}=400\cdot d$ rollouts.
Here $M_{\mathrm{MC}}$ denotes the number of \emph{independent simulated trajectories} used to approximate the relevant $q$-expectations
(in Stage~1/DPO for stochastic gradient estimates, and in Stage~2 for adjoint/mixed-moment estimators at queried states).
Each trajectory draws its own latent $\theta\sim q$ at time $0$ and its own Brownian path thereafter; we do not use a separate inner loop
of $\theta$-samples per trajectory.
All methods share the same discretization scheme (Euler) and the same market instances within each $d$; implementation details
(network architecture, optimizer settings, and exact sampling conventions for Stage~1 vs.\ Stage~2) are reported in the implementation appendix and code release.

\medskip
\noindent\textbf{Analytic reference and decision-time evaluation.}
In the static Gaussian benchmark, the analytic decision-time reference under constant portfolios is available in closed form.
We use this closed-form rule only as an external decision-time target for evaluation; training does not impose the constant-portfolio restriction, and all methods learn from simulated trajectories over $[0,T]$ under the same $\theta$-blind constraint.
For $\gamma>1$ we use the CRRA reference \eqref{eq:pi_q_const_gaussian_ref} at $(t,\tau)=(0,T)$, i.e.,
$\pi_{q,\gamma}^{\mathrm{const}}(0,T)$.

We evaluate each method at $t=0$ with a fixed horizon $T=1.5$ on a fixed grid $\{X_0^{(i)}\}_{i=1}^{N_{\mathrm{eval}}}$ and report Euclidean RMSE to the analytic reference:
\begin{equation}
  \mathrm{RMSE}\!\big(\pi_0^{\mathrm{meth}},\pi_{q,\gamma}^{\mathrm{const}}\big)
  :=
  \bigg(
    \frac{1}{N_{\mathrm{eval}}}\sum_{i=1}^{N_{\mathrm{eval}}}
    \big\|
      \pi_0^{\mathrm{meth}}\!\big(X_0^{(i)};T\big) - \pi_{q,\gamma}^{\mathrm{const}}(0,T)
    \big\|^2
  \bigg)^{1/2},
  \label{eq:exp_rmse}
\end{equation}
where $\|\cdot\|$ is the Euclidean norm on $\mathbb{R}^d$ and $\pi_0^{\mathrm{meth}}(\cdot;T)$ denotes the method's \emph{deployable} decision-time action at time $0$
for the fixed horizon $T$ (a $\theta$-blind output rule).
With the benchmark fixed and with $(m,\Sigma,P)$ constructed as in
\eqref{eq:exp_m_apt}--\eqref{eq:exp_misaligned_P},
the remaining subsections compare (i) the Stage~1 (\textbf{DPO}) warm-up policy, (ii) the Stage~2 \textbf{Pontryagin projection} output (completing \textbf{PG--DPO}),
and (iii) interactive distillation variants, under matched simulation budgets.

\subsection{High-dimensional CRRA benchmark: projection and amortization}
\label{subsec:crra_high_dim_results}

\noindent\textbf{Mixed-moment estimation and a decoupling approximation.}
A practical issue throughout our experiments (both aligned and misaligned) is the estimation of \emph{mixed moments} across the latent parameter, such as
$\mathbb{E}_{\theta\sim q}[p_t^\theta(z)\,\theta]$ (and analogous products entering $\widehat g_t^{\mathrm{mix}}$), because the costate $p_t^\theta(z)$ is
$\theta$-dependent and high-dimensional.
Finite-sample covariance between $p_t^\theta$ and $\theta$ can create large Monte Carlo variance, which is then amplified by the subsequent linear solve in the
projection step.
For numerical stability and a uniform protocol across geometries, we therefore use a simple \emph{decoupling} (independence) approximation for these mixed moments,
\[
  \mathbb{E}_{\theta\sim q}[p_t^\theta(z)\,\theta]\;\approx\;\mathbb{E}_{\theta\sim q}[p_t^\theta(z)]\,\mathbb{E}_{\theta\sim q}[\theta],
\]
(and similarly for other mixed products).
This approximation is exact when the relevant Pontryagin objects are effectively $\theta$-invariant and is accurate whenever
$\mathrm{Cov}_{q}(p_t^\theta(z),\theta)$ is small relative to marginal scales.
In this CRRA benchmark it does not change the qualitative scaling conclusions: projection remains stable in aligned regimes, and the misaligned degradation is
consistent with residual growth and curvature mismatch rather than catastrophic mixed-moment blow-ups.\footnote{In extreme uncertainty/misalignment---where
$\theta$--costate dependence becomes pronounced---the decoupling can break down. In that case one should revert to full mixed-moment estimation (with larger
budgets and/or regularized/certified projection).}

\begin{figure*}[t!]
  \centering
  \includegraphics[width=\linewidth]{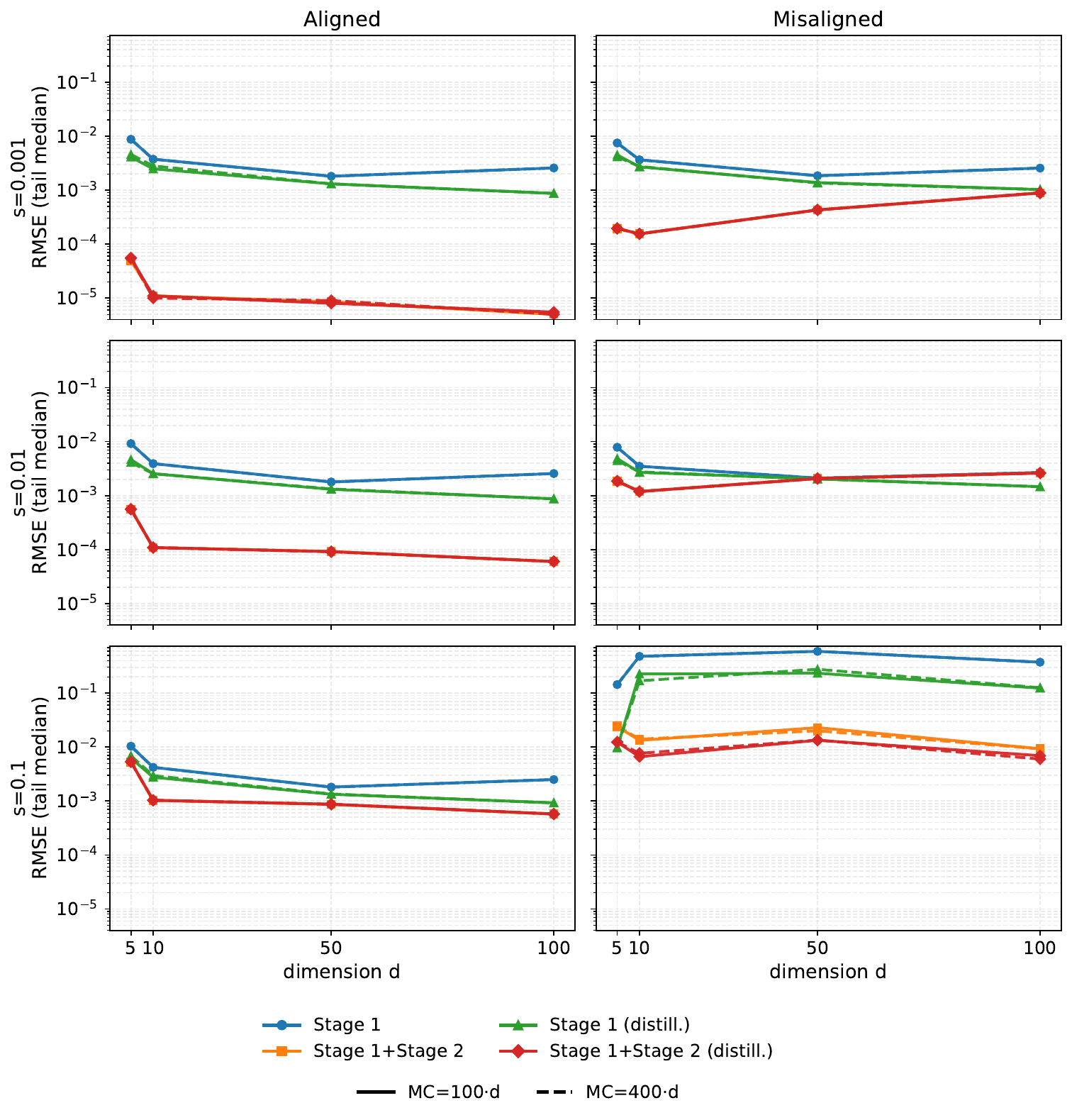}
  \caption{Decision-time Euclidean RMSE at $t=0$ versus dimension $d$ (log scale), summarized by a \emph{tail median} over the late-training window.
  Concretely, for each condition we take the median RMSE over the last evaluation snapshots recorded during training.
  Rows: uncertainty magnitude $s\in\{10^{-3},10^{-2},10^{-1}\}$. Columns: aligned vs.\ misaligned geometry.
  Curves compare the Stage~1 warm-up policy (DPO) to the Stage~2 output (Pontryagin projection; completing PG--DPO), with and without interactive distillation (amortization).
  Solid vs.\ dashed lines correspond to MC base ($100\cdot d$) vs.\ high ($400\cdot d$) trajectory budgets.}
  \label{fig:crra_scaling_rmse}
\end{figure*}

\paragraph{Protocol and summary statistic.}
We consider the CRRA benchmark with $\gamma=2$ under Gaussian drift uncertainty $q$ and evaluate against the analytic constant $q$-reference
\eqref{eq:pi_q_const_gaussian_ref}.
We track (i) the Monte Carlo objective estimate $\widehat{J}$ during training and (ii) the decision-time error at $t=0$ via RMSE \eqref{eq:exp_rmse} with
fixed horizon $T=1.5$.
Because stochastic optimization produces non-monotone and noisy RMSE curves, we summarize each condition by a robust \emph{tail median}:
the median RMSE over the final evaluation snapshots in the late-training window.
Unless stated otherwise, the Stage~2 projection / on-the-fly teacher uses the mixed-moment ($p$-weighted) aggregation of
Section~\ref{subsec:ppgdpo-uncertainty} (implemented with the decoupling approximation above).

\medskip
\noindent\textbf{What is compared in Figure~\ref{fig:crra_scaling_rmse}.}
Stage~1 (\textbf{DPO}; Section~\ref{subsec:pgdpo-bptt-pmp}) trains a deployable $\theta$-blind policy $\pi_\varphi$ by maximizing $\widehat{J}$ via pathwise
gradients through the Euler simulator.
Stage~2 (\textbf{Pontryagin projection}; Section~\ref{subsec:ppgdpo-uncertainty}) post-processes a Stage~1 checkpoint by estimating (aggregated) adjoint inputs
and applying the statewise linear solve; we use the residual form of Section~\ref{subsubsec:cv-projection}.
The resulting projected control is the output of the full \textbf{PG--DPO} pipeline.
Interactive distillation (Section~\ref{subsubsec:interactive-distillation}) treats the Stage~2 projected control as a teacher signal and amortizes it
back into a deployable policy network.

Thus Figure~\ref{fig:crra_scaling_rmse} separates \emph{projection quality} (Stage~2 / PG--DPO: post-hoc projected, still $\theta$-blind)
from \emph{amortized deployable quality} (distilled Stage~1: single forward pass).

\begin{figure*}[t]
  \centering
  \includegraphics[width=\linewidth]{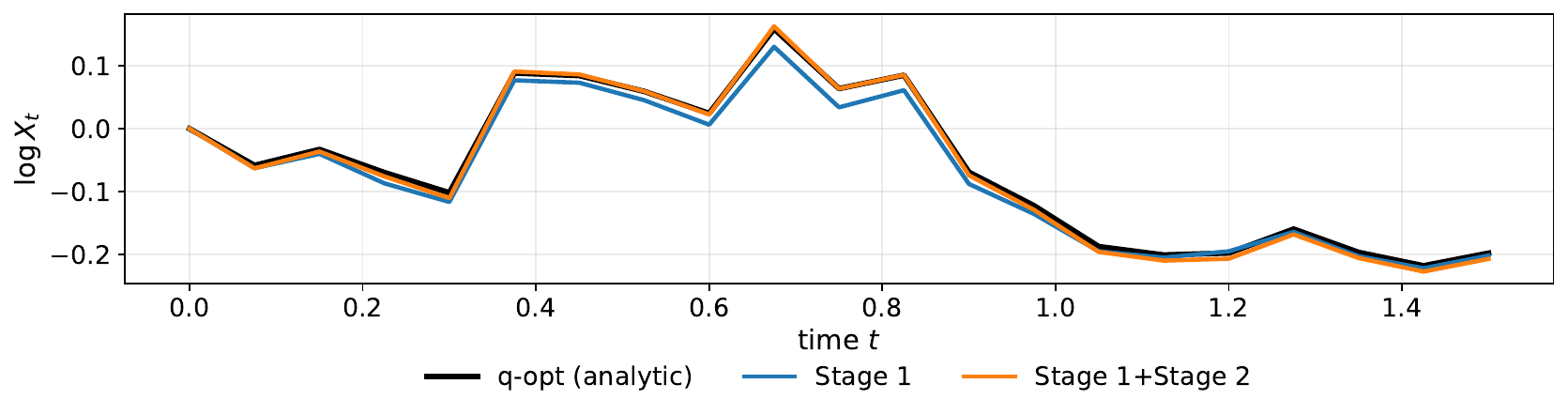}\\[2mm]
  \includegraphics[width=\linewidth]{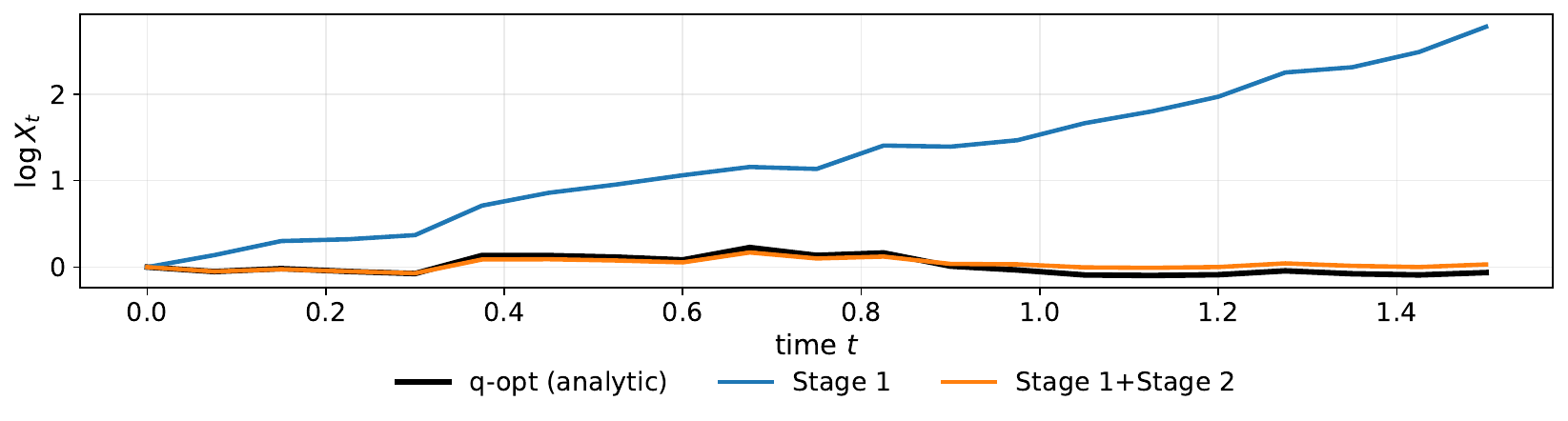}
  \caption{Pathwise sanity check at $d=100$ under common random numbers. Black: analytic Gaussian constant-fraction benchmark evaluated with the remaining horizon
  $\tau=T-t$ (fixed $q$). Blue: Stage~1 (DPO) warm-up policy. Orange: on-the-fly Stage~2 teacher obtained by Pontryagin projection at visited states (PG--DPO).
  Top/bottom: aligned/misaligned.}
  \label{fig:crra_wealth_paths_d100}
\end{figure*}

\medskip
\noindent\textbf{Stage~2 projection versus amortization: scaling with dimension.}
For the \textbf{aligned geometry}, small and moderate uncertainty ($s=10^{-3},10^{-2}$) yield a sharp Stage~2 reduction in decision-time error across all tested dimensions,
bringing RMSE down to the $10^{-5}$--$10^{-4}$ range, while Stage~1 (DPO) policies remain around $10^{-3}$.
Interactive distillation consistently improves the deployable policy (distilled Stage~1 below basic Stage~1) while leaving the Stage~2 projection essentially unchanged,
confirming the intended division of labor: Stage~2 supplies a structured stationarity-correction signal, and distillation injects that signal into $\pi_\varphi$ to
reduce the remaining approximation/optimization gap.

In contrast, under the \textbf{misaligned geometry} the picture becomes more heterogeneous.
For small to moderate uncertainty ($s=10^{-3},10^{-2}$), Stage~2 still improves decision-time RMSE at small $d$,
but its advantage shrinks with dimension and can approach the $10^{-3}$ level by $d=100$.
For the largest uncertainty scale ($s=10^{-1}$), Stage~1 (DPO) becomes markedly less reliable, whereas Stage~2 remains substantially better, indicating that
projection can act as a stabilizing correction even when end-to-end learning is stressed.
Across settings, the base and high MC budgets tend to yield similar tail-median RMSE curves, suggesting that linear-in-$d$ scaling of simulation budgets
is sufficient for stable comparisons in this benchmark.

\medskip
\noindent\textbf{Mechanism: why misalignment can reduce projection gains.}
To explain when and why the projection gains shrink, we analyze Stage~2 diagnostic statistics reported in Appendix~\ref{app:stage2-diagnostics}; see
Figures~\ref{fig:app_stage2_residual_q50}--\ref{fig:app_stage2_bad_sign_frac}.
The diagnostics indicate that the degradation under misalignment is driven primarily by increased stationarity residuals and curvature mismatch, rather
than by catastrophic denominator sign failures:
(i) the Stage~2 residual norm grows with dimension and becomes especially large in the hardest misaligned regime,
(ii) the projection denominator magnitude stays away from zero at typical quantiles, and
(iii) the bad-sign fraction remains negligible, while
(iv) the effective curvature statistic $\kappa_{\mathrm{curv}}$ stays near the nominal $1/\gamma$ reference in easy regimes but can deviate substantially in the hardest
misaligned/high-uncertainty setting.
These patterns are consistent with a geometric explanation: when $P$ and $\Sigma$ do not commute, the inverse operations implicit in the projection mix directions
of uncertainty that are misaligned with risk directions, thereby amplifying estimation noise in the projected correction as $d$ grows.

\medskip
\noindent\textbf{Pathwise sanity check.}
Figure~\ref{fig:crra_wealth_paths_d100} complements the decision-time RMSE with a trajectory-level view under common random numbers.
In the aligned case, the on-the-fly Stage~2 teacher tracks the analytic $q$-reference closely along a realized path and reduces the deviation
$\Delta\log X_t$ relative to the warm Stage~1 (DPO) policy.
In the misaligned case, the teacher can deviate more noticeably under the same common-noise protocol, mirroring the reduced projection advantage in the
hardest regimes of Figure~\ref{fig:crra_scaling_rmse} and motivating amortization/reliability mechanisms in interactive distillation.

\begin{table*}[p]
\centering
\begingroup
\sisetup{detect-weight=true,detect-family=true}
\providecommand{\best}[1]{{\bfseries\num{#1}}}

\small
\setlength{\tabcolsep}{6pt}
\renewcommand{\arraystretch}{1.10}

\textbf{Aligned}\par\vspace{2pt}
\begin{adjustbox}{max width=\textwidth}
\begin{tabular}{cl|cccc}
\toprule
$s$ & Method & $d=5$ & $10$ & $50$ & $100$ \\
\midrule
$10^{-3}$ & DPO (basic)
& \num{8.76e-03} & \num{3.74e-03} & \num{1.81e-03} & \num{2.58e-03} \\
$10^{-3}$ & PG--DPO (basic)
& \best{4.95e-05} & \best{1.10e-05} & \num{8.50e-06} & \best{5.00e-06} \\
$10^{-3}$ & DPO (distill.)
& \num{4.10e-03} & \num{2.48e-03} & \num{1.31e-03} & \num{1.70e-03} \\
$10^{-3}$ & PG--DPO (distill.)
& \num{5.55e-05} & \best{1.10e-05} & \best{8.00e-06} & \num{5.50e-06} \\
$10^{-3}$ & PPO (baseline)
& \num{2.76e-01} & \num{1.14e-01} & \num{1.39e-01} & \num{1.67e-01} \\
\midrule
$10^{-2}$ & DPO (basic)
& \num{9.19e-03} & \num{3.91e-03} & \num{1.79e-03} & \num{2.57e-03} \\
$10^{-2}$ & PG--DPO (basic)
& \best{5.57e-04} & \best{1.09e-04} & \num{9.20e-05} & \best{6.00e-05} \\
$10^{-2}$ & DPO (distill.)
& \num{4.64e-03} & \num{2.55e-03} & \num{1.31e-03} & \num{1.80e-03} \\
$10^{-2}$ & PG--DPO (distill.)
& \num{5.60e-04} & \best{1.09e-04} & \best{9.15e-05} & \best{6.00e-05} \\
$10^{-2}$ & PPO (baseline)
& \num{3.00e-01} & \num{7.88e-02} & \num{1.50e-01} & \num{1.58e-01} \\
\midrule
$10^{-1}$ & DPO (basic)
& \num{1.04e-02} & \num{4.21e-03} & \num{1.81e-03} & \num{2.50e-03} \\
$10^{-1}$ & PG--DPO (basic)
& \best{5.29e-03} & \best{1.03e-03} & \best{8.68e-04} & \best{5.76e-04} \\
$10^{-1}$ & DPO (distill.)
& \num{6.23e-03} & \num{2.76e-03} & \num{1.33e-03} & \num{1.53e-03} \\
$10^{-1}$ & PG--DPO (distill.)
& \num{5.33e-03} & \best{1.03e-03} & \num{8.73e-04} & \best{5.76e-04} \\
$10^{-1}$ & PPO (baseline)
& \num{2.70e-01} & \num{9.65e-02} & \num{1.55e-01} & \num{1.78e-01} \\
\bottomrule
\end{tabular}
\end{adjustbox}

\vspace{10pt}

\textbf{Misaligned}\par\vspace{2pt}
\begin{adjustbox}{max width=\textwidth}
\begin{tabular}{cl|cccc}
\toprule
$s$ & Method & $d=5$ & $10$ & $50$ & $100$ \\
\midrule
$10^{-3}$ & DPO (basic)
& \num{7.49e-03} & \num{3.65e-03} & \num{1.85e-03} & \num{2.56e-03} \\
$10^{-3}$ & PG--DPO (basic)
& \best{1.94e-04} & \best{1.55e-04} & \best{4.30e-04} & \best{8.89e-04} \\
$10^{-3}$ & DPO (distill.)
& \num{4.49e-03} & \num{2.70e-03} & \num{1.38e-03} & \num{1.02e-03} \\
$10^{-3}$ & PG--DPO (distill.)
& \num{1.95e-04} & \best{1.55e-04} & \best{4.30e-04} & \best{8.89e-04} \\
$10^{-3}$ & PPO (baseline)
& \num{2.87e-01} & \num{8.25e-02} & \num{1.51e-01} & \num{1.69e-01} \\
\midrule
$10^{-2}$ & DPO (basic)
& \num{7.87e-03} & \num{3.50e-03} & \num{2.11e-03} & \num{2.68e-03} \\
$10^{-2}$ & PG--DPO (basic)
& \best{1.85e-03} & \best{1.19e-03} & \num{2.07e-03} & \num{2.62e-03} \\
$10^{-2}$ & DPO (distill.)
& \num{4.86e-03} & \num{2.69e-03} & \best{2.04e-03} & \best{1.46e-03} \\
$10^{-2}$ & PG--DPO (distill.)
& \best{1.85e-03} & \best{1.19e-03} & \num{2.07e-03} & \num{2.62e-03} \\
$10^{-2}$ & PPO (baseline)
& \num{2.58e-01} & \num{8.99e-02} & \num{1.59e-01} & \num{1.38e-01} \\
\midrule
$10^{-1}$ & DPO (basic)
& \num{1.44e-01} & \num{4.81e-01} & \num{5.94e-01} & \num{3.74e-01} \\
$10^{-1}$ & PG--DPO (basic)
& \num{2.40e-02} & \num{1.33e-02} & \num{2.29e-02} & \num{9.31e-03} \\
$10^{-1}$ & DPO (distill.)
& \best{9.67e-03} & \num{2.26e-01} & \num{2.34e-01} & \num{1.24e-01} \\
$10^{-1}$ & PG--DPO (distill.)
& \num{1.23e-02} & \best{6.63e-03} & \best{1.34e-02} & \best{6.93e-03} \\
$10^{-1}$ & PPO (baseline)
& \num{2.78e-01} & \num{8.84e-02} & \num{1.47e-01} & \num{1.77e-01} \\
\bottomrule
\end{tabular}
\end{adjustbox}

\caption{Decision-time Euclidean RMSE at $t=0$ (tail median over the late-training window; last six evaluation snapshots),
computed using the common Euclidean-norm RMSE definition \eqref{eq:exp_rmse}.
DPO rows report the deployable policy output (Stage~1).
PG--DPO rows report the post-hoc Pontryagin projection output (Stage~2; residual form).
``Distill.'' rows correspond to the amortized deployable policy trained via interactive distillation.
PPO is a model-free baseline trained under the same benchmark setting (same simulator, horizon $T=1.5$, and $\theta$-blind restriction).
\textbf{Bold} entries indicate the best (lowest) RMSE among the five methods for each $(s,d)$ and geometry.}
\label{tab:crra_stage12_distill_ppo_landscape}

\endgroup
\end{table*}

\subsection{A strong RL baseline: PPO, and why it falls short in our benchmark}
\label{subsec:ppo_baseline}

\noindent\textbf{Why include PPO, and how we match the setting.}
Proximal Policy Optimization (PPO) is a widely used and robust model-free policy-gradient baseline for continuous control \citep{schulman2017proximal}.
We include PPO to answer a concrete question: can a generic, well-tuned model-free RL method recover the decision-time $q$-optimal $\theta$-blind allocation
in our high-dimensional drift-uncertainty benchmark under comparable simulation budgets?
This comparison is especially informative in our static Gaussian benchmark because the target decision-time rule is structurally simple (constant and available
in closed form), so performance gaps primarily reflect optimization difficulty and credit assignment rather than policy-class expressiveness.

Since classical HJB solvers and value-function-based deep PDE surrogates are not practical baselines in the high-dimensional uncertain regime targeted here
(Section~\ref{app:limits-dp}), PPO serves as a strong \emph{simulation-only} comparator that operates on the same sampled trajectories without exploiting
value-function PDE structure.
For a fair comparison, PPO is trained on the same Euler simulator and time discretization as our DPO/PG--DPO setup, under the same deployability restriction
(the policy never observes the latent $\theta$), and under the same terminal-utility objective.
We match simulation budgets at the level of total environment interactions (simulated trajectories $\times$ time steps); implementation and PPO tuning details
are deferred to the appendix and code release.

\medskip
\noindent\textbf{Empirical outcome.}
Table~\ref{tab:crra_stage12_distill_ppo_landscape} shows that PPO remains far from the analytic decision-time $q$-reference across essentially all
conditions, with RMSE typically on the order of $10^{-1}$.
In contrast, the Pontryagin-based pipeline attains substantially smaller errors:
in aligned regimes the Stage~2 projection (PG--DPO) reaches the $10^{-5}$--$10^{-4}$ range for small and moderate uncertainty, while in misaligned regimes the
projection advantage narrows but remains systematic.
Distillation improves the \emph{deployable} Stage~1 policy relative to basic Stage~1 (DPO), but does not eliminate the remaining gap to the post-hoc projection,
consistent with the amortization interpretation in Section~\ref{subsec:crra_high_dim_results}.

\medskip
\noindent\textbf{Why PPO underperforms in this benchmark.}
The gap is not evidence that PPO is intrinsically weak; rather, it reflects that our benchmark stresses regimes where a generic likelihood-ratio policy gradient is
statistically disadvantaged compared to pathwise/adjoint-based updates.

One contributor is \emph{terminal-reward credit assignment}: with terminal utility as the only reward signal, most intermediate transitions carry no direct learning signal,
so advantage estimates become noisy and the gradient signal weakens as the horizon and dimension grow.
A second contributor is \emph{high-dimensional action noise and variance scaling}: in continuous control, exploration is typically implemented through stochastic policies
(e.g., diagonal Gaussians), and as $d$ increases the effective exploration noise and likelihood-ratio gradient variance can inflate rapidly, making it difficult to
consistently recover the precise constant-fraction structure of the $q$-optimal $\theta$-blind rule under fixed simulation budgets.

In contrast, Stage~1 (DPO) exploits backpropagation through the differentiable simulator (pathwise gradients), and Stage~2 leverages the affine-in-control Pontryagin
structure through a $q$-aggregated projection, replacing a noisy high-dimensional policy-gradient update by a structured stationarity correction tailored to the
$\theta$-blind ex--ante objective.


\section{Recovering Intertemporal Hedging Demand in Factor-Driven Markets}
\label{sec:hedging-recovery}

Section~\ref{sec:hd-geometry} stressed \emph{scaling} under static drift uncertainty, where the target
$q$-reference is time-homogeneous and largely myopic. Here we shift the focus to an \emph{economic}
target: recovering the \emph{intertemporal hedging demand} induced by factor-driven investment
opportunities when return shocks are correlated with factor shocks \citep{campbell2002strategic,xia2001learning}.

We adopt the mean-reverting Gaussian premium benchmark of
Section~\ref{subsubsec:gaussian-drift-ou-belief}. At decision time ($t=0$), the investor is assumed to
have access only to an interval/posterior uncertainty description $q_0(dy)$ for the initial premium
state $Y_0$ (e.g., $Y_0\sim\mathcal{N}(m_0,P_0)$), and we restrict attention to \emph{deployable $Y$-blind}
policies that depend only on observable wealth and time-to-go.
When return and factor shocks are correlated ($\rho\neq 0$), the analytic OU decision-time reference
\eqref{eq:pi-q-const-ou} contains the cross term $M_{\mathrm{cross}}(\tau)$ (with $\tau:=T-t$) that
encodes intertemporal hedging demand.

We evaluate: (i) Stage~1 \textbf{DPO} (a $Y$-blind warm-up trained by pathwise/BPTT gradients),
(ii) Stage~2 \textbf{Pontryagin projection} completing \textbf{PG--DPO} (a post-hoc $q_0$-aggregated stationarity correction),
and (iii) interactive distillation (amortization), against a model-free PPO baseline---all under the same $Y$-blind deployability restriction.
Performance is measured by decision-time RMSE at $t=0$ relative to the analytic constant-portfolio OU
reference \eqref{eq:pi-q-const-ou}.

\subsection{Experimental setting}
\label{subsec:hedging_setting}

\medskip
\noindent\textbf{Decision-time interval/posterior estimate for $Y_0$ and deployability.}
Throughout this section, all reported policies are \emph{deployable and $Y$-blind}: the control is a
function of observable wealth and time-to-go only, and never takes as input (i) a point estimate of
$Y_0$, (ii) the realized initial premium $Y_0$, or (iii) the factor path $(Y_t)$ (including the PPO
baseline). At deployment, the investor is assumed to have access only to an uncertainty description
$q_0(dy)$ for $Y_0$ (interval/posterior output of an external estimation pipeline), and commits to a
single policy optimized \emph{ex ante} under $Y_0\sim q_0$ (without filtering/updating over $[0,T]$).

The latent premium factor is sampled and propagated \emph{only inside the simulator} to generate
trajectories and to form Monte Carlo averages used by the Stage~2 projection (and by the teacher in
distillation). Any $Y$-indexed quantities are used only for offline evaluation and diagnostics.

\medskip
\noindent\textbf{OU premium market with a hedging channel.}
We adopt the OU premium benchmark of Section~\ref{subsubsec:gaussian-drift-ou-belief}. To avoid a
notational clash with other uses of $\kappa$ elsewhere in the paper, we denote the OU mean-reversion
rate by $\kappa_Y$ in this section. Let $Y_t\in\mathbb{R}^{m}$ be a mean-reverting premium factor and
$R_t\in\mathbb{R}^{d}$ the risky excess returns:
\begin{align*}
  dY_t &= \kappa_Y(\bar y-Y_t)\,dt + \Xi\,dW_t^Y,
  \qquad Y_0 \sim \mathcal{N}(m_0,P_0),\\
  dR_t &:= \frac{dS_t}{S_t}-r\mathbf{1}\,dt
  = B Y_t\,dt + \Sigma^{1/2}\,dW_t,\\
  d\langle W, W^Y\rangle_t &= \rho\,dt.
\end{align*}
Here $\rho\in\mathbb{R}^{d\times m}$ denotes the instantaneous return--factor shock correlation matrix,
satisfying the feasibility condition $\rho^\top\rho \preceq I_m$. A nonzero $\rho$ induces
intertemporal hedging demand and enters the CRRA decision-time reference through the cross-covariance
term $M_{\mathrm{cross}}(\tau)$ in \eqref{eq:pi-q-const-ou}, where $\tau:=T-t$ is time-to-go (so at
decision time $t=0$, $\tau=T$). When $\rho=0$ (independent return and factor shocks), the hedging
channel vanishes ($M_{\mathrm{cross}}(\tau)=0$) and the reference reduces to the independence-case
benchmark.

\medskip
\noindent\textbf{Decision-time uncertainty geometry for $Y_0\sim \mathcal{N}(m_0,P_0)$.}
We control the magnitude of decision-time estimation uncertainty by a scalar $s_0>0$ and construct
$P_0$ from an identification-motivated baseline
\[
  \widetilde P_0 := (B^\top \Sigma^{-1} B)^{-1}\in\mathbb{R}^{m\times m}.
\]
We consider two geometries. In the \emph{aligned} case, we keep the principal directions of
$\widetilde P_0$ and rescale it so that the average marginal variance equals $s_0$:
\begin{equation}
  P_0^{\mathrm{aligned}}(s_0)
  := \frac{s_0\,m}{\mathrm{tr}(\widetilde P_0)}\,\widetilde P_0,
  \label{eq:P0_aligned}
\end{equation}
so that $\mathrm{tr}(P_0^{\mathrm{aligned}})/m=s_0$. In the \emph{misaligned} case, we preserve the
eigenvalue spectrum of $\widetilde P_0$ but randomize its eigenvectors via an orthogonal rotation:
letting $\widetilde P_0 = U\mathrm{diag}(\lambda)U^\top$ be an eigen-decomposition and drawing an
orthogonal matrix $R$ (e.g., Haar), we define
\begin{equation}
  P_0^{\mathrm{misaligned}}(s_0)
  := \frac{s_0\,m}{\mathrm{tr}(\widetilde P_0)}\,
     U R\,\mathrm{diag}(\lambda)\,R^\top U^\top,
  \label{eq:P0_misaligned}
\end{equation}
which matches the same trace normalization while rotating the uncertainty directions away from those
of $\widetilde P_0$. We sweep $s_0\in\{10^{-3},10^{-2},10^{-1}\}$ under both aligned and misaligned $P_0$.

\medskip
\noindent\textbf{Two-stage solver, amortization, and evaluation protocol.}
We use the two-stage pipeline of Section~\ref{sec:pgdpo-uncertainty}. Stage~1 (\textbf{DPO}) trains a deployable
$Y$-blind policy by stochastic gradient ascent using pathwise/BPTT gradients
(Section~\ref{subsec:pgdpo-bptt-pmp}). Stage~2 (\textbf{Pontryagin projection}) applies the $q_0$-aggregated stationarity
correction computed under a warm-up policy (Section~\ref{subsec:ppgdpo-uncertainty}), implemented in the residual/control-variate
form (Section~\ref{subsubsec:cv-projection}); the resulting projected rule is the output of \textbf{PG--DPO}.
Interactive distillation amortizes the projected teacher into a fast deployable policy network
(Section~\ref{subsubsec:interactive-distillation}). As a model-free baseline, we also train a PPO policy under the same
$Y$-blind observation restriction and report its decision-time full RMSE in Table~\ref{tab:ou_rmse_s0_sweep_landscape}.

We sweep $d\in\{5,10,50,100\}$ (one fixed market instance per $d$, following the same market-generation
protocol used in Section~\ref{sec:hd-geometry}), train for $5000$ epochs, and evaluate every $100$ epochs.
We evaluate the decision-time action at $t=0$ (i.e., $\tau=T$) and report RMSE to the analytic constant-portfolio
OU reference \eqref{eq:pi-q-const-ou}.

Unless stated otherwise we set $\gamma=2$, $r=0.03$, $\kappa_Y=1.0$, $\xi_{\mathrm{scale}}=0.25$, and a
correlation \emph{magnitude} $\rho_0=0.5$. We instantiate the model parameters as follows:
(i) we take $\Xi=\xi_{\mathrm{scale}} I_m$,
(ii) we implement the return--factor correlation by choosing a matrix $\rho=\rho_0 Q$ with
$Q\in\mathbb{R}^{d\times m}$ having orthonormal columns ($Q^\top Q=I_m$), ensuring $\rho^\top\rho=\rho_0^2 I_m\preceq I_m$.
(A simple choice is $Q=[I_m;0]\in\mathbb{R}^{d\times m}$ when $m\le d$; any fixed orthonormal-column $Q$ yields an equivalent correlation magnitude.)
We fix the horizon to $T=1.5$ in this section.

In addition to the full allocation error, we use the myopic+hedging decomposition induced by the OU
structure, and report component-wise diagnostics for the projected (Stage~2 / PG--DPO) rules.
To reduce noise from stochastic optimization, for each condition we summarize each metric by a \emph{tail median} over
the last six evaluation checkpoints.

\subsection{Results: hedging-demand recovery, amortization, and robustness to decision-time uncertainty}
\label{subsec:hedging_results}

\begin{table*}[p]
\centering
\begingroup
\sisetup{print-zero-exponent=true,detect-weight=true,detect-family=true}
\providecommand{\best}[1]{{\bfseries\num{#1}}}

\small
\setlength{\tabcolsep}{5pt}
\renewcommand{\arraystretch}{1.08}

\textbf{Aligned $P_0$}\par\vspace{2pt}
\begin{adjustbox}{max width=\textwidth}
\begin{tabular}{cl|cccc}
\toprule
$s_0$ & Method & $d=5$ & $10$ & $50$ & $100$ \\
\midrule
$10^{-3}$ & DPO (basic)
& \num{6.31e-03} & \num{5.19e-03} & \num{4.01e-03} & \num{3.54e-03} \\
$10^{-3}$ & PG--DPO (basic)
& \num{4.71e-05} & \best{5.10e-05} & \num{1.39e-04} & \num{1.56e-04} \\
$10^{-3}$ & DPO (distill.)
& \num{2.46e-03} & \num{3.57e-03} & \num{3.58e-03} & \num{3.22e-03} \\
$10^{-3}$ & PG--DPO (distill.)
& \best{4.43e-05} & \num{5.39e-05} & \best{1.37e-04} & \best{1.42e-04} \\
$10^{-3}$ & PPO (baseline)
& \num{7.78e-02} & \num{1.03e-01} & \num{4.20e+00} & \num{2.41e+00} \\
\midrule
$10^{-2}$ & DPO (basic)
& \num{5.61e-03} & \num{4.83e-03} & \num{3.83e-03} & \num{3.50e-03} \\
$10^{-2}$ & PG--DPO (basic)
& \num{5.03e-05} & \best{4.54e-05} & \best{1.45e-04} & \num{1.75e-04} \\
$10^{-2}$ & DPO (distill.)
& \num{3.08e-03} & \num{3.16e-03} & \num{3.53e-03} & \num{3.21e-03} \\
$10^{-2}$ & PG--DPO (distill.)
& \best{4.48e-05} & \num{5.11e-05} & \num{1.47e-04} & \best{1.56e-04} \\
$10^{-2}$ & PPO (baseline)
& \num{7.30e-02} & \num{9.15e+00} & \num{4.27e+00} & \num{2.43e+00} \\
\midrule
$10^{-1}$ & DPO (basic)
& \num{6.93e-03} & \num{4.53e-03} & \num{3.96e-03} & \num{3.47e-03} \\
$10^{-1}$ & PG--DPO (basic)
& \num{5.50e-05} & \best{4.16e-05} & \num{2.63e-04} & \num{2.97e-04} \\
$10^{-1}$ & DPO (distill.)
& \num{2.85e-03} & \num{3.70e-03} & \num{3.65e-03} & \num{3.28e-03} \\
$10^{-1}$ & PG--DPO (distill.)
& \best{4.77e-05} & \num{5.50e-05} & \best{2.60e-04} & \best{2.88e-04} \\
$10^{-1}$ & PPO (baseline)
& \num{7.34e-02} & \num{1.00e-01} & \num{4.22e+00} & \num{2.64e+00} \\
\bottomrule
\end{tabular}
\end{adjustbox}

\vspace{10pt}

\textbf{Misaligned $P_0$}\par\vspace{2pt}
\begin{adjustbox}{max width=\textwidth}
\begin{tabular}{cl|cccc}
\toprule
$s_0$ & Method & $d=5$ & $10$ & $50$ & $100$ \\
\midrule
$10^{-3}$ & DPO (basic)
& \num{6.11e-03} & \num{5.40e-03} & \num{3.87e-03} & \num{3.62e-03} \\
$10^{-3}$ & PG--DPO (basic)
& \num{5.12e-05} & \best{5.11e-05} & \best{1.35e-04} & \num{1.57e-04} \\
$10^{-3}$ & DPO (distill.)
& \num{3.51e-03} & \num{2.93e-03} & \num{3.56e-03} & \num{3.21e-03} \\
$10^{-3}$ & PG--DPO (distill.)
& \best{4.36e-05} & \num{5.48e-05} & \num{1.41e-04} & \best{1.44e-04} \\
$10^{-3}$ & PPO (baseline)
& \num{8.96e-02} & \num{9.28e-02} & \num{4.37e+00} & \num{2.46e+00} \\
\midrule
$10^{-2}$ & DPO (basic)
& \num{6.17e-03} & \num{4.89e-03} & \num{3.89e-03} & \num{3.57e-03} \\
$10^{-2}$ & PG--DPO (basic)
& \num{4.99e-05} & \best{5.96e-05} & \best{1.54e-04} & \num{1.75e-04} \\
$10^{-2}$ & DPO (distill.)
& \num{2.84e-03} & \num{3.97e-03} & \num{3.68e-03} & \num{3.21e-03} \\
$10^{-2}$ & PG--DPO (distill.)
& \best{4.63e-05} & \num{6.31e-05} & \best{1.54e-04} & \best{1.57e-04} \\
$10^{-2}$ & PPO (baseline)
& \num{8.02e-02} & \num{7.75e+00} & \num{4.47e+00} & \num{2.53e+00} \\
\midrule
$10^{-1}$ & DPO (basic)
& \num{6.09e-03} & \num{4.17e-03} & \num{3.97e-03} & \num{3.56e-03} \\
$10^{-1}$ & PG--DPO (basic)
& \num{5.63e-05} & \best{2.44e-04} & \num{3.24e-04} & \num{3.29e-04} \\
$10^{-1}$ & DPO (distill.)
& \num{3.22e-03} & \num{3.08e-03} & \num{3.65e-03} & \num{3.18e-03} \\
$10^{-1}$ & PG--DPO (distill.)
& \best{5.31e-05} & \num{2.46e-04} & \best{3.22e-04} & \best{3.21e-04} \\
$10^{-1}$ & PPO (baseline)
& \num{6.30e-02} & \num{1.08e-01} & \num{4.26e+00} & \num{2.46e+00} \\
\bottomrule
\end{tabular}
\end{adjustbox}

\caption{Decision-time RMSE at $t=0$ in the OU premium benchmark (tail median over the last six evaluation checkpoints), sweeping the decision-time uncertainty scale $s_0$ in $Y_0\sim\mathcal{N}(m_0,P_0)$ under both aligned \eqref{eq:P0_aligned} and misaligned \eqref{eq:P0_misaligned} geometries.
DPO rows report the Stage~1 warm-up (deployable) policy output.
PG--DPO rows report the Stage~2 post-hoc $q_0$-aggregated Pontryagin projection output (residual/control-variate implementation; Section~\ref{subsubsec:cv-projection}) computed under the corresponding warm policy.
``Distill.'' rows use interactive distillation (Section~\ref{subsubsec:interactive-distillation}) to amortize the projected teacher into a deployable network.
RMSE is computed identically for all methods (no method-specific or post-hoc scale correction).
The PPO baseline is frequently unstable in this terminal-only, $Y$-blind setting; very large entries (e.g., $s_0=10^{-2},\,d=10$) reflect genuine training divergence rather than typographical error.}
\label{tab:ou_rmse_s0_sweep_landscape}

\endgroup
\end{table*}

\medskip
\noindent\textbf{Metrics and diagnostics.}
Table~\ref{tab:ou_rmse_s0_sweep_landscape} reports the full decision-time RMSE at $t=0$ for all
deployable policies, including the PPO baseline. To isolate the economic hedging channel, we use the
OU-induced decomposition of the analytic reference
\[
  \pi^{\mathrm{ref}}(\tau)=\pi^{\mathrm{myo}}(\tau)+\pi^{\mathrm{hedge}}(\tau),
\]
where $\pi^{\mathrm{myo}}$ is the independence-case reference obtained by setting $\rho=0$
(equivalently, dropping $M_{\mathrm{cross}}(\tau)$ in \eqref{eq:pi-q-const-ou}), and
$\pi^{\mathrm{hedge}}:=\pi^{\mathrm{ref}}-\pi^{\mathrm{myo}}$ is the residual $\rho$-dependent term.
We report the hedging-component RMSE for the projected (Stage~2 / PG--DPO) rules in
Table~\ref{tab:ou_hedge_rmse_s0_sweep_landscape}; the myopic-component RMSE and hedging-direction cosine
similarity are deferred to Appendix~\ref{app:hedging_decomp}. Since Stage~1 (DPO) and PPO do not expose a
compatible myopic/hedging split under our diagnostic protocol, component-wise diagnostics are reported
only for the projected rules.

\medskip
\noindent\textbf{Projection and economic hedging-demand recovery.}
Across all $d$ and $s_0$, the post-hoc Pontryagin projection (PG--DPO) substantially reduces
decision-time RMSE relative to the deployable Stage~1 warm-up (DPO) policy
(Table~\ref{tab:ou_rmse_s0_sweep_landscape}).
For example, under aligned $P_0$ with $s_0=10^{-3}$ and $d=100$, DPO attains $\num{3.54e-03}$ whereas
PG--DPO achieves $\num{1.56e-04}$.
The component-wise diagnostics indicate that the remaining discrepancy is largely driven by the hedging channel:
in the same setting, the hedging RMSE is $\num{1.55e-04}$ (basic) and $\num{1.42e-04}$ (distill.)
(Table~\ref{tab:ou_hedge_rmse_s0_sweep_landscape}),
while the myopic RMSE is an order of magnitude smaller (Appendix Table~\ref{tab:ou_myo_rmse_s0_sweep_landscape}).
This pattern is consistent with the economic mechanism in this benchmark: once the (mostly) myopic component is captured,
the dominant remaining challenge is to recover the intertemporal hedge induced by correlated return--factor shocks.

\medskip
\noindent\textbf{Amortization, robustness, and the PPO baseline.}
Interactive distillation improves the \emph{deployable} Stage~1 policy relative to basic Stage~1 (DPO),
while the most accurate object remains the post-hoc projected policy (PG--DPO)
(Table~\ref{tab:ou_rmse_s0_sweep_landscape}).
This matches the intended division of labor in Section~\ref{subsec:stage1-stage2-coupling}: Stage~2 provides
a structured stationarity-correction signal through the aggregated Pontryagin projection, and distillation
amortizes that correction into a single forward pass, up to policy-class approximation limits.

As the decision-time uncertainty scale $s_0$ increases, both the full RMSE and the hedging-component RMSE
increase, with the most visible degradation at $s_0=10^{-1}$, especially at larger dimensions
(Tables~\ref{tab:ou_rmse_s0_sweep_landscape}--\ref{tab:ou_hedge_rmse_s0_sweep_landscape}).
Misalignment has a limited effect for small and moderate uncertainty scales, but can induce noticeable deterioration in the hardest settings,
where the direction-of-hedge diagnostic can also weaken (Appendix Table~\ref{tab:ou_hedge_cos_s0_sweep_landscape}).

\begin{table}[t!]
\centering
\begingroup
\sisetup{detect-weight=true,detect-family=true}
\providecommand{\best}[1]{{\bfseries\num{#1}}}

\small
\setlength{\tabcolsep}{6pt}
\renewcommand{\arraystretch}{1.10}

\begin{adjustbox}{max width=\columnwidth}
\begin{tabular}{c l cccc}
\toprule
$s_0$ & Method & $d=5$ & $10$ & $50$ & $100$ \\
\midrule
\multicolumn{6}{l}{\textbf{Aligned $P_0$}}\\
\midrule
$10^{-3}$ & PG--DPO (basic)    & \num{4.59e-05} & \best{4.87e-05} & \num{1.37e-04} & \num{1.55e-04} \\
$10^{-3}$ & PG--DPO (distill.) & \best{4.39e-05} & \num{5.25e-05} & \best{1.37e-04} & \best{1.42e-04} \\
\midrule
$10^{-2}$ & PG--DPO (basic)    & \num{4.98e-05} & \best{4.47e-05} & \best{1.44e-04} & \num{1.74e-04} \\
$10^{-2}$ & PG--DPO (distill.) & \best{4.43e-05} & \num{4.95e-05} & \num{1.47e-04} & \best{1.55e-04} \\
\midrule
$10^{-1}$ & PG--DPO (basic)    & \num{5.27e-05} & \best{3.99e-05} & \num{2.60e-04} & \num{2.95e-04} \\
$10^{-1}$ & PG--DPO (distill.) & \best{4.72e-05} & \num{5.36e-05} & \best{2.58e-04} & \best{2.86e-04} \\
\midrule
\multicolumn{6}{l}{\textbf{Misaligned $P_0$}}\\
\midrule
$10^{-3}$ & PG--DPO (basic)    & \num{5.00e-05} & \best{4.87e-05} & \best{1.34e-04} & \num{1.57e-04} \\
$10^{-3}$ & PG--DPO (distill.) & \best{4.30e-05} & \num{5.33e-05} & \num{1.40e-04} & \best{1.43e-04} \\
\midrule
$10^{-2}$ & PG--DPO (basic)    & \num{4.93e-05} & \best{5.65e-05} & \best{1.53e-04} & \num{1.75e-04} \\
$10^{-2}$ & PG--DPO (distill.) & \best{4.56e-05} & \num{6.09e-05} & \best{1.53e-04} & \best{1.56e-04} \\
\midrule
$10^{-1}$ & PG--DPO (basic)    & \num{5.45e-05} & \best{1.55e-04} & \num{3.19e-04} & \num{3.28e-04} \\
$10^{-1}$ & PG--DPO (distill.) & \best{5.20e-05} & \num{1.57e-04} & \best{3.13e-04} & \best{3.20e-04} \\
\bottomrule
\end{tabular}
\end{adjustbox}

\caption{Decision-time RMSE at $t=0$ for the \emph{hedging component} of the OU decision-time reference,
evaluated on the projected (Stage~2 / PG--DPO) policies (tail median over the last six evaluation checkpoints).
We use the decomposition $\pi^{\mathrm{ref}}(\tau)=\pi^{\mathrm{myo}}(\tau)+\pi^{\mathrm{hedge}}(\tau)$ where
$\pi^{\mathrm{myo}}$ is the $\rho=0$ reference and $\pi^{\mathrm{hedge}}$ is the residual $\rho$-dependent term
(cf.\ \eqref{eq:pi-q-const-ou}). Component-wise diagnostics are reported for projected rules since Stage~1 (DPO) and PPO do
not explicitly output a compatible myopic/hedging split under our protocol.}
\label{tab:ou_hedge_rmse_s0_sweep_landscape}

\endgroup
\end{table}

Finally, the PPO baseline remains far from the analytic OU reference under the same $Y$-blind deployability restriction,
and can be unstable in this terminal-only setting (Table~\ref{tab:ou_rmse_s0_sweep_landscape}).
This is consistent with PPO facing a difficult credit-assignment problem under latent-factor heterogeneity,
in contrast to the pathwise-sensitivity and affine-in-control correction exploited by our two-stage pipeline.
Since PPO does not provide a compatible myopic/hedging decomposition under our evaluation protocol, we include it only in the full-RMSE table.


\section{Conclusion}

We studied continuous-time CRRA portfolio choice in diffusion markets whose coefficients are estimated and therefore statistically uncertain
(Section~\ref{subsec:model-objective}).
Our modeling choice is to treat estimation risk as an exogenous \emph{decision-time} input, represented by a law $q(d\theta)$ over a latent parameter $\theta$:
Nature draws $\theta$ once at time $0$ and keeps it fixed over the investment horizon, while the investor does not observe $\theta$ and must deploy a single
$\theta$-blind Markov feedback policy evaluated under an ex--ante objective (Remark~\ref{rem:latent-theta}).
Under this information structure, the appropriate optimality notion is not the infeasible $\theta$-conditional (full-information) criticality, but rather a
\emph{$q$-aggregated} Pontryagin first-order condition that is enforceable within the deployable $\theta$-blind policy class
(Section~\ref{subsec:pmp-latent}, Theorem~\ref{thm:q-agg-foc-theta-blind}).

Methodologically, we proposed a simulation-only two-stage solver, PG--DPO, that matches this deployability constraint.
In Stage~1, DPO optimizes the ex--ante objective by sampling $\theta$ only inside the simulator and performing stochastic gradient ascent via exact discrete-time
pathwise gradients computed by BPTT through an Euler discretization (Section~\ref{subsec:pgdpo-bptt-pmp}).
We then leverage a conditional BPTT--PMP correspondence to extract the costate objects needed for structured control updates, including the second-adjoint blocks
entering the portfolio Hamiltonian gradient (Theorem~\ref{thm:bptt-pmp-uncertainty}).
In Stage~2, we apply a $q$-aggregated Pontryagin projection: we aggregate Monte Carlo Pontryagin quantities across $\theta\sim q$ and project onto the deployable
$q$-aggregated stationarity condition to obtain a single $\theta$-blind rule (Section~\ref{subsec:ppgdpo-uncertainty}).
On the theory side, we established a residual-based ex--ante policy-gap bound under local stability of the aggregated projection map, with discretization and Monte
Carlo errors made explicit (Theorem~\ref{thm:policy-gap-residual}).

Empirically, across benchmarks the projection step provides the dominant gains, improving stability and decision-time accuracy in high-dimensional drift-uncertainty
scaling tests and recovering intertemporal hedging demand in factor-driven markets, while interactive distillation amortizes the projected rule into a fast
deployable network (Sections~\ref{sec:hd-geometry} and~\ref{sec:hedging-recovery}; Section~\ref{subsec:stage1-stage2-coupling}).

Several extensions are natural.
A first direction is to allow time-varying uncertainty descriptions $q_t$ (e.g., produced by an external filter) and connect the present fixed-$q$ projection to
decision-time plug-in replanning and belief-aware decision rules (Remark~\ref{rem:belief-state-vs-agg}, Appendix~\ref{app:kalman-bucy}).
A second direction is to incorporate realistic frictions and constraints---such as transaction costs, leverage, and short-sale limits---and develop certified or
regularized projection steps for regimes where mixed-moment estimation becomes fragile (Section~\ref{app:limits-dp};
Section~\ref{subsec:crra_high_dim_results}; Appendix~\ref{app:impl-stabilizers}).
Finally, applying the framework to large cross-sectional datasets with modern estimation pipelines would further clarify the practical benefits of inference-agnostic,
simulation-only optimization under parameter uncertainty (Section~\ref{sec:intro}).

\section*{Acknowledgments}

This work was supported by the National Research Foundation of Korea (NRF) grant
funded by the Korea government (MSIT) (RS-2025-00562904).

\bibliographystyle{apalike}
\bibliography{gf_bib}


\appendix

\section{Derivations for Gaussian decision-time references}
\label{app:gaussian-references-derivations}

This appendix derives the closed-form decision-time constant-portfolio references stated in Section~\ref{subsec:gaussian-drift}.
Throughout, we fix a decision time $t$ and hold the input law $q_t$ fixed over the remaining horizon $[t,T]$ (plug-in replanning).

\subsection{Static latent drift: tilted FOC and Gaussian shrinkage}
\label{app:gaussian-static-derivation}

Consider the static latent-drift model \eqref{eq:gaussian-model}--\eqref{eq:wealth-gaussian} and restrict to constant portfolio fractions
$\pi_s\equiv\pi$ on $[t,T]$.
Let $\tau:=T-t$ and CRRA utility $U(x)=x^{1-\gamma}/(1-\gamma)$.

\paragraph{Step 1: conditional evaluation under fixed $\theta$.}
For constant $\pi$, conditional on $\theta$ the wealth satisfies
\[
  \log\frac{X_T^\pi}{x}
  =
  \Big(r + \pi^\top \theta - \tfrac12 \pi^\top\Sigma\pi\Big)\tau
  + \pi^\top \Sigma^{1/2}(W_T-W_t),
\]
hence $X_T^\pi$ is lognormal and
\begin{equation}
  \mathbb{E}\big[U(X_T^\pi)\mid \theta,\mathcal{F}_t\big]
  =
  \frac{x^{1-\gamma}}{1-\gamma}\,
  \exp\Big\{ (1-\gamma)r\tau + (1-\gamma)\tau\,\pi^\top\theta
             - \tfrac12 \gamma(1-\gamma)\tau\,\pi^\top\Sigma\pi \Big\}.
  \label{eq:app-crra-cond-constantpi}
\end{equation}

\paragraph{Step 2: ex-ante mixing over $q_t$.}
Mixing over $\theta\sim q_t$ yields
\begin{equation}
  J_t(\pi)
  =
  \frac{x^{1-\gamma}}{1-\gamma}\,
  \exp\Big\{ (1-\gamma)r\tau - \tfrac12 \gamma(1-\gamma)\tau\,\pi^\top\Sigma\pi \Big\}\,
  M_{q_t}\big((1-\gamma)\tau\pi\big),
  \label{eq:crra-exante-constantpi}
\end{equation}
where $M_{q_t}(u):=\mathbb{E}_{\theta\sim q_t}[\exp(u^\top\theta)]$ is the mgf of $q_t$.

\paragraph{Step 3: tilted first-order condition.}
Maximizing $J_t(\pi)$ is equivalent to maximizing $\log J_t(\pi)$.
Differentiating \eqref{eq:crra-exante-constantpi} and using interior optimality gives the tilted FOC
\begin{equation}
  \gamma\,\Sigma\,\pi
  =
  \nabla_u \log M_{q_t}(u)\Big|_{u=(1-\gamma)\tau\,\pi}.
  \label{eq:tilted-FOC-general-q}
\end{equation}

\paragraph{Step 4: Gaussian input law.}
If $q_t=\mathcal{N}(m_t,P_t)$, then $M_{q_t}(u)=\exp\{u^\top m_t+\tfrac12 u^\top P_t u\}$ and
$\nabla_u\log M_{q_t}(u)=m_t+P_tu$.
Substituting into \eqref{eq:tilted-FOC-general-q} yields the linear system
\begin{equation}
  \big(\gamma\Sigma - (1-\gamma)\tau\,P_t\big)\,\pi
  = m_t,
  \label{eq:linear-system-crra-gaussian}
\end{equation}
and hence
\begin{equation}
  \pi
  =
  \big(\gamma\Sigma - (1-\gamma)\tau\,P_t\big)^{-1} m_t
  =
  \big(\gamma\Sigma + (\gamma-1)\tau\,P_t\big)^{-1} m_t,
  \qquad (\gamma>1),
  \label{eq:pi_q_const_gaussian_ref_app}
\end{equation}
which matches \eqref{eq:pi_q_const_gaussian_ref} in the main text.

\subsection{OU premium factor: induced effective law and constant-portfolio reference}
\label{app:gaussian-ou-derivation}

We derive the OU-based reference stated in \eqref{eq:pi-q-const-ou}.
Fix $t$ and let $\tau:=T-t$.
Under \eqref{eq:theta-ou-belief}--\eqref{eq:ou-corr}, we take as decision-time input $Y_t\sim\mathcal{N}(m_t,P_t)$ and hold this input fixed over $[t,T]$.

\subsubsection{Integrated factor: mean and covariance}
\label{app:ou-integrated-moments}

Define $I_{t,T}:=\int_t^T Y_s\,ds$.
Assume $K$ is positive stable (equivalently, $-K$ is Hurwitz), so that $e^{-Ku}$ decays as $u\to\infty$.
For $u\ge 0$ define
\[
  \mathcal{A}(u):=K^{-1}\big(I_m-e^{-Ku}\big),
\]
where $I_m$ is the $m\times m$ identity matrix.
The OU solution is
\[
  Y_{t+u}
  =
  \bar y + e^{-Ku}(Y_t-\bar y) + \int_0^u e^{-K(u-s)}\Xi\,dW_{t+s}^Y.
\]
Integrating over $u\in[0,\tau]$ gives
\[
  I_{t,T}
  =
  \tau\bar y + \mathcal{A}(\tau)(Y_t-\bar y) + \int_0^\tau \mathcal{A}(u)\,\Xi\,dW_{t+u}^Y.
\]
Taking expectation under the input law yields
\begin{equation}
  m_I(t,\tau)
  :=
  \mathbb{E}[I_{t,T}]
  =
  \tau\bar y + \mathcal{A}(\tau)\,(m_t-\bar y),
  \label{eq:mean-integrated-premium-tau}
\end{equation}
and the covariance decomposition gives
\begin{equation}
  C_I(t,\tau)
  :=
  \mathrm{Cov}(I_{t,T})
  =
  \mathcal{A}(\tau)\,P_t\,\mathcal{A}(\tau)^\top
  + \int_0^\tau \mathcal{A}(u)\,\Xi\Xi^\top\,\mathcal{A}(u)^\top\,du.
  \label{eq:cov-integrated-premium-tau}
\end{equation}

\subsubsection{Effective premium law}
\label{app:effective-premium-law}

Define the horizon-averaged effective premium
\[
  \bar\theta_{t,\tau} := \frac{1}{\tau}B\,I_{t,T}.
\]
Then $\bar\theta_{t,\tau}$ is Gaussian with
\begin{equation}
  m_{\bar\theta}(t,\tau)=\frac{1}{\tau}B\,m_I(t,\tau),
  \qquad
  P_{\bar\theta}(t,\tau)=\frac{1}{\tau^2}B\,C_I(t,\tau)\,B^\top,
  \label{eq:theta-bar-law-tau-app}
\end{equation}
matching \eqref{eq:theta-bar-law-tau}.

\subsubsection{Cross term induced by return--factor correlation}
\label{app:cross-term}

Let $\Delta W:=W_T-W_t$.
Using $d\langle W, W^Y\rangle_s=\rho\,ds$, we obtain
\begin{equation}
  C_{IW}(\tau)
  :=
  \mathrm{Cov}\!\big(I_{t,T},\,\Delta W\big)
  =
  \int_0^\tau \mathcal{A}(u)\,\Xi\,\rho^\top\,du
  \;\in\;\mathbb{R}^{m\times d},
  \label{eq:cov_IW_tau}
\end{equation}
and the induced symmetric cross term
\begin{equation}
  M_{\mathrm{cross}}(\tau)
  :=
  B\,C_{IW}(\tau)\,\big(\Sigma^{1/2}\big)^\top
  +
  \Sigma^{1/2}\,C_{IW}(\tau)^\top\,B^\top
  \;\in\;\mathbb{R}^{d\times d}.
  \label{eq:M_cross_tau}
\end{equation}

\subsubsection{Constant-portfolio CRRA objective and FOC}
\label{app:ou-constant-portfolio-foc}

Restrict to constant $\pi$ over $[t,T]$.
Write the excess-return integral as
\[
  \int_t^T \pi^\top dR_s
  =
  \pi^\top B I_{t,T} + \pi^\top \Sigma^{1/2}\Delta W.
\]
As in the static case, CRRA evaluation reduces to the log-mgf of a jointly Gaussian vector.
Collecting the terms involving $\pi$ and applying the interior first-order condition yields the linear system
\begin{equation}
  \Big(\gamma \tau\Sigma + (\gamma-1)\big(B\,C_I(t,\tau)\,B^\top + M_{\mathrm{cross}}(\tau)\big)\Big)\,
  \pi
  = B\,m_I(t,\tau),
  \qquad (\gamma>1),
  \label{eq:pi-q-const-ou-app}
\end{equation}
which matches \eqref{eq:pi-q-const-ou} in the main text.

\paragraph{Independence case.}
If $\rho=0$, then $C_{IW}(\tau)=0$ and $M_{\mathrm{cross}}(\tau)=0$.
In this case, \eqref{eq:pi-q-const-ou-app} reduces to the static Gaussian shrinkage formula applied to the effective law
\eqref{eq:theta-bar-law-tau-app}.

\section{Online uncertainty updates: Kalman--Bucy filtering and a plug-in decision-time benchmark}
\label{app:kalman-bucy}

\noindent\textbf{Purpose and scope.}
Sections~\ref{subsubsec:gaussian-drift-static} and \ref{subsubsec:gaussian-drift-ou-belief} focus on \emph{decision-time} benchmarks in which an
uncertainty description $q$ is treated as given and the investor optimizes under the corresponding $\theta$-blind deployability constraint.
In practice, however, new data arrive and the uncertainty description is updated over time by an external estimation/filtering engine, a viewpoint
that aligns with learning/estimation-risk portfolio choice and Bayesian decision-time formulations \citep{barberis2000investor,pastor2000portfolio,xia2001learning}.
This appendix records a simple linear--Gaussian example in which such an updated description $q_t$ arises endogenously via a Kalman--Bucy filter,
and then formalizes a \emph{plug-in} workflow: at each decision time, treat the current uncertainty description $q_t$ as given and compute a decision-time
optimal control under that $q_t$.

We emphasize that solving the fully optimal partial-observation (belief-state) control problem is \emph{not} the goal of this paper; rather, we view the
resulting $q_t$ as an external input to decision-time optimization. In particular, the simulation-based Pontryagin-guided solvers developed in
Section~\ref{sec:pgdpo-uncertainty} can be used as inner-loop engines that are refreshed whenever a new uncertainty description $q_t$ becomes available.

\medskip
\noindent\textbf{A linear--Gaussian hidden-premium model (OU state, observed returns).}
Let $Y_t\in\mathbb{R}^m$ be a latent premium factor following an OU dynamics
\begin{equation}
  dY_t = K(\bar y - Y_t)\,dt + \Xi\,dW_t^{Y},
  \qquad
  Y_0\sim \mathcal{N}(m_0,P_0),
  \label{eq:online_kb_ou_state_full}
\end{equation}
where $K$ is positive stable (equivalently $-K$ is Hurwitz), $\bar y\in\mathbb{R}^m$, and $\Xi\in\mathbb{R}^{m\times m}$.
Risky assets satisfy
\begin{equation}
  \frac{dS_t}{S_t}
  = r\mathbf{1}\,dt + B Y_t\,dt + \Sigma^{1/2}\,dW_t,
  \qquad \Sigma\in\mathbb{R}^{d\times d}\ \text{s.p.d.},
  \label{eq:online_kb_return_full}
\end{equation}
with $B\in\mathbb{R}^{d\times m}$.
Equivalently, the investor observes the excess-return signal
\begin{equation}
  dZ_t := \frac{dS_t}{S_t} - r\mathbf{1}\,dt
  = B Y_t\,dt + \Sigma^{1/2}\,dW_t.
  \label{eq:online_kb_obs_full}
\end{equation}
We write $\mathbb{F}^Z=(\mathcal{F}_t^Z)_{t\in[0,T]}$ for the filtration generated by $(Z_s)_{s\le t}$.

For clarity, we present the independent-noise case $W\perp W^Y$; the correlated-noise extension remains linear--Gaussian but leads to more cumbersome gain formulas.

\medskip
\noindent\textbf{Kalman--Bucy posterior $q_t=\mathcal{L}(Y_t\mid \mathcal{F}_t^Z)$.}
Under \eqref{eq:online_kb_ou_state_full}--\eqref{eq:online_kb_obs_full}, the conditional law of the latent factor remains Gaussian:
\begin{equation}
  q_t(dy) := \mathcal{L}(Y_t\mid \mathcal{F}_t^Z) = \mathcal{N}(\hat Y_t, P_t),
  \label{eq:online_kb_posterior_full}
\end{equation}
where $(\hat Y_t,P_t)$ satisfy the Kalman--Bucy equations
\begin{align}
  d\hat Y_t
  &= K(\bar y-\hat Y_t)\,dt
   + P_t B^\top \Sigma^{-1}\Big(dZ_t - B\hat Y_t\,dt\Big),
  \label{eq:online_kb_filter_mean_full}\\
  \dot P_t
  &= -K P_t - P_t K^\top + \Xi\Xi^\top - P_t B^\top \Sigma^{-1} B P_t,
  \qquad P_0\ \text{given}.
  \label{eq:online_kb_riccati_full}
\end{align}
Thus, even though $q_t$ is distribution-valued, in this affine/Gaussian regime it is fully characterized by the finite-dimensional sufficient statistics
$(\hat Y_t,P_t)$, with $P_t$ evolving deterministically via \eqref{eq:online_kb_riccati_full}.

\medskip
\noindent\textbf{From a posterior on $Y_t$ to a Gaussian uncertainty description for decision-time optimization.}
Define the remaining-horizon time-averaged premium
\begin{equation}
  \bar \theta_{t,\tau}
  := \frac{1}{\tau}\int_t^T B Y_s\,ds
  \in\mathbb{R}^d,
  \qquad \tau:=T-t.
  \label{eq:online_kb_theta_bar_def_full}
\end{equation}
Let $\mathcal{A}(u):=K^{-1}(I_m-e^{-Ku})$, where $I_m$ is the $m\times m$ identity matrix.
The integrated factor admits the representation
\begin{equation}
  \int_t^T Y_s\,ds
  =
  \tau \bar y
  + \mathcal{A}(\tau)\,(Y_t-\bar y)
  + \int_0^\tau \mathcal{A}(u)\,\Xi\,dW_{t+u}^Y.
  \label{eq:online_kb_integrated_ou_full}
\end{equation}
Conditioning on $\mathcal{F}_t^Z$, we have $Y_t\mid\mathcal{F}_t^Z\sim\mathcal{N}(\hat Y_t,P_t)$ by \eqref{eq:online_kb_posterior_full}, and the future increments
$(W_{t+u}^Y-W_t^Y)_{u\ge 0}$ are independent of $\mathcal{F}_t^Z$ in the independent-noise case. Hence $\bar\theta_{t,\tau}\mid \mathcal{F}_t^Z$
is Gaussian:
\begin{equation}
  \bar\theta_{t,\tau}\mid \mathcal{F}_t^Z
  \sim
  \mathcal{N}\!\big(m_{t,\tau},\,P_{t,\tau}\big),
  \qquad
  m_{t,\tau}:=\frac{B\,m_I(t,\tau)}{\tau},
  \qquad
  P_{t,\tau}:=\frac{1}{\tau^2}\,B\,C_I(t,\tau)\,B^\top,
  \label{eq:online_kb_theta_bar_conditional_gaussian_full}
\end{equation}
where
\begin{align}
  m_I(t,\tau)
  &:= \tau \bar y + \mathcal{A}(\tau)\,(\hat Y_t-\bar y),
  \label{eq:online_kb_mI_full}\\
  C_I(t,\tau)
  &:= \mathcal{A}(\tau)\,P_t\,\mathcal{A}(\tau)^\top
   + \int_0^{\tau} \mathcal{A}(u)\,\Xi\Xi^\top\,\mathcal{A}(u)^\top\,du.
  \label{eq:online_kb_CI_full}
\end{align}
Equation \eqref{eq:online_kb_theta_bar_conditional_gaussian_full} thus provides an explicit example of an \emph{online-updated} Gaussian uncertainty description
\(
  q_{t,\tau} := \mathcal{L}(\bar\theta_{t,\tau}\mid\mathcal{F}_t^Z)=\mathcal{N}(m_{t,\tau},P_{t,\tau}).
\)

\medskip
\noindent\textbf{A plug-in decision-time benchmark (receding-horizon fixed-$q_{t,\tau}$).}
Given $(m_{t,\tau},P_{t,\tau})$ from \eqref{eq:online_kb_theta_bar_conditional_gaussian_full}, a simple decision-time rule is obtained by treating $q_{t,\tau}$ as
fixed over the remaining horizon and applying the Gaussian constant-allocation benchmark of Section~\ref{subsubsec:gaussian-drift-static} with horizon $\tau$:
\begin{equation}
  \pi_{t}^{\mathrm{plug}}
  :=
  \Big(\gamma\,\Sigma + (\gamma-1)\tau\,P_{t,\tau}\Big)^{-1} m_{t,\tau}.
  \label{eq:online_kb_plugin_const_rule_full}
\end{equation}
One may interpret \eqref{eq:online_kb_plugin_const_rule_full} as a \emph{receding-horizon} decision-time policy driven by an externally updated uncertainty
description, consistent with the ``update beliefs, then optimize'' workflow used in Bayesian/learning-based portfolio choice
\citep{barberis2000investor,pastor2000portfolio,xia2001learning}.

\medskip
\noindent\textbf{Remarks (relation to belief-aware control).}
The plug-in rule \eqref{eq:online_kb_plugin_const_rule_full} is intentionally decision-time: it conditions on the current uncertainty description and does
not optimize over how the posterior will evolve. Even in linear--Gaussian regimes where the belief state is finite-dimensional, the \emph{fully optimal}
partial-observation portfolio problem would treat the belief state (here, $(\hat Y_t,P_t)$) as part of the controlled state and optimize the policy in that
enlarged state space \citep{bensoussan1985optimal,pham2017dynamic}.
Related necessary conditions under partial information can also be expressed via partial-observation maximum principles
\citep{haussmann1987maximum,li1995general,baghery2007maximum}.
Developing a belief-aware Pontryagin-guided policy optimizer that operates directly in belief-state space is an important direction that we defer to future work.


\section{Why dynamic programming and deep PDE surrogates are not primary baselines under the $\theta$-blind fixed-$q$ interface}
\label{app:limits-dp}

This appendix clarifies why we do \emph{not} treat classical dynamic programming (DP/HJB) or value-function-based deep PDE surrogates
(PINNs / deep BSDE methods) as primary baselines for the high-dimensional uncertain markets targeted in this paper.
The main obstruction is \emph{not} that PDE solvers are ``invalid'';
rather, under our \emph{$\theta$-blind fixed-$q$ deployability} interface, the \emph{natural DP/PDE/BSDE target is typically not a low-dimensional Markov value function}.
Derivative instability and high-dimensional conditioning issues are secondary to this target mismatch.

\subsection{The value-function target for latent $\theta$ is a belief-state functional}
\label{app:limits-dp-belief}

\paragraph{Fixed-$q$ ex--ante objective and $\theta$-blind deployability.}
Our decision-time interface is
\begin{equation}
  \sup_{\pi\in\mathcal{A}^{\mathrm{obs}}}
  J(\pi)
  \;:=\;
  \sup_{\pi\in\mathcal{A}^{\mathrm{obs}}}
  \mathbb{E}_{\theta\sim q}\!\left[
    \mathbb{E}\!\left[U(X_T^\pi)\mid \theta\right]
  \right],
  \label{eq:app-exante-objective}
\end{equation}
where $\theta\sim q$ is latent, fixed along each episode, and $\mathcal{A}^{\mathrm{obs}}$ denotes the class of \emph{deployable} controls,
i.e.\ controls measurable w.r.t.\ the \emph{observable} filtration (so deployed policies cannot condition on $\theta$).
(One may further restrict to the feedback subclass $\mathcal{A}^{\mathrm{fb}}$ as defined in Section~\ref{subsec:model-objective}.)

\paragraph{Why $(t,X_t,Y_t)$ is not a closed DP state under latent $\theta$.}
Even if the \emph{prior} $q$ is fixed at decision time (and even if the controller is \emph{constrained} to be $\theta$-blind),
the \emph{conditional law} of $\theta$ given observations typically evolves as data arrive.
Formally, the continuation value at time $t$ depends on the conditional law
\begin{equation}
  q_t(\cdot) := \mathcal{L}\!\big(\theta \mid \mathcal{F}_t^{\mathrm{obs}}\big),
  \label{eq:app-belief-def}
\end{equation}
not only on $(t,X_t,Y_t)$.
This is the standard partially observed control phenomenon: the \emph{unconditional} (mixed-over-$\theta$) law of the observed state
is generally \emph{not Markov} in $(X_t,Y_t)$, and DP becomes Markov only after augmenting the state by a belief
(or a finite-dimensional sufficient statistic when available) \citep{bensoussan1985optimal,pham2017dynamic}.

\paragraph{Belief-state DP is infinite-dimensional under our interface.}
The principled value object is therefore a functional of a probability measure:
\begin{equation}
  V(t,x,y,\eta)
  :=
  \sup_{\pi\in\mathcal{A}^{\mathrm{obs}}}
  \mathbb{E}^{t,x,y,\eta}\!\left[U(X_T^\pi)\right],
  \qquad \eta\in\mathcal{P}(\Theta),
  \label{eq:app-belief-value}
\end{equation}
where $\mathbb{E}^{t,x,y,\eta}$ denotes the expectation for a system started from $(t,x,y)$ with latent parameter $\theta\sim\eta$.
Hence the associated DP/HJB is posed on a (typically infinite-dimensional) space of measures in general \citep{bensoussan1985optimal,pham2017dynamic}.

Under our inference-agnostic interface, $q$ need not be conjugate and we do not specify a filtering model that would keep $q_t$ in a tractable
finite-dimensional family. Consequently, any value-function baseline that assumes a low-dimensional Markov target is, by construction,
solving a \emph{different} object from the $\theta$-blind fixed-$q$ problem studied in this paper.

\subsection{Even with deterministic parameters, value-function baselines are indirect under our interface}
\label{app:limits-dp-hjb-derivatives}

When parameters are deterministic, DP leads to an HJB for a Markov value function $V(t,x,y)$ \citep{fleming2006controlled,pham2009continuous}.
In portfolio problems, however, optimal feedback rules are typically recovered through nonlinear operations on \emph{derivatives} of $V$, e.g.
\begin{equation}
  \pi^\star(t,x,y)
  \;\propto\;
  -\frac{V_x(t,x,y)}{V_{xx}(t,x,y)} \times (\text{risk model})^{-1}\times(\text{drift model}),
  \label{eq:app-policy-from-derivatives}
\end{equation}
with intertemporal hedging driven by mixed sensitivities such as $V_{xy}$.
This ``derivative bottleneck'' is a practical issue for grid-based solvers \citep{bellman1961adaptive,kushner2001numerical} and also affects neural surrogates,
but it is \emph{not} the core reason we exclude DP/PDE baselines here. The core reason is the belief-state mismatch in
Section~\ref{app:limits-dp-belief}.

\subsection{Deep PDE surrogates (PINNs / deep BSDE) do not resolve the target mismatch}
\label{app:limits-dp-deep-pde}

PINNs and deep BSDE methods replace grids with neural approximators trained on sampled points/paths
\citep{raissi2019physics,sirignano2018dgm,han2018solving,beck2019machine,hure2020deep}.
However, both remain value-function-based PDE/BSDE solvers: they implicitly assume a (finite-dimensional) Markov target for the value/decoupling field.
Under latent $\theta$, a solver must therefore choose between:
\begin{itemize}
\item learning $\theta$-conditional objects (making the policy/value depend on $\theta$, hence \emph{not deployable}); or
\item enforcing a $\theta$-agnostic Markov representation (e.g.\ a single $V(t,x,y)$) by sampling $\theta$ only inside the fitting loss, which generally
does \emph{not} correspond to the $\theta$-blind ex--ante optimum \eqref{eq:app-exante-objective}.
\end{itemize}
The next two subsections make this mismatch explicit.

\subsection{Why $\theta$-conditional value functions do not solve the $\theta$-blind ex--ante problem}
\label{app:limits-dp-theta-conditional}

\paragraph{$\theta$-conditional HJB / BSDE.}
If $\theta$ were observed, the full-information value function $V^\theta$ solves
\begin{equation}
  0
  =
  \partial_t V^\theta(t,x,y)
  + \sup_{\pi\in\mathbb{R}^d}
  \Big\{
    \mathcal{L}^{\pi,\theta} V^\theta(t,x,y)
  \Big\},
  \qquad
  V^\theta(T,x,y)=U(x),
  \label{eq:app-hjb-theta}
\end{equation}
\citep{fleming2006controlled,pham2009continuous}, yielding a $\theta$-dependent optimizer $\pi^\star(t,x,y,\theta)$.
This is infeasible under our $\theta$-blind deployability restriction.

\paragraph{Supremum--expectation non-commutativity is the structural obstruction.}
Even perfect access to $\theta$-conditional optimizers does not solve the $\theta$-blind objective because
\begin{equation}
  \sup_{\pi\in\mathcal{A}^{\mathrm{obs}}}
  \mathbb{E}_{\theta\sim q}\!\left[J^\theta(\pi)\right]
  \;\le\;
  \mathbb{E}_{\theta\sim q}\!\left[\sup_{\pi} J^\theta(\pi)\right],
  \qquad
  J^\theta(\pi):=\mathbb{E}\!\left[U(X_T^\pi)\mid\theta\right],
  \label{eq:app-sup-expectation-gap}
\end{equation}
and the inequality is typically strict. The right-hand side corresponds to an (infeasible) $\theta$-adaptive controller.
Equivalently, our Pontryagin stationarity is $q$-aggregated over \emph{$\theta$-blind variations}, not enforced parameter-by-parameter.

\subsection{Why ``sampling $\theta$ only in the PDE/BSDE loss'' is mismatched to the ex--ante objective}
\label{app:limits-dp-collocation}

A common workaround is to learn a single value/decoupling network that does \emph{not} take $\theta$ as input,
and sample $\theta\sim q$ only inside a training loss.
For PINNs this appears as an averaged collocation residual; for deep BSDE methods it appears as an averaged pathwise BSDE/HJB fitting objective.
Schematically,
\begin{equation}
  \min_\varphi\;
  \mathbb{E}_{\theta\sim q}\!\Big[
    \big\| \mathrm{HJB}_\theta\!\big(V_\varphi\big)(t,x,y)\big\|^2
  \Big],
  \qquad \text{(schematic)}
  \label{eq:app-avg-residual}
\end{equation}
where $\mathrm{HJB}_\theta(\cdot)$ denotes the full-information HJB operator for parameter $\theta$.

The issue is conceptual (target mismatch), not merely numerical:
\begin{itemize}
\item \textbf{Conflicting optimality conditions.} Different $\theta$ impose different local optimality conditions at the same $(t,x,y)$, so there is generally no
single low-dimensional Markov $V$ (or decoupling field) that simultaneously satisfies $\mathrm{HJB}_\theta(V)=0$ for many $\theta$.
\item \textbf{Loss mismatch.} Minimizing an averaged PDE/BSDE fitting loss such as \eqref{eq:app-avg-residual} is not equivalent to maximizing expected terminal utility
under $\theta$-blind controls in \eqref{eq:app-exante-objective}.
\item \textbf{Wrong stationarity notion.} Our deployable optimum is characterized by a \emph{$q$-aggregated} first-order condition with $\theta$-blind admissible variations;
averaging \emph{full-information} HJB/BSDE conditions across $\theta$ targets a different stationarity concept.
\end{itemize}

Therefore, while $\theta$-sampling can be a sensible \emph{approximation heuristic}, it does not produce a principled DP/PDE/BSDE baseline aligned with the
fixed-$q$, $\theta$-blind optimality notion studied in this paper.

\subsection{A toy calculation: $\mathbb{E}[\pi^\star(\theta)]\neq \pi_q^\star$ under CRRA}
\label{app:limits-dp-toy}

The aggregation issue is visible even in a one-asset Merton model with static latent drift:
\[
  \frac{dS_t}{S_t}= r\,dt + \theta\,dt + \sigma\,dW_t,
  \qquad
  \theta\sim\mathcal{N}(m,P),
\]
and constant fractions $\pi$ over remaining horizon $\tau$.
For CRRA utility $U(x)=x^{1-\gamma}/(1-\gamma)$, $\gamma>1$, the $\theta$-conditional full-information rule is
\begin{equation}
  \pi^\star(\theta)=\frac{\theta}{\gamma\sigma^2}.
  \label{eq:app-merton-theta-conditional}
\end{equation}
Averaging gives $\mathbb{E}[\pi^\star(\theta)]=m/(\gamma\sigma^2)$.
By contrast, the deployable $\theta$-blind ex--ante optimizer within the constant-fraction class is
\begin{equation}
  \pi_{q}^{\mathrm{const}}(\tau)
  =
  \frac{m}{\gamma\sigma^2 + (\gamma-1)\tau\,P},
  \qquad (\gamma>1),
  \label{eq:app-merton-q-opt}
\end{equation}
which exhibits shrinkage and horizon dependence. Thus $\mathbb{E}[\pi^\star(\theta)]\neq \pi_q^{\mathrm{const}}(\tau)$ whenever $P>0$.

\subsection{Implication for baselines in this paper}
\label{app:limits-dp-implication}

The discussion above motivates our baseline choices:
\begin{itemize}
\item We validate on controlled decision-time Gaussian benchmarks that match the fixed-$q$ deployability restriction.
\item For learning-based comparisons, we emphasize simulation-only direct policy optimization methods that operate under the same observation restriction
(no access to $\theta$; $\theta$ sampled only inside the simulator), rather than value-function PDE/BSDE solvers whose natural DP target is belief-augmented.
\end{itemize}


\section{Proof of Theorem~\ref{thm:bptt-pmp-uncertainty}}
\label{app:bptt-pmp-proof}

\begin{proof}[Proof of Theorem~\ref{thm:bptt-pmp-uncertainty}]
Theorem~\ref{thm:bptt-pmp-uncertainty} extends the BPTT--PMP (equivalently, BPTT--BSDE) correspondence established
for deterministic-parameter models in our prior work on PG--DPO (see~\citet{huh2025breaking}).
Here the only substantive change is that the market coefficients are indexed by a random but
\emph{frozen} parameter \(\theta\sim q\), and we need convergence statements that hold
\emph{conditionally on \(\theta\)} and \emph{uniformly over \(\theta\)} in compact subsets of \(\Theta\).

\medskip

\noindent\textbf{Important remark (what this proof does \emph{not} use).}
This proof concerns the \(\theta\)-conditional Pontryagin adjoint/costate for the fixed-\(\theta\)
control problem induced by \eqref{eq:wealth-theta-section2}--\eqref{eq:factor-theta-section2}.
It does \emph{not} use the $q$-aggregated $\theta$-blind stationarity condition
(e.g.\ Theorem~\ref{thm:q-agg-foc-theta-blind}), which enters only the stage~2 projection target.

\medskip

\noindent\textbf{Notation, filtration, and the pathwise-vs-adapted issue.}
Fix \(\theta\in\Theta\). We work conditionally on this \(\theta\) and consider the augmented
(simulator) filtration
\[
  \mathbb{G}^\theta := (\mathcal{G}_t^\theta)_{t\in[0,T]},
  \qquad
  \mathcal{G}_t^\theta := \sigma\!\big(\theta,\{W_s,W_s^Y:0\le s\le t\}\big)
  \ \text{(with the usual augmentation)}.
\]
All conditional expectations and conditional \(L^2\) projections below are taken w.r.t.\ \(\mathcal{G}_{t_k}^\theta\).

\emph{Key point.}
The raw BPTT ``adjoints'' are \emph{pathwise derivatives} of the terminal utility along the full
discrete computation graph. Hence, for \(k<N\), they are typically \(\mathcal{G}_T^\theta\)-measurable and
need \emph{not} be \(\mathcal{G}_{t_k}^\theta\)-measurable (they can depend on future increments).
In contrast, the stochastic maximum principle (SMP/PMP) costate is \(\mathbb{G}^\theta\)-adapted.
Therefore, the correspondence is formulated \emph{after} a standard one-step conditional \(L^2\)
projection (equivalently, a discrete-time BSDE representation), which turns pathwise objects
into adapted discrete adjoints.

\medskip

Let \(W\) be \(d_W\)-dimensional and \(W^Y\) be \(d_Y\)-dimensional (allowing instantaneous correlation).
Fix \(\Delta t>0\), \(t_k:=k\Delta t\), \(k=0,\dots,N\), with \(N\Delta t=T\).
For readability we suppress the policy parameters \(\varphi\) and write
\(\pi_k := \pi_\varphi(t_k,X_k^{\Delta t,\theta},Y_k^{\Delta t,\theta})\), where \(\pi_\varphi\) is \(\theta\)-blind.

\medskip

\noindent\textbf{Standing assumptions (smooth Markov regime; uniform on compacts).}
Fix a compact set \(\Theta_0\subset\Theta\). Assume the following hold \emph{uniformly for \(\theta\in \Theta_0\)}:
\begin{itemize}
  \item \textbf{Uniform well-posedness and moments.}
  The coefficients in \eqref{eq:wealth-theta-section2}--\eqref{eq:factor-theta-section2} satisfy the usual
  Lipschitz and linear-growth conditions (in the state variables), yielding unique strong solutions and uniform \(L^2\) moment bounds.
  \item \textbf{Uniform smoothness.}
  The coefficients are sufficiently smooth in state variables so that the \(\theta\)-conditional adjoint BSDE is well posed and the
  \emph{state-derivative blocks used in the portfolio Hamiltonian gradient} exist and are square-integrable.
  (It is sufficient to assume \(a,b,\beta,\sigma\) are \(C^2\) in state variables with derivatives obeying uniform-in-\(\theta\) growth/Lipschitz bounds on \(\Theta_0\).)
  \item \textbf{Policy regularity.}
  The deployed feedback \((t,x,y)\mapsto \pi_\varphi(t,x,y)\) is \(C^2\) in \((x,y)\) on the working domain,
  with the required derivatives uniformly bounded there.
  \item \textbf{Smooth decoupling field (identifying the blocks used here).}
  In the smooth Markov regime, the \(\theta\)-conditional adjoint admits a decoupling field
  \(p_t^\theta = u^\theta(t,X_t^{\pi,\theta},Y_t^\theta)\), and we may identify
  \(p_{x,t}^\theta:=\partial_x u^\theta(t,X_t^{\pi,\theta},Y_t^\theta)\),
  \(p_{y,t}^\theta:=\partial_y u^\theta(t,X_t^{\pi,\theta},Y_t^\theta)\).
  These are precisely the sensitivity blocks needed for the portfolio Hamiltonian gradient.
  \item \textbf{One-step projection nondegeneracy.}
  The conditional covariance matrix of \((\Delta W_k,\Delta W_k^Y)\) given \(\mathcal{G}_{t_k}^\theta\) is uniformly nondegenerate on \(\Theta_0\),
  so that the one-step conditional \(L^2\)-projection onto \(\mathrm{span}\{1,\Delta W_k,\Delta W_k^Y\}\) is well defined with unique coefficients.
\end{itemize}
These are the deterministic assumptions used in~\citet{huh2025breaking}, now imposed \emph{conditionally on \(\theta\)}
and \emph{uniformly} over \(\theta\in \Theta_0\).

\medskip

\noindent\textbf{Step 1: Conditioning on \(\theta\) and uniformity of forward Euler error.}
Under the standing assumptions, for each fixed \(\theta\in \Theta_0\) the controlled SDE system is well posed and admits uniform-in-time
\(L^2\) moment bounds. Moreover, the Euler--Maruyama scheme enjoys the standard strong error bound
\[
  \sup_{t\in[0,T]}
  \mathbb{E}\!\left[\left\|(X_t^{\pi,\theta},Y_t^\theta)-(X_t^{\Delta t,\theta},Y_t^{\Delta t,\theta})\right\|^2\right]^{1/2}
  \le C_{\Theta_0}\,\Delta t^{1/2},
\]
with a constant \(C_{\Theta_0}\) independent of \(\theta\in \Theta_0\).

\medskip

\noindent\textbf{Step 2: Discrete forward scheme and BPTT pathwise adjoints (fixed \(\theta\)).}
Fix \(\theta\in \Theta_0\). Consider the Euler scheme on the grid \((t_k)\):
\[
  Y_{k+1}^{\Delta t,\theta}
  =
  Y_k^{\Delta t,\theta}
  + a(Y_k^{\Delta t,\theta},\theta)\Delta t
  + \beta(Y_k^{\Delta t,\theta},\theta)\Delta W_k^Y,
\]
\[
  X_{k+1}^{\Delta t,\theta}
  =
  X_k^{\Delta t,\theta}
  + X_k^{\Delta t,\theta}\Big(r + \pi_k^\top b(Y_k^{\Delta t,\theta},\theta)\Big)\Delta t
  + X_k^{\Delta t,\theta}\,\pi_k^\top \sigma(Y_k^{\Delta t,\theta},\theta)\Delta W_k,
\]
with terminal reward \(U(X_N^{\Delta t,\theta})\).

Define the discrete \emph{pathwise} wealth costate
\[
  p_k^{\mathrm{pw},\theta}
  :=
  \frac{\partial}{\partial X_k^{\Delta t,\theta}}\,U(X_N^{\Delta t,\theta}),
  \qquad k=0,\dots,N,
\]
and the additional pathwise blocks
\[
  p_{x,k}^{\mathrm{pw},\theta}
  :=
  \frac{\partial p_k^{\mathrm{pw},\theta}}{\partial X_k^{\Delta t,\theta}},
  \qquad
  p_{y,k}^{\mathrm{pw},\theta}
  :=
  \frac{\partial p_k^{\mathrm{pw},\theta}}{\partial Y_k^{\Delta t,\theta}}.
\]
Automatic differentiation/BPTT computes
\(\{(p_k^{\mathrm{pw},\theta},p_{x,k}^{\mathrm{pw},\theta},p_{y,k}^{\mathrm{pw},\theta})\}_{k=0}^N\)
via the backward chain rule along the discrete forward graph.
Conditionally on \(\theta\), this is exactly the deterministic-parameter BPTT construction of~\citet{huh2025breaking}.

\medskip

\noindent\textbf{Step 3: One-step conditional \(L^2\) projection \(\Rightarrow\) adapted discrete BSDE variables.}
Fix \(\theta\in \Theta_0\). On each step, perform the conditional \(L^2\)-projection of \(p_{k+1}^{\mathrm{pw},\theta}\) onto
\(\mathrm{span}\{1,\Delta W_k,\Delta W_k^Y\}\) given \(\mathcal{G}_{t_k}^\theta\).
Equivalently, define regression coefficients \(z_k^\theta\in\mathbb{R}^{d_W}\) and \(\tilde z_k^\theta\in\mathbb{R}^{d_Y}\) such that
\[
  p_{k+1}^{\mathrm{pw},\theta}
  =
  \mathbb{E}\!\left[p_{k+1}^{\mathrm{pw},\theta}\mid \mathcal{G}_{t_k}^\theta\right]
  + (z_k^\theta)^\top \Delta W_k
  + (\tilde z_k^\theta)^\top \Delta W_k^Y
  + \varepsilon_{k+1}^\theta,
\]
where \(\varepsilon_{k+1}^\theta\) is orthogonal (in conditional \(L^2\)) to
\(\mathrm{span}\{1,\Delta W_k,\Delta W_k^Y\}\) given \(\mathcal{G}_{t_k}^\theta\).

Define the adapted discrete-time adjoint blocks by conditional expectation:
\[
  p_k^{\Delta t,\theta}
  := \mathbb{E}\!\left[p_k^{\mathrm{pw},\theta}\mid \mathcal{G}_{t_k}^\theta\right],\quad
  p_{x,k}^{\Delta t,\theta}
  := \mathbb{E}\!\left[p_{x,k}^{\mathrm{pw},\theta}\mid \mathcal{G}_{t_k}^\theta\right],\quad
  p_{y,k}^{\Delta t,\theta}
  := \mathbb{E}\!\left[p_{y,k}^{\mathrm{pw},\theta}\mid \mathcal{G}_{t_k}^\theta\right].
\]
By construction, \((p_k^{\Delta t,\theta},z_k^\theta,\tilde z_k^\theta)\) satisfy a discrete-time BSDE representation.
Moreover, substituting the projection into the one-step BPTT recursion and taking conditional expectations yields that the discrete driver
coincides with the Euler discretization of the \(\theta\)-conditional adjoint BSDE under the fixed policy \(\pi_\varphi\),
up to standard one-step remainder terms controlled by Taylor/Euler estimates.
Under the standing smoothness assumptions, the projection residual \(\varepsilon_{k+1}^\theta\) contributes only higher-order
discretization error; all bounds can be chosen uniformly for \(\theta\in \Theta_0\) (as in~\citet{huh2025breaking}).

\medskip

\noindent\textbf{Step 4: Convergence as \(\Delta t\to 0\) and uniformity on compacts.}
Let \((p_t^\theta,z_t^\theta,\tilde z_t^\theta)\) denote the continuous-time \(\theta\)-conditional adjoint triple under \(\pi_\varphi\),
and let \((p_{x,t}^\theta,p_{y,t}^\theta)\) denote the corresponding sensitivity blocks (in smooth Markov regimes, value-derivative blocks).

Define the piecewise-constant interpolations of the \emph{adapted} discrete variables on \([0,T]\):
\[
  p_t^{\Delta t,\theta} := p_k^{\Delta t,\theta},\quad
  p_{x,t}^{\Delta t,\theta} := p_{x,k}^{\Delta t,\theta},\quad
  p_{y,t}^{\Delta t,\theta} := p_{y,k}^{\Delta t,\theta},\quad
  z_t^{\Delta t,\theta} := z_k^\theta,\quad
  \tilde z_t^{\Delta t,\theta} := \tilde z_k^\theta,
  \qquad t\in[t_k,t_{k+1}).
\]
Standard stability of Euler for the forward SDE together with convergence of the corresponding discrete BSDE scheme
implies, for each fixed \(\theta\in \Theta_0\),
\[
  \sup_{t\in[0,T]}\mathbb{E}\big[\,|p_t^{\Delta t,\theta} - p_t^\theta|^2 \,\big]\to 0,\qquad
  \sup_{t\in[0,T]}\mathbb{E}\big[\,\|p_{x,t}^{\Delta t,\theta} - p_{x,t}^\theta\|^2 \,\big]\to 0,\qquad
  \sup_{t\in[0,T]}\mathbb{E}\big[\,\|p_{y,t}^{\Delta t,\theta} - p_{y,t}^\theta\|^2 \,\big]\to 0,
\]
and, for the martingale coefficients,
\[
  \mathbb{E}\!\left[\int_0^T \|z_t^{\Delta t,\theta}-z_t^\theta\|^2\,dt\right]\to 0,\qquad
  \mathbb{E}\!\left[\int_0^T \|\tilde z_t^{\Delta t,\theta}-\tilde z_t^\theta\|^2\,dt\right]\to 0,
\]
as \(\Delta t\to 0\).
Because all Lipschitz, growth, nondegeneracy, and smoothness constants were assumed uniform on \(\Theta_0\),
the constants in these convergence bounds can be chosen independently of \(\theta\in \Theta_0\).
This yields the claimed BPTT--PMP correspondence conditionally on \(\theta\) and uniformly over \(\theta\) in compact subsets of \(\Theta\).
\end{proof}


\section{Auxiliary results for Theorem~\ref{thm:policy-gap-residual}}
\label{app:aux-thm3}

\subsection{Stability of the projection map $(A,g)\mapsto -A^{-1}g$}
\label{app:proj-map-stability}

\begin{proposition}[Stability of the projection map $(A,g)\mapsto -A^{-1}g$]
\label{prop:proj-map-stability}
Let $\mathcal{D}$ be a measurable domain and let $\mu$ be a reference measure on $\mathcal{D}$.
Let $A,\widetilde A : \mathcal{D} \to \mathbb{R}^{d\times d}$ and $g,\widetilde g : \mathcal{D} \to \mathbb{R}^d$ be measurable.
Assume:
\begin{enumerate}
\item[(i)] $A(z)$ is invertible for $\mu$-a.e.\ $z\in \mathcal{D}$ and $\|A^{-1}\|_{L^\infty(\mathcal{D})} \le \kappa$ for some $\kappa>0$;
\item[(ii)] $\|g\|_{L^\infty(\mathcal{D})} \le M_g$ for some $M_g>0$;
\item[(iii)] $\|\widetilde A - A\|_{L^\infty(\mathcal{D})} \le (2\kappa)^{-1}$.
\end{enumerate}
Define $\pi := -A^{-1}g$ and $\widetilde\pi := -\widetilde A^{-1}\widetilde g$.
Then $\widetilde A(z)$ is invertible for $\mu$-a.e.\ $z\in \mathcal{D}$ with $\|\widetilde A^{-1}\|_{L^\infty(\mathcal{D})} \le 2\kappa$, and
\begin{equation}
\label{eq:proj-map-stability}
\|\widetilde\pi - \pi\|_{L^2(\mu)}
\;\le\;
2\kappa\,\|\widetilde g - g\|_{L^2(\mu)}
\;+\;
2\kappa^2 M_g\,\|\widetilde A - A\|_{L^2(\mu)}.
\end{equation}
\end{proposition}

\begin{proof}
Throughout, $\|\cdot\|$ denotes the operator norm induced by the Euclidean norm, and
$\|\cdot\|_{L^\infty(\mathcal{D})}$ denotes the $\mu$-essential supremum of the corresponding pointwise norm on $\mathcal{D}$.

\noindent\textbf{Step 1: invertibility and uniform bound for $\widetilde A^{-1}$.}
For $\mu$-a.e.\ $z$, write
\[
\widetilde A(z) \;=\; A(z)\Bigl(I_d + A(z)^{-1}\bigl(\widetilde A(z)-A(z)\bigr)\Bigr),
\]
where $I_d$ is the $d\times d$ identity matrix.
By (i) and (iii),
\[
\bigl\|A^{-1}(z)\bigl(\widetilde A(z)-A(z)\bigr)\bigr\|
\;\le\;
\|A^{-1}\|_{L^\infty(\mathcal{D})}\,\|\widetilde A-A\|_{L^\infty(\mathcal{D})}
\;\le\;\kappa\cdot(2\kappa)^{-1}
\;=\;\tfrac12.
\]
Hence $I_d + A^{-1}(\widetilde A-A)$ is invertible and
$\bigl\|\bigl(I_d + A^{-1}(\widetilde A-A)\bigr)^{-1}\bigr\|\le (1-\tfrac12)^{-1}=2$.
Therefore $\widetilde A(z)$ is invertible and
\[
\|\widetilde A^{-1}(z)\|
\;=\;
\bigl\|\bigl(I_d + A^{-1}(\widetilde A-A)\bigr)^{-1}A^{-1}(z)\bigr\|
\;\le\;2\|A^{-1}\|_{L^\infty(\mathcal{D})}
\;\le\;2\kappa,
\]
so $\|\widetilde A^{-1}\|_{L^\infty(\mathcal{D})}\le 2\kappa$.

\noindent\textbf{Step 2: difference of inverses.}
Using the identity $\widetilde A^{-1}-A^{-1}=\widetilde A^{-1}(A-\widetilde A)A^{-1}$ and H\"older ($L^\infty\times L^2\times L^\infty\to L^2$),
\[
\|\widetilde A^{-1}-A^{-1}\|_{L^2(\mu)}
\;\le\;
\|\widetilde A^{-1}\|_{L^\infty(\mathcal{D})}\,\|\widetilde A-A\|_{L^2(\mu)}\,\|A^{-1}\|_{L^\infty(\mathcal{D})}
\;\le\;
2\kappa^2\,\|\widetilde A-A\|_{L^2(\mu)}.
\]

\noindent\textbf{Step 3: control error bound.}
Since $\pi=-A^{-1}g$ and $\widetilde\pi=-\widetilde A^{-1}\widetilde g$,
\[
\widetilde\pi-\pi
=
-\widetilde A^{-1}(\widetilde g-g) - (\widetilde A^{-1}-A^{-1})g.
\]
Taking $L^2(\mu)$ norms and applying H\"older ($L^\infty\times L^2\to L^2$) yields
\[
\|\widetilde\pi-\pi\|_{L^2(\mu)}
\le
\|\widetilde A^{-1}\|_{L^\infty(\mathcal{D})}\,\|\widetilde g-g\|_{L^2(\mu)}
+
\|\widetilde A^{-1}-A^{-1}\|_{L^2(\mu)}\,\|g\|_{L^\infty(\mathcal{D})}.
\]
Using Step 1, Step 2, and (ii),
\[
\|\widetilde\pi-\pi\|_{L^2(\mu)}
\le
2\kappa\,\|\widetilde g-g\|_{L^2(\mu)}
+
2\kappa^2\,M_g\,\|\widetilde A-A\|_{L^2(\mu)},
\]
which is \eqref{eq:proj-map-stability}.
\end{proof}

\subsection{Slab-wise small-gain for the $q$-aggregated projection inputs}
\label{app:slab-small-gain}

\paragraph{Time-slab decomposition.}
Assume the working domain carries a time coordinate and, for concreteness,
$\mathcal{D} \subset [0,T]\times \mathcal{S}$ and $\mu(dt,d\xi)=dt\otimes \nu(d\xi)$ for some reference measure $\nu$ on $\mathcal{S}$.
Fix a partition $0=t_0<t_1<\cdots<t_L=T$ with slab lengths $\tau_k:=t_k-t_{k-1}$ and define
\[
\mathcal{D}_k := \mathcal{D}\cap([t_{k-1},t_k]\times \mathcal{S}),\qquad
\mu_k := \mu|_{\mathcal{D}_k},\qquad
\|f\|_k := \|f\|_{L^2(\mu_k)}.
\]
Then $\|f\|_{L^2(\mu)}^2=\sum_{k=1}^L \|f\|_k^2$.

\begin{proposition}[Short-time (slab) Lipschitz gain]
\label{prop:slab-lipschitz-gain}
Let $\mathcal{U}$ be a subset of the admissible policy class that is an $L^2(\mu)$-neighborhood of a locally optimal deployable
$\theta$-blind policy $\pi^{\star,\mathrm{blind}}$ (in particular $\pi^{\star,\mathrm{blind}}\in \mathcal{U}$), and on which the following
uniform bounds hold for all $\pi\in \mathcal{U}$:
\[
\|A_\pi^{-1}\|_{L^\infty(\mathcal{D})}\le \kappa,
\qquad
\|g^{\mathrm{mix}}_\pi\|_{L^\infty(\mathcal{D})}\le M_g
\]
for some constants $\kappa,M_g>0$.
Assume that on each slab $\mathcal{D}_k$ the $q$-aggregated projection inputs satisfy
\begin{equation}
\label{eq:slab-Lip-AG}
\|A_{\pi_1}-A_{\pi_2}\|_k \le \bar L_A\,\tau_k^{1/2}\,\|\pi_1-\pi_2\|_k,
\qquad
\|g^{\mathrm{mix}}_{\pi_1}-g^{\mathrm{mix}}_{\pi_2}\|_k \le \bar L_g\,\tau_k^{1/2}\,\|\pi_1-\pi_2\|_k,
\end{equation}
for all $\pi_1,\pi_2\in \mathcal{U}$, with constants $\bar L_A,\bar L_g>0$ depending only on band data.
Define the (population) projection map $\mathcal{T}(\pi):=-A_\pi^{-1}g^{\mathrm{mix}}_\pi$ and
\[
\varrho(\tau) := \Bigl(\kappa\,\bar L_g + \kappa^2 M_g\,\bar L_A\Bigr)\tau^{1/2}.
\]
Then for each slab $\mathcal{D}_k$ and all $\pi_1,\pi_2\in \mathcal{U}$,
\[
\|\mathcal{T}(\pi_1)-\mathcal{T}(\pi_2)\|_k \le \varrho(\tau_k)\,\|\pi_1-\pi_2\|_k.
\]
In particular, if the partition is chosen so that
\[
\varrho_* := \max_{1\le k\le L}\varrho(\tau_k) < 1,
\]
then $\mathcal{T}$ is a strict contraction on every slab with constant at most $\varrho_*$.
\end{proposition}

\begin{proof}
Fix $k$ and $\pi_1,\pi_2\in \mathcal{U}$. Set $A_i:=A_{\pi_i}$ and $g_i:=g^{\mathrm{mix}}_{\pi_i}$.
Then
\[
\mathcal{T}(\pi_1)-\mathcal{T}(\pi_2) = -A_1^{-1}(g_1-g_2) - (A_1^{-1}-A_2^{-1})g_2,
\qquad
A_1^{-1}-A_2^{-1} = A_1^{-1}(A_2-A_1)A_2^{-1}.
\]
Using H\"older ($L^\infty\times L^2\to L^2$) on $\mathcal{D}_k$ together with $\|A_i^{-1}\|_{L^\infty(\mathcal{D})}\le \kappa$ and $\|g_2\|_{L^\infty(\mathcal{D})}\le M_g$ yields
\[
\|\mathcal{T}(\pi_1)-\mathcal{T}(\pi_2)\|_k
\le
\kappa\|g_1-g_2\|_k + \kappa^2 M_g \|A_1-A_2\|_k.
\]
Applying \eqref{eq:slab-Lip-AG} gives the claim with $\varrho(\tau_k)$.
\end{proof}

\begin{remark}[Where the $\tau^{1/2}$ gain comes from; relation to \citet{huh2025breaking}]
\label{rem:tau-half-gain}
The short-time factor $\tau^{1/2}$ in \eqref{eq:slab-Lip-AG} is the same parabolic smoothing
mechanism used in our earlier slab-wise PGDPO analysis:
\begin{itemize}
\item One represents the relevant value-gap / adjoint differences on a short slab via a
Duhamel (semigroup) formula, differentiates in space, and applies Young-type convolution bounds.
\item In mixed norms, the gradient smoothing bound contributes a factor $\tau^{1/2-1/q}$ and the
trivial embedding $L^\infty_t\hookrightarrow L^q_t$ contributes $\tau^{1/q}$, yielding net $\tau^{1/2}$.
\end{itemize}
A full PDE derivation of this short-time contraction philosophy is given in
\citet[Appendix~B]{huh2025breaking}; see also the companion supplementary note to \citet{huh2025breaking}
(``On Proof of Theorem~4 (Policy Gap Bound)'') for a step-by-step slab contraction
argument in a closely related setting.
In the present paper the same smoothing principle is applied not to the costate update map directly,
but to the induced $q$-aggregated projection inputs $(A_\pi,g^{\mathrm{mix}}_\pi)$.
\end{remark}

\subsection{Proof of Theorem~\ref{thm:policy-gap-residual}}
\label{app:policy-gap-proof}

\begin{proof}[Proof of Theorem~\ref{thm:policy-gap-residual}]
We keep the notation of the main text.
Let $\pi^{\mathrm{warm}}$ be the warm-up (deployable, $\theta$-blind) policy and define the population
$q$-aggregated projection inputs $(A_{\mathrm{warm}},g_{\mathrm{warm}})$ by
\[
A_{\mathrm{warm}} := A_{\pi^{\mathrm{warm}}},\qquad
g_{\mathrm{warm}} := g^{\mathrm{mix}}_{\pi^{\mathrm{warm}}}.
\]
Define the \emph{population projected} policy
\[
\pi^{\mathrm{proj}} := \mathcal{T}(\pi^{\mathrm{warm}}) = -A_{\mathrm{warm}}^{-1}g_{\mathrm{warm}}.
\]
Let $\widehat A$ and $\widehat g$ be the stage-2 (BPTT/Monte Carlo) estimators of
$A_{\mathrm{warm}}$ and $g_{\mathrm{warm}}$ used to form the deployable projected control
\[
\widehat\pi^{\mathrm{agg,mix}} := -\widehat A^{-1}\widehat g,
\]
where $\widehat A$ is understood as the \emph{effective} matrix after any inversion stabilizers
(ridge/diagnostics/skip) so that $\widehat A$ is well-defined and invertible on the working domain.
Let $\pi^{\star,\mathrm{blind}}$ denote a locally optimal interior deployable $\theta$-blind policy,
so in particular it satisfies the fixed-point (stationarity) identity
\[
\pi^{\star,\mathrm{blind}} = \mathcal{T}(\pi^{\star,\mathrm{blind}}).
\]

\medskip\noindent
\textbf{Step 1 (Estimator-to-population projection stability).}
To apply Proposition~\ref{prop:proj-map-stability} with $A=A_{\mathrm{warm}}$ and $g=g_{\mathrm{warm}}$,
we assume we are in the stable inversion regime enforced by the stage-2 stabilizers, i.e.
\begin{equation}
\label{eq:stable-inversion-assumption}
\|\widehat A-A_{\mathrm{warm}}\|_{L^\infty(\mathcal{D})} \le (2\kappa)^{-1},
\end{equation}
where $\kappa$ is such that $\|A_{\mathrm{warm}}^{-1}\|_{L^\infty(\mathcal{D})}\le \kappa$.
Define the $L^2$ estimation error level
\[
\delta_{\mathrm{BPTT}} := \|\widehat A-A_{\mathrm{warm}}\|_{L^2(\mu)}+\|\widehat g-g_{\mathrm{warm}}\|_{L^2(\mu)}.
\]
Then Proposition~\ref{prop:proj-map-stability} yields $\|\widehat A^{-1}\|_{L^\infty(\mathcal{D})}\le 2\kappa$ and
\begin{equation}
\label{eq:step1-est-to-pop}
\|\widehat\pi^{\mathrm{agg,mix}}-\pi^{\mathrm{proj}}\|_{L^2(\mu)}
\le
2\kappa\,\|\widehat g-g_{\mathrm{warm}}\|_{L^2(\mu)}
+
2\kappa^2\,\|g_{\mathrm{warm}}\|_{L^\infty(\mathcal{D})}\,\|\widehat A-A_{\mathrm{warm}}\|_{L^2(\mu)}
\le
C_2\,\delta_{\mathrm{BPTT}},
\end{equation}
where one may take, for instance,
\[
C_2 := 2\kappa + 2\kappa^2\,\|g_{\mathrm{warm}}\|_{L^\infty(\mathcal{D})}.
\]
(Any equivalent control of the joint $L^2$ estimation error in $(\widehat A,\widehat g)$ suffices, provided the stable inversion condition
\eqref{eq:stable-inversion-assumption} holds.)

\medskip\noindent
\textbf{Step 2 (Slab-wise contraction for the population projection map).}
Assume $\pi^{\mathrm{warm}},\pi^{\star,\mathrm{blind}}\in \mathcal{U}$ and let $\varrho_*<1$ be the slab contraction constant of
Proposition~\ref{prop:slab-lipschitz-gain}.
Since $\pi^{\mathrm{proj}}=\mathcal{T}(\pi^{\mathrm{warm}})$ and $\pi^{\star,\mathrm{blind}}=\mathcal{T}(\pi^{\star,\mathrm{blind}})$,
for each slab $\mathcal{D}_k$ we have
\begin{equation}
\label{eq:step2-contraction}
\|\pi^{\mathrm{proj}}-\pi^{\star,\mathrm{blind}}\|_k
=
\|\mathcal{T}(\pi^{\mathrm{warm}})-\mathcal{T}(\pi^{\star,\mathrm{blind}})\|_k
\le
\varrho(\tau_k)\,\|\pi^{\mathrm{warm}}-\pi^{\star,\mathrm{blind}}\|_k
\le
\varrho_*\,\|\pi^{\mathrm{warm}}-\pi^{\star,\mathrm{blind}}\|_k.
\end{equation}

\medskip\noindent
\textbf{Step 3 (Residual identity and slab-wise warm-up deviation).}
Define the population mixed-moment aggregated stationarity residual of the warm-up policy
\[
r_{\mathrm{FOC,mix}}^{\mathrm{warm}} := A_{\mathrm{warm}}\pi^{\mathrm{warm}} + g_{\mathrm{warm}},
\qquad
\varepsilon_{\mathrm{warm}}^{\mathrm{mix}} := \|r_{\mathrm{FOC,mix}}^{\mathrm{warm}}\|_{L^2(\mu)},
\qquad
\varepsilon_{\mathrm{warm},k}^{\mathrm{mix}} := \|r_{\mathrm{FOC,mix}}^{\mathrm{warm}}\|_k.
\]
By definition of $\pi^{\mathrm{proj}}$,
\[
\pi^{\mathrm{warm}}-\pi^{\mathrm{proj}}
=
A_{\mathrm{warm}}^{-1}r_{\mathrm{FOC,mix}}^{\mathrm{warm}}.
\]
Hence, on each slab,
\begin{equation}
\label{eq:step3-residual}
\|\pi^{\mathrm{warm}}-\pi^{\mathrm{proj}}\|_k
\le
\|A_{\mathrm{warm}}^{-1}\|_{L^\infty(\mathcal{D})}\,\|r_{\mathrm{FOC,mix}}^{\mathrm{warm}}\|_k
\le
\kappa\,\varepsilon_{\mathrm{warm},k}^{\mathrm{mix}}.
\end{equation}

\medskip\noindent
\textbf{Step 4 (Slab-wise closure).}
Combine the triangle inequality with \eqref{eq:step2-contraction} and \eqref{eq:step3-residual}:
\[
\|\pi^{\mathrm{warm}}-\pi^{\star,\mathrm{blind}}\|_k
\le
\kappa\,\varepsilon_{\mathrm{warm},k}^{\mathrm{mix}} + \varrho_*\,\|\pi^{\mathrm{warm}}-\pi^{\star,\mathrm{blind}}\|_k.
\]
Since $\varrho_*<1$, we can close slab-wise:
\begin{equation}
\label{eq:step4-close}
\|\pi^{\mathrm{warm}}-\pi^{\star,\mathrm{blind}}\|_k
\le
\frac{\kappa}{1-\varrho_*}\,\varepsilon_{\mathrm{warm},k}^{\mathrm{mix}}.
\end{equation}
Plugging \eqref{eq:step4-close} into \eqref{eq:step2-contraction} yields
\begin{equation}
\label{eq:step4-proj-gap-slab}
\|\pi^{\mathrm{proj}}-\pi^{\star,\mathrm{blind}}\|_k
\le
\frac{\varrho_*\kappa}{1-\varrho_*}\,\varepsilon_{\mathrm{warm},k}^{\mathrm{mix}}.
\end{equation}

\medskip\noindent
\textbf{Step 5 (Global bound and completion).}
Sum over slabs using $\|f\|_{L^2(\mu)}^2=\sum_{k=1}^L\|f\|_k^2$:
\[
\|\pi^{\mathrm{proj}}-\pi^{\star,\mathrm{blind}}\|_{L^2(\mu)}
\le
\frac{\varrho_*\kappa}{1-\varrho_*}\,\varepsilon_{\mathrm{warm}}^{\mathrm{mix}}.
\]
Finally, combine with \eqref{eq:step1-est-to-pop} and the triangle inequality:
\[
\|\widehat\pi^{\mathrm{agg,mix}}-\pi^{\star,\mathrm{blind}}\|_{L^2(\mu)}
\le
\|\widehat\pi^{\mathrm{agg,mix}}-\pi^{\mathrm{proj}}\|_{L^2(\mu)}
+
\|\pi^{\mathrm{proj}}-\pi^{\star,\mathrm{blind}}\|_{L^2(\mu)}
\le
C_2\,\delta_{\mathrm{BPTT}}
+
\frac{\varrho_*\kappa}{1-\varrho_*}\,\varepsilon_{\mathrm{warm}}^{\mathrm{mix}},
\]
which is the claimed slab-wise residual-based $\theta$-blind policy-gap bound (under the stable inversion condition \eqref{eq:stable-inversion-assumption}).
\end{proof}


\section{Implementation details for Section~\ref{sec:pgdpo-uncertainty}}
\label{app:impl-sec3}

This appendix provides reproducible step-by-step templates for the methods in
Section~\ref{sec:pgdpo-uncertainty}. The high-level pipeline is summarized in
Figure~\ref{fig:sec3-pipeline-sideways}. For compactness we present one template per
subsection of Section~\ref{sec:pgdpo-uncertainty}.

\medskip
\noindent\textbf{Conventions.}
Throughout, the deployed policy is \(\theta\)-blind. The latent parameter \(\theta\sim q\) is sampled
\emph{only inside the simulator}. Stage~2 computes projection ingredients on a working domain
\(\mathcal{D}\subset[0,T]\times(0,\infty)\times\mathbb{R}^m\) with reference measure \(\mu\).

\subsection{Stage~1 (DPO) template for Section~\ref{subsec:pgdpo-bptt-pmp}}
\label{app:impl-31}

Stage~1 performs stochastic gradient ascent on the fixed-$q$ ex--ante objective
\eqref{eq:exante-objective-section3}, with latent \(\theta\sim q\) sampled inside the simulator
while the policy remains \(\theta\)-blind.

\paragraph{Inputs.}
Policy parameters \(\varphi\); sampler \(\nu\) over initial states;
prior \(q\); time grid \((N,\Delta t)\); batch size \(M\); optimizer and step size \(\alpha\).

\paragraph{Template (one training iteration).}
\begin{enumerate}
  \item \textbf{Sample initial states.}
  Draw a mini-batch \(\{z_0^{(i)}=(t_0^{(i)},x_0^{(i)},y_0^{(i)})\}_{i=1}^M \sim \nu\).

  \item \textbf{Sample latent parameter(s) inside the simulator.}
  Sample \(\theta\sim q\) and keep it frozen along the rollout(s). (Variant: sample \(\theta^{(i)}\sim q\) independently per episode;
  both are unbiased for the discretized gradient of \(J(\varphi)\).)

  \item \textbf{Simulate Euler rollouts.}
  For each episode \(i\), simulate the Euler scheme for
  \eqref{eq:wealth-theta-section2}--\eqref{eq:factor-theta-section2} under the \(\theta\)-blind policy
  \(\pi_\varphi\), and compute terminal utilities \(\{U(X_T^{(i)})\}_{i=1}^M\).

  \item \textbf{BPTT (pathwise gradient).}
  Compute the Monte Carlo gradient estimator
  \[
    \widehat g \;\gets\; \frac{1}{M}\sum_{i=1}^M \nabla_\varphi U(X_T^{(i)}).
  \]

  \item \textbf{Parameter update.}
  Update \(\varphi \leftarrow \varphi + \alpha \cdot \texttt{OptimizerStep}(\widehat g)\),
  consistent with the ascent form \eqref{eq:dpo-update}.

  \item \textbf{Checkpoint.}
  Periodically save a warm-up checkpoint \(\varphi^{\mathrm{warm}}\) for stage~2 projection.
\end{enumerate}

\subsection{Stage~2 (Pontryagin projection; mixed-moment $q$-aggregation) template for Section~\ref{subsec:ppgdpo-uncertainty}}
\label{app:impl-32}

Stage~2 is a post-processing map: given a warm-up \(\theta\)-blind policy \(\pi_{\varphi^{\mathrm{warm}}}\),
it estimates Pontryagin sensitivity objects by Monte Carlo and constructs a \emph{deployable} projected
control on queried working-domain states \(z\sim\mu\).

The aggregation used here matches the mixed-moment $q$-aggregation in
\eqref{eq:agg-A}--\eqref{eq:agg-G}, yielding the projected control \eqref{eq:pi-agg-mix}.

\paragraph{Inputs.}
Warm-up policy \(\pi_{\varphi^{\mathrm{warm}}}\); working-domain sampler \(\mu\) on \(\mathcal{D}\);
budgets \((M_z, M_\theta, M_{\mathrm{MC}})\).

\paragraph{Template (constructing projection targets on a batch of query states).}
\begin{enumerate}
  \item \textbf{Sample working-domain query states.}
  Draw \(\{z_j=(t_j,x_j,y_j)\}_{j=1}^{M_z}\sim \mu\).

  \item \textbf{For each query state \(z_j\), sample frozen parameters.}
  Sample \(\{\theta_\ell\}_{\ell=1}^{M_\theta}\sim q\).

  \item \textbf{For each frozen \(\theta_\ell\), estimate costate blocks at \(z_j\).}
  For each \(\ell=1,\dots,M_\theta\):
  \begin{enumerate}
    \item Simulate \(M_{\mathrm{MC}}\) trajectories from \(z_j\) under \(\pi_{\varphi^{\mathrm{warm}}}\) with frozen \(\theta_\ell\).
    \item Compute pathwise sensitivities by autodiff/BPTT and average as in \eqref{eq:mc-costate-estimates} to obtain
    \(\widehat p_t^{\theta_\ell}(z_j)\), \(\widehat p_{x,t}^{\theta_\ell}(z_j)\), \(\widehat p_{y,t}^{\theta_\ell}(z_j)\).
    \item Form the \(\theta\)-conditional inputs (cf.\ \eqref{eq:Ahat-theta}--\eqref{eq:Ghat-theta}):
    \[
      \widehat A_t^{\theta_\ell}(z_j)
      \gets
      x_j\,\widehat p_{x,t}^{\theta_\ell}(z_j)\,\Sigma(y_j,\theta_\ell),
      \qquad
      \widehat g_t^{\theta_\ell}(z_j)
      \gets
      \widehat p_t^{\theta_\ell}(z_j)\,b(y_j,\theta_\ell)
      +
      \Sigma_{SY}(y_j,\theta_\ell)\,\widehat p_{y,t}^{\theta_\ell}(z_j).
    \]
  \end{enumerate}

  \item \textbf{Aggregate across \(\theta\sim q\) (mixed-moment).}
  Compute
  \[
    \widehat A_t(z_j) \;\gets\; \frac{1}{M_\theta}\sum_{\ell=1}^{M_\theta}\widehat A_t^{\theta_\ell}(z_j),
    \qquad
    \widehat g_t^{\mathrm{mix}}(z_j) \;\gets\; \frac{1}{M_\theta}\sum_{\ell=1}^{M_\theta}\widehat g_t^{\theta_\ell}(z_j),
  \]
  consistent with \eqref{eq:agg-A}--\eqref{eq:agg-G}.

  \item \textbf{Solve the projection (deployable \(\theta\)-blind control).}
  Whenever \(\widehat A_t(z_j)\) is numerically invertible and passes stability diagnostics,
  compute
  \[
    \widehat\pi^{\mathrm{agg,mix}}(z_j)
    \;\gets\;
    -\big(\widehat A_t(z_j)\big)^{-1}\widehat g_t^{\mathrm{mix}}(z_j),
  \]
  which matches \eqref{eq:pi-agg-mix}.
  Otherwise, fall back to \(\pi_{\varphi^{\mathrm{warm}}}(z_j)\) (or a ridge-stabilized solve; see Appendix~\ref{app:impl-stabilizers}).
\end{enumerate}

\subsection{Coupling I: residual/control-variate projection (Section~\ref{subsubsec:cv-projection})}
\label{app:impl-331}

This subsection records a variance-reduced implementation of the stage~2 map using the residual identity
\eqref{eq:pi-agg-residual-form}. The residual form is applied around the warm-up policy and uses the mixed-moment
aggregated inputs \((\widehat A_t,\widehat g_t^{\mathrm{mix}})\).

\paragraph{Inputs.}
Warm-up policy \(\pi_{\varphi^{\mathrm{warm}}}\); query state(s) \(z=(t,x,y)\sim\mu\); and
stage~2 ingredients \((\widehat A_t(z),\widehat g_t^{\mathrm{mix}}(z))\) from Appendix~\ref{app:impl-32}.

\paragraph{Template (statewise residual projection; mixed-moment aggregation).}
\begin{enumerate}
  \item \textbf{Evaluate warm-up control.}
  Compute \(\pi_{\varphi^{\mathrm{warm}}}(z)\).

  \item \textbf{Form the aggregated residual.}
  Compute
  \[
    \widehat r_{\mathrm{FOC}}(z)
    \;\gets\;
    \widehat A_t(z)\,\pi_{\varphi^{\mathrm{warm}}}(z) + \widehat g_t^{\mathrm{mix}}(z).
  \]

  \item \textbf{Apply the residual correction.}
  Compute
  \[
    \widehat\pi^{\mathrm{agg,mix}}(z)
    \;\gets\;
    \pi_{\varphi^{\mathrm{warm}}}(z)
    -
    \big(\widehat A_t(z)\big)^{-1}\widehat r_{\mathrm{FOC}}(z),
  \]
  with ridge/diagnostics/skip as needed (Appendix~\ref{app:impl-stabilizers}).
\end{enumerate}

\subsection{Coupling II: interactive distillation (Section~\ref{subsubsec:interactive-distillation})}
\label{app:impl-332}

Interactive distillation amortizes the stage~2 projected control into the policy network.
A lagged policy copy \(\pi_{\varphi^-}\) is used to construct a teacher on a slower timescale, and the student \(\pi_\varphi\)
is trained against a mixed objective \eqref{eq:mixed-objective-with-teacher-narr}.

\paragraph{Inputs.}
Student parameters \(\varphi\); teacher refresh interval \(n_{\mathrm{refresh}}\);
distillation schedule \(\lambda(n)\); working-domain sampler \(\mu\).

\paragraph{Template (training loop with intermittent teacher refresh).}
\begin{enumerate}
  \item \textbf{Initialize.}
  Set \(\varphi^- \leftarrow \varphi\) and initialize an empty teacher buffer \(\mathcal{B}\leftarrow\emptyset\).

  \item \textbf{Repeat for iterations \(n=1,2,\dots\):}
  \begin{enumerate}
    \item \textbf{Stage~1 update (DPO step).}
    Perform one DPO ascent step on \(J(\varphi)\) as in Appendix~\ref{app:impl-31}.

    \item \textbf{Teacher refresh (every \(n_{\mathrm{refresh}}\) steps).}
    If \(n \bmod n_{\mathrm{refresh}} = 0\):
    \begin{enumerate}
      \item Set \(\varphi^- \leftarrow \varphi\).
      \item Sample working-domain states \(\{z_j\}_{j=1}^{M_z}\sim\mu\).
      \item For each \(z_j\), run stage~2 under \(\pi_{\varphi^-}\) (mixed-moment aggregation) to compute the projected teacher
      \(\widehat\pi^{\mathrm{agg,mix}}_{\varphi^-}(z_j)\) (typically in residual form \eqref{eq:teacher-pi-from-33-1-narr}).
      \item Optionally filter unreliable states using diagnostics (Appendix~\ref{app:impl-stabilizers}) and update the buffer:
      \[
        \mathcal{B}\ \leftarrow\ \{(z_j,\widehat\pi^{\mathrm{agg,mix}}_{\varphi^-}(z_j))\}_{j=1}^{M_z}\ \ \text{(after filtering)}.
      \]
    \end{enumerate}

    \item \textbf{Distillation step (when enabled).}
    If \(\lambda(n)>0\) and \(\mathcal{B}\neq\emptyset\):
    \begin{enumerate}
      \item Sample \((z,\pi^{\mathrm{teach}})\) from \(\mathcal{B}\).
      \item Apply one optimizer step for the teacher proximity term in \eqref{eq:mixed-objective-with-teacher-narr}
      (treating \(\pi^{\mathrm{teach}}\) as \(\mathrm{stopgrad}\)-fixed).
    \end{enumerate}
  \end{enumerate}
\end{enumerate}

\subsection{Engineering notes and stabilizers}
\label{app:impl-stabilizers}

This subsection collects practical stabilizers that are useful for reliable training and projection.

\begin{itemize}
  \item \textbf{Antithetic sampling for \(\theta\).}
  When \(q\) is symmetric (e.g.\ Gaussian in a latent-normal parameterization), sample \(\theta\) in antithetic pairs.
  This reduces the variance of \(q\)-averaged quantities and improves stage~2 concentration.

  \item \textbf{Blockwise Monte Carlo and robust aggregation.}
  Split Monte Carlo replications into blocks and aggregate blockwise estimates using robust statistics (median / median-of-means)
  to reduce tail domination.

  \item \textbf{Curvature/denominator checks for \(\widehat A_t\).}
  Monitor conditioning of \(\widehat A_t(z)\) (e.g.\ eigenvalue floors / condition numbers).
  If unstable, skip projection at \(z\) or use ridge:
  \[
    \widehat A_{t,\lambda}(z) := \widehat A_t(z)+\lambda I_d,
    \qquad
    \widehat\pi_{\lambda}(z) := -\widehat A_{t,\lambda}(z)^{-1}\widehat g_t^{\mathrm{mix}}(z).
  \]

  \item \textbf{Residual magnitude as a reliability diagnostic.}
  Track \(\widehat r_{\mathrm{FOC}}(z)=\widehat A_t(z)\pi^{\mathrm{warm}}(z)+\widehat g_t^{\mathrm{mix}}(z)\).
  Small \(\|\widehat r_{\mathrm{FOC}}(z)\|\) empirically correlates with reliable teacher targets.

  \item \textbf{Diagnostics-based teacher selection.}
  For distillation, keep only \(z\) that pass stability predicates (invertibility/conditioning + moderate residual),
  to prevent noisy stage~2 targets from contaminating the teacher buffer.

  \item \textbf{Scheduling \(\lambda\) (teacher loss).}
  Use an initial warm-up with \(\lambda=0\) (pure DPO), then increase \(\lambda\) only after stage~2 diagnostics stabilize.
  Optionally cap the effective coefficient via
  \[
    \lambda_{\mathrm{eff}}
    :=
    \min\Big\{\lambda,\;
    c\,\frac{|L_{\mathrm{main}}|}{L_{\mathrm{distill}}+\varepsilon}\Big\},
  \]
  with \(c\in(0,1)\), \(\varepsilon>0\).

  \item \textbf{Initialization/scale control in high dimensions.}
  Initialize policy outputs near zero and/or scale outputs by \(d^{-1/2}\) to avoid early-time blow-ups
  (e.g.\ via \(\pi^\top\Sigma\pi\)).
\end{itemize}


\section{Coupling stage~1 and stage~2: residual projection and interactive distillation}
\label{app:coupling}
\label{subsec:stage1-stage2-coupling}

We keep the ex--ante objective \eqref{eq:exante-objective-section3} and the $\theta$-blind deployability constraint throughout.
Stage~2 is \emph{not} a separate optimization problem: it is a (warm-started) post-processing map that reuses the current stage~1 policy as a warm-up control,
estimates projection ingredients under this warm-up policy, and then applies a $q$-aggregated Pontryagin projection as a deterministic transformation.

Concretely, stage~2 takes \(\pi^{\mathrm{warm}}\) as input, runs BPTT/Monte Carlo under frozen \(\theta\sim q\) to estimate the adjoint blocks required for
the Hamiltonian gradient, aggregates them across \(\theta\), and outputs a single deployable \(\theta\)-blind rule \(\widehat\pi^{\mathrm{agg,mix}}\)
(Section~\ref{subsec:ppgdpo-uncertainty}).

This appendix records two couplings between the two stages:
\begin{itemize}
  \item \textbf{Residual (control-variate) projection.}
  The projected rule is evaluated in a residual form around the warm-up policy. This is algebraically equivalent to the direct solve,
  but can reduce Monte Carlo variance and improve numerical stability.

  \item \textbf{Interactive distillation (amortized projection).}
  The projected rule is used as a teacher signal at intermittent refresh times. Distillation amortizes the cost of projection by compressing
  projected controls into the policy network.
\end{itemize}

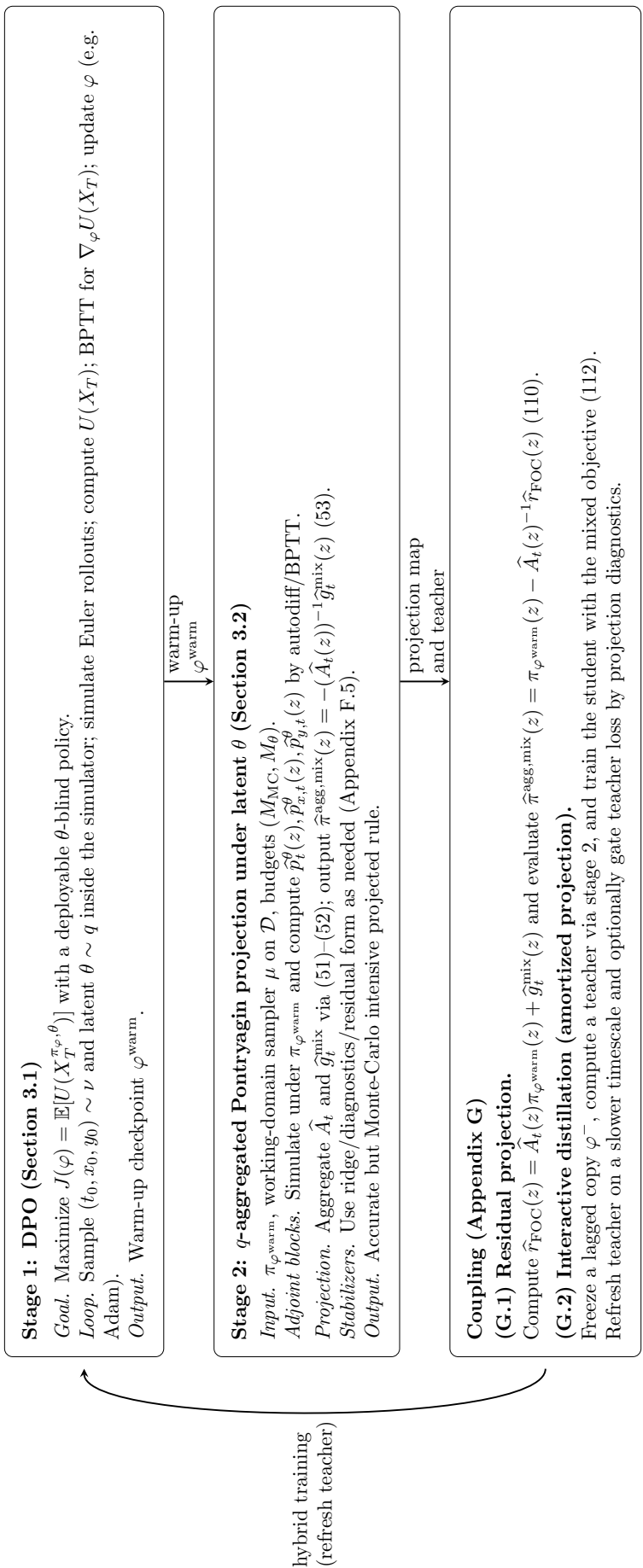
\begin{sidewaysfigure*}[p]
\centering
\resizebox{0.99\textheight}{!}{%
\begin{tikzpicture}[
  font=\small,
  >=stealth,
  box/.style={
    draw,
    rounded corners,
    align=left,
    inner sep=8pt,
    text width=0.92\linewidth
  },
  arrow/.style={->, thick}
]

\node[box] (s1) {%
  \textbf{Stage~1: DPO (Section~\ref{subsec:pgdpo-bptt-pmp})}\\[2pt]
  \textit{Goal.} Maximize \(J(\varphi)=\mathbb{E}[U(X_T^{\pi_\varphi,\theta})]\) with a deployable \(\theta\)-blind policy.\\
  \textit{Loop.} Sample \((t_0,x_0,y_0)\sim\nu\) and latent \(\theta\sim q\) inside the simulator; simulate Euler rollouts; compute \(U(X_T)\);
  BPTT for \(\nabla_\varphi U(X_T)\); update \(\varphi\) (e.g.\ Adam).\\
  \textit{Output.} Warm-up checkpoint \(\varphi^{\mathrm{warm}}\).
};

\node[box, below=24pt of s1] (s2) {%
  \textbf{Stage~2: $q$-aggregated Pontryagin projection under latent \(\theta\) (Section~\ref{subsec:ppgdpo-uncertainty})}\\[2pt]
  \textit{Input.} \(\pi_{\varphi^{\mathrm{warm}}}\), working-domain sampler \(\mu\) on \(\mathcal{D}\), budgets \((M_{\mathrm{MC}},M_\theta)\).\\
  \textit{Adjoint blocks.} Simulate under \(\pi_{\varphi^{\mathrm{warm}}}\) and compute \(\widehat p_t^\theta(z),\widehat p_{x,t}^\theta(z),\widehat p_{y,t}^\theta(z)\) by autodiff/BPTT.\\
  \textit{Projection.} Aggregate \(\widehat A_t\) and \(\widehat g_t^{\mathrm{mix}}\) via \eqref{eq:agg-A}--\eqref{eq:agg-G}; output
  \(\widehat\pi^{\mathrm{agg,mix}}(z)=-(\widehat A_t(z))^{-1}\widehat g_t^{\mathrm{mix}}(z)\) \eqref{eq:pi-agg-mix}.\\
  \textit{Stabilizers.} Use ridge/diagnostics/residual form as needed (Appendix~\ref{app:impl-stabilizers}).\\
  \textit{Output.} Accurate but Monte-Carlo intensive projected rule.
};

\node[box, below=24pt of s2] (s3) {%
  \textbf{Coupling (Appendix~\ref{app:coupling})}\\[2pt]
  \textbf{(\ref{subsubsec:cv-projection}) Residual projection.}\\
  Compute \(\widehat r_{\mathrm{FOC}}(z)=\widehat A_t(z)\pi_{\varphi^{\mathrm{warm}}}(z)+\widehat g_t^{\mathrm{mix}}(z)\) and evaluate
  \(\widehat\pi^{\mathrm{agg,mix}}(z)=\pi_{\varphi^{\mathrm{warm}}}(z)-\widehat A_t(z)^{-1}\widehat r_{\mathrm{FOC}}(z)\) \eqref{eq:pi-agg-residual-form}.\\[4pt]
  \textbf{(\ref{subsubsec:interactive-distillation}) Interactive distillation (amortized projection).}\\
  Freeze a lagged copy \(\varphi^-\), compute a teacher via stage~2, and train the student with the mixed objective \eqref{eq:mixed-objective-with-teacher-narr}.\\
  Refresh teacher on a slower timescale and optionally gate teacher loss by projection diagnostics.
};

\draw[arrow] (s1.south) -- (s2.north) node[midway, right, align=left] {warm-up\\\(\varphi^{\mathrm{warm}}\)};
\draw[arrow] (s2.south) -- (s3.north) node[midway, right, align=left] {projection map\\and teacher};

\draw[arrow] ([xshift=-0.2cm]s3.west) .. controls +(-1.0,0.0) and +(-1.0,0.0) .. ([xshift=-0.2cm]s1.west)
node[midway, left, align=left] {hybrid training\\(refresh teacher)};

\end{tikzpicture}%
}
\caption{Pipeline of Section~\ref{sec:pgdpo-uncertainty}: stage~1 learning, stage~2 $q$-aggregated projection, and the coupling mechanisms in Appendix~\ref{app:coupling}. Distillation stabilizes training and amortizes projection by compressing projected controls into the policy network.}
\label{fig:sec3-pipeline-sideways}
\end{sidewaysfigure*}

\subsection{Control-variate (residual) form of the projected rule}
\label{subsubsec:cv-projection}

Recall the mixed-moment projected rule \eqref{eq:pi-agg-mix}.
In high dimensions, Monte Carlo noise in the projection inputs can be non-negligible, and solving a linear system with \(\widehat A_t\) can amplify this noise.
A convenient stabilization is to compute the \emph{same} projected rule in a residual (control-variate) form around the warm-up policy \(\pi_{\varphi^{\mathrm{warm}}}\).

Define the \(\theta\)-conditional residual (under frozen-\(\theta\) simulations)
\begin{equation}
  \widehat r_{\mathrm{FOC}}^\theta(t,x,y)
  :=
  \widehat A_t^\theta(t,x,y)\,\pi_{\varphi^{\mathrm{warm}}}(t,x,y)
  + \widehat g_t^\theta(t,x,y),
  \label{eq:r-foc-theta-recall-33}
\end{equation}
and the aggregated residual (the quantity we actually solve against)
\begin{equation}
  \widehat r_{\mathrm{FOC}}(t,x,y)
  :=
  \widehat A_t(t,x,y)\,\pi_{\varphi^{\mathrm{warm}}}(t,x,y) + \widehat g_t^{\mathrm{mix}}(t,x,y).
  \label{eq:r-foc-agg-33}
\end{equation}
Whenever \(\widehat A_t(t,x,y)\) is (numerically) invertible, the projected rule admits the identity
\begin{equation}
  \widehat\pi^{\mathrm{agg,mix}}(t,x,y)
  =
  \pi_{\varphi^{\mathrm{warm}}}(t,x,y)
  - \widehat A_t(t,x,y)^{-1}\,\widehat r_{\mathrm{FOC}}(t,x,y),
  \label{eq:pi-agg-residual-form}
\end{equation}
which is algebraically equivalent to \eqref{eq:pi-agg-mix} (hence it does not change the target).

\medskip
\noindent\textbf{Practical value (variance reduction and stability).}
When the warm-up policy is already close to a projected fixed point on the working domain, the residual \(\widehat r_{\mathrm{FOC}}\) tends to be small.
Moreover, the ingredients entering \(\widehat A_t\pi_{\varphi^{\mathrm{warm}}}\) and \(\widehat g_t^{\mathrm{mix}}\) are computed from the same Monte Carlo pool
and often partially cancel, improving concentration of \(\widehat r_{\mathrm{FOC}}\).
In implementation, we combine \eqref{eq:pi-agg-residual-form} with the inversion stabilizers of Appendix~\ref{app:impl-stabilizers}
(e.g.\ ridge regularization and diagnostic-based fallback to \(\pi_{\varphi^{\mathrm{warm}}}\) on unreliable states).

\subsection{Interactive distillation: projection-guided training and amortized deployment}
\label{subsubsec:interactive-distillation}

Let \(\pi_\varphi\) be the trainable stage~1 policy network.
At intermittent refresh times, freeze a lagged copy \(\pi_{\varphi^-}\) and run stage~2 under \(\pi_{\varphi^-}\) to construct a \(q\)-aggregated projected teacher.
This coupling serves two purposes:
\begin{itemize}
  \item \textbf{Training aid.} Projection-guided targets can stabilize and accelerate stage~1 optimization once the warm-up policy enters a locally stable regime.
  \item \textbf{Amortization.} Stage~2 projection is accurate but Monte-Carlo intensive; distillation compresses it into a single forward pass for deployment.
\end{itemize}

\paragraph{Teacher definition (computed by stage~2; treated as fixed).}
In residual form \eqref{eq:pi-agg-residual-form}, define the teacher as the \(\theta\)-blind map
\begin{equation}
  \widehat\pi^{\mathrm{agg,mix}}_{\varphi^-}(t,x,y)
  :=
  \pi_{\varphi^-}(t,x,y)
  -
  \big(\widehat A_t^{\varphi^-}(t,x,y)\big)^{-1}\,
  \widehat r_{\mathrm{FOC}}^{\varphi^-}(t,x,y),
  \label{eq:teacher-pi-from-33-1-narr}
\end{equation}
where \(\widehat A_t^{\varphi^-}\) and \(\widehat g_t^{\mathrm{mix},\varphi^-}\) denote the stage~2 estimators computed under the lagged policy \(\pi_{\varphi^-}\),
and \(\widehat r_{\mathrm{FOC}}^{\varphi^-}:=\widehat A_t^{\varphi^-}\pi_{\varphi^-}+\widehat g_t^{\mathrm{mix},\varphi^-}\).
(We omit stabilizer notation; in code, \((\widehat A_t^{\varphi^-})^{-1}\) is implemented with ridge/diagnostics/skip as needed.)

\paragraph{Student objective (hybrid DPO + distillation penalty).}
Train \(\pi_\varphi\) by combining the ex--ante objective with a proximity penalty to the teacher on the working domain:
\begin{equation}
  \max_{\varphi}\;
  J(\varphi)
  \;-\;
  \lambda\,
  \mathbb{E}_{z=(t,x,y)\sim\mu}\Big[
    \big\|\pi_\varphi(z) - \mathrm{stopgrad}\big(\widehat\pi^{\mathrm{agg,mix}}_{\varphi^-}(z)\big)\big\|^2
  \Big],
  \label{eq:mixed-objective-with-teacher-narr}
\end{equation}
where \(\lambda\ge 0\) controls guidance strength and \(\mathrm{stopgrad}(\cdot)\) indicates that gradients are not propagated through stage~2.

\paragraph{Refresh schedule, annealing, and diagnostic gating.}
In practice, \(\varphi^-\) and the teacher are refreshed on a slower timescale than stage~1 updates:
\begin{itemize}
  \item \textbf{Two-time-scale refresh.}
  Hold \(\varphi^-\) fixed for \(n_{\mathrm{refresh}}\) stage~1 steps, then set \(\varphi^-\leftarrow\varphi\) and recompute the teacher.

  \item \textbf{Anneal \(\lambda\).}
  Start with \(\lambda=0\) (pure DPO) and increase \(\lambda\) only after projection diagnostics on the working domain stabilize.

  \item \textbf{Adaptive teacher selection (gating).}
  Apply teacher loss only on states where projection diagnostics certify reliability (stable \(\widehat A_t\) solve and moderate residual),
  falling back to pure DPO elsewhere.
\end{itemize}

\medskip
\noindent\textbf{Deployment.}
After training, the deployed controller is the student network \(\pi_\varphi\), which approximates the stage~2 projected rule with a single forward pass.

\clearpage


\clearpage
\section{Stage~2 projection diagnostics}
\label{app:stage2-diagnostics}

We report Stage~2 diagnostic statistics as a visual supplement to
Section~\ref{subsec:crra_high_dim_results}. Each figure uses the same
tail-median protocol and layout; see captions for definitions.

\providecommand{\DiagCaptionCommon}{%
All panels report tail medians over epochs 9500--10000 (final six evaluation snapshots).
Layout matches Figure~\ref{fig:crra_scaling_rmse}: rows correspond to $s\in\{10^{-3},10^{-2},10^{-1}\}$ and
columns correspond to aligned vs.\ misaligned uncertainty. Solid vs.\ dashed lines are MC base ($100\cdot d$)
vs.\ high ($400\cdot d$).}

\begin{center}
  \includegraphics[width=\textwidth,height=0.80\textheight,keepaspectratio]{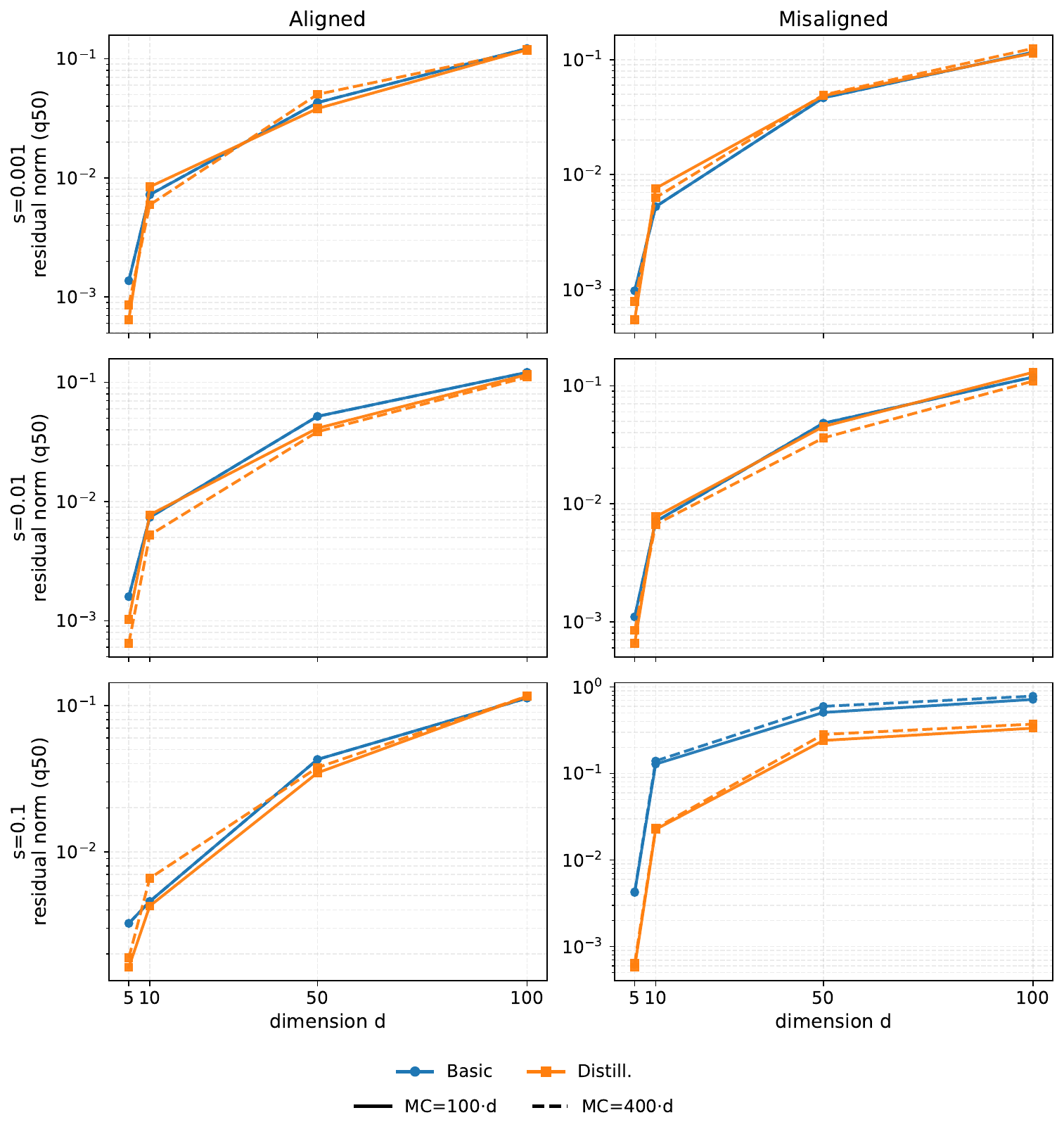}
  \captionof{figure}{\textbf{Stage~2 stationarity residual (q50).} \DiagCaptionCommon\ %
  We plot the median (q50) of the estimated Hamiltonian first-order condition residual norm at the query states.
  Larger residual indicates the warm policy is farther from stationarity, implying a larger correction is required in
  the residual-form projection. Growth of this residual with $d$ (especially under misalignment) supports the mechanism
  that projection becomes more sensitive in high dimension due to larger correction magnitudes and amplified mixed-moment noise.}
  \label{fig:app_stage2_residual_q50}
\end{center}

\clearpage
\begin{center}
  \includegraphics[width=\textwidth,height=0.80\textheight,keepaspectratio]{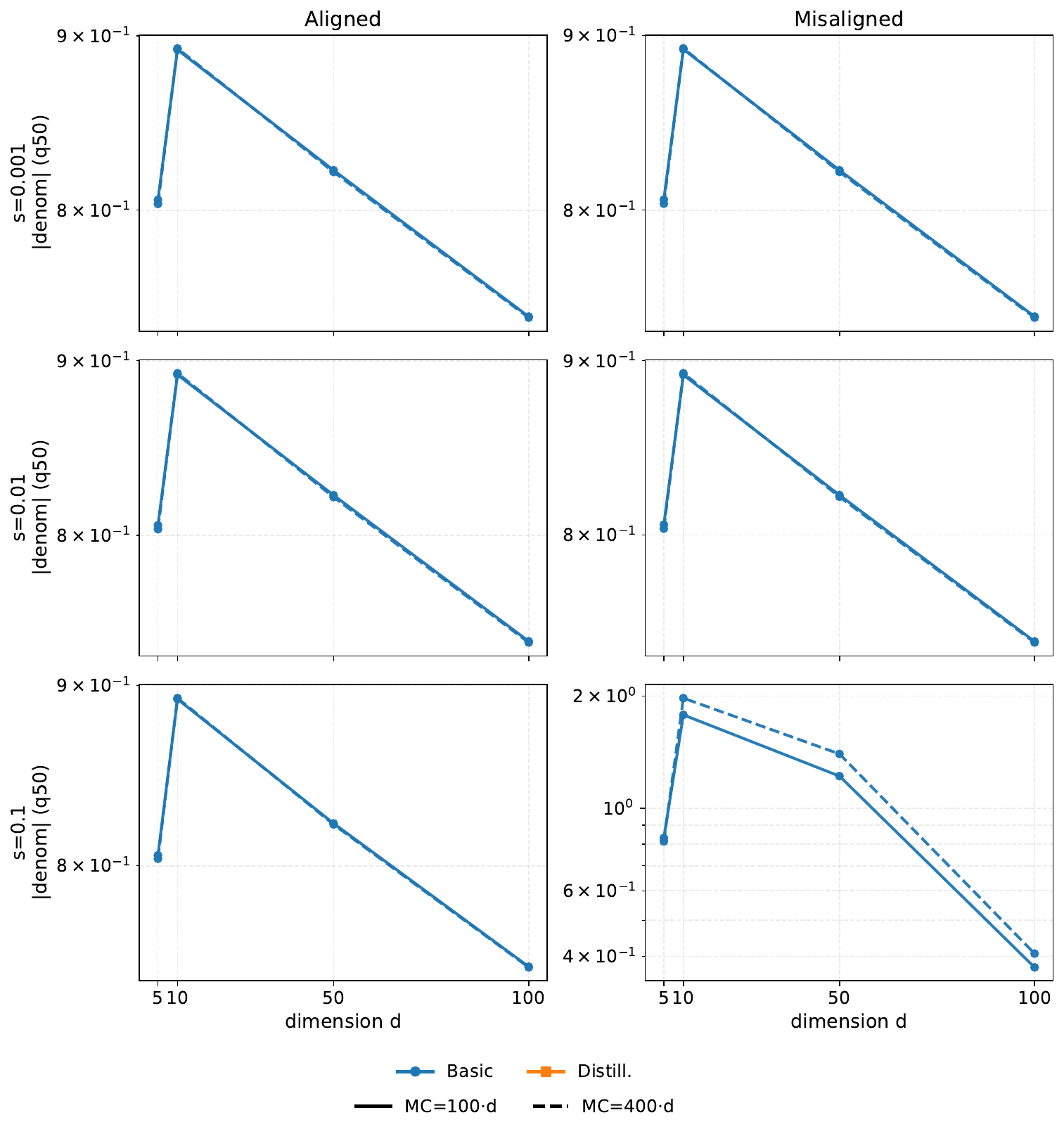}
  \captionof{figure}{\textbf{Stage~2 denominator magnitude (q50).} \DiagCaptionCommon\ %
  We plot a typical (q50) magnitude of the projection denominator/curvature term used in the residual-form update.
  Values bounded away from zero indicate that projection is not operating in a near-singular regime at typical quantiles.
  This helps rule out ``catastrophic inversion'' as the primary driver of degradation; instead, residual growth and
  curvature mismatch (Fig.~\ref{fig:app_stage2_kappa_med}) provide a more consistent explanation in misaligned/high-$d$ regimes.}
  \label{fig:app_stage2_denom_q50}
\end{center}

\clearpage
\begin{center}
  \includegraphics[width=\textwidth,height=0.80\textheight,keepaspectratio]{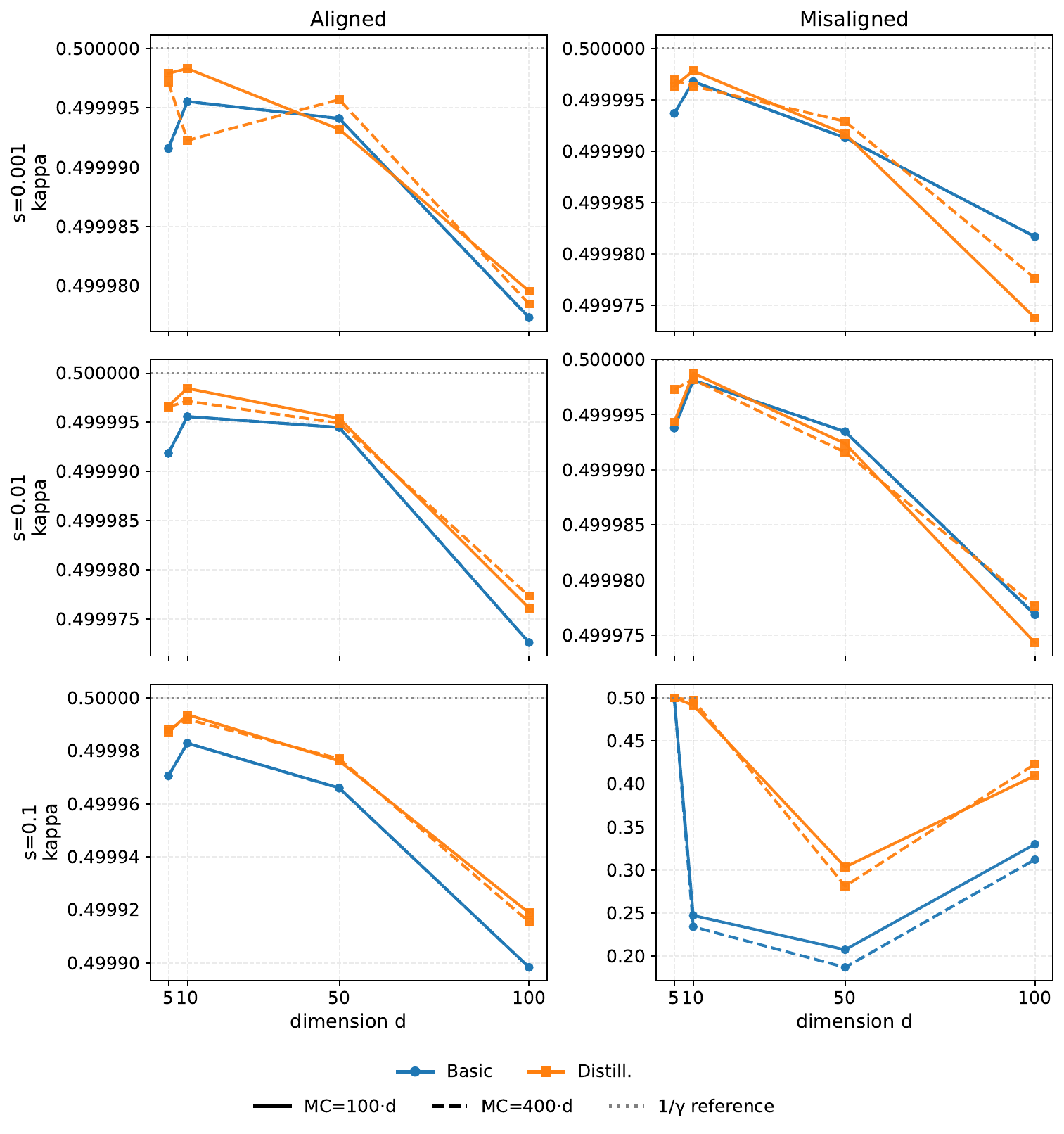}
  \captionof{figure}{\textbf{Stage~2 curvature-consistency statistic $\kappa$.} \DiagCaptionCommon\ %
  We report the stabilized median-after-floor statistic $\kappa$ and compare it to the nominal reference $1/\gamma$
  (horizontal dotted line). For CRRA, costate ratios imply a characteristic curvature scale; sustained deviations of
  $\kappa$ from $1/\gamma$ indicate costate inconsistency and/or bias in mixed-moment estimation, and are most visible
  in the hardest misaligned/high-uncertainty regime.}
  \label{fig:app_stage2_kappa_med}
\end{center}

\clearpage
\begin{center}
  \includegraphics[width=\textwidth,height=0.80\textheight,keepaspectratio]{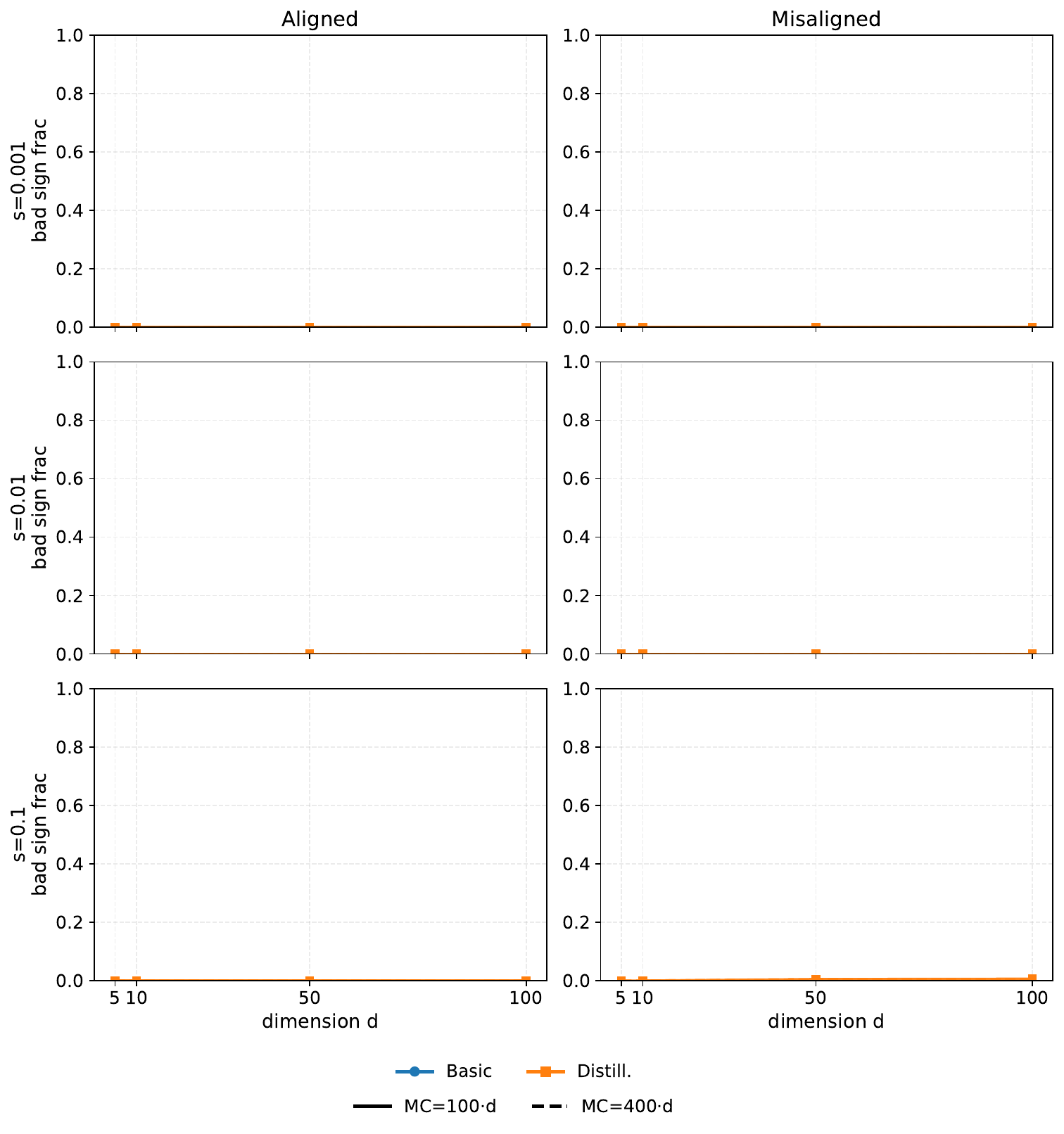}
  \captionof{figure}{\textbf{Stage~2 bad-sign fraction.} \DiagCaptionCommon\ %
  We plot the fraction of samples in which the estimated curvature/denominator violates the expected sign condition
  (loss of local concavity on the sampled batch). Near-zero bad-sign fractions across most regimes suggest that the
  projection typically operates in a locally well-behaved region and that failures are not dominated by sign flips,
  supporting the main-text conclusion that misalignment primarily increases residual/costate mismatch rather than inducing
  widespread concavity violations.}
  \label{fig:app_stage2_bad_sign_frac}
\end{center}

\clearpage


\clearpage
\section{Supplementary decomposition diagnostics for Section~\ref{sec:hedging-recovery}}
\label{app:hedging_decomp}

\noindent Tables~\ref{tab:ou_myo_rmse_s0_sweep_landscape}--\ref{tab:ou_hedge_cos_s0_sweep_landscape}
report Stage~2 decomposition diagnostics at $t=0$ (tail medians over the last six evaluation checkpoints).
Diagnostics are reported for the projected (Stage~2 / PG--DPO) rules, since Stage~1 (DPO) and PPO do not expose a compatible myopic/hedging split under our protocol.

\vspace{6pt} 

\begingroup
\sisetup{detect-weight=true,detect-family=true}
\providecommand{\best}[1]{{\bfseries\num{#1}}}

\small
\setlength{\tabcolsep}{4pt}
\renewcommand{\arraystretch}{1.05}

\begin{center}
\begin{tabular}{c l cccc}
\toprule
$s_0$ & Method & $d=5$ & $10$ & $50$ & $100$ \\
\midrule
\multicolumn{6}{l}{\textbf{Aligned $P_0$}}\\
\midrule
$10^{-3}$ & PG--DPO (basic)    & \num{7.17e-06} & \num{1.51e-05} & \num{1.60e-05} & \num{1.23e-05} \\
$10^{-3}$ & PG--DPO (distill.) & \num{5.20e-06} & \num{9.81e-06} & \num{1.64e-05} & \num{1.62e-05} \\
\midrule
$10^{-2}$ & PG--DPO (basic)    & \num{6.21e-06} & \num{1.38e-05} & \num{1.63e-05} & \num{1.42e-05} \\
$10^{-2}$ & PG--DPO (distill.) & \num{7.13e-06} & \num{7.13e-06} & \num{1.62e-05} & \num{1.76e-05} \\
\midrule
$10^{-1}$ & PG--DPO (basic)    & \num{8.18e-06} & \num{1.41e-05} & \num{3.82e-05} & \num{2.42e-05} \\
$10^{-1}$ & PG--DPO (distill.) & \num{6.72e-06} & \num{9.21e-06} & \num{3.67e-05} & \num{3.16e-05} \\
\midrule
\multicolumn{6}{l}{\textbf{Misaligned $P_0$}}\\
\midrule
$10^{-3}$ & PG--DPO (basic)    & \num{1.10e-05} & \num{1.41e-05} & \num{1.70e-05} & \num{1.18e-05} \\
$10^{-3}$ & PG--DPO (distill.) & \num{7.71e-06} & \num{5.94e-06} & \num{1.84e-05} & \num{1.62e-05} \\
\midrule
$10^{-2}$ & PG--DPO (basic)    & \num{7.90e-06} & \num{2.02e-05} & \num{1.46e-05} & \num{1.20e-05} \\
$10^{-2}$ & PG--DPO (distill.) & \num{6.00e-06} & \num{1.71e-05} & \num{2.21e-05} & \num{1.43e-05} \\
\midrule
$10^{-1}$ & PG--DPO (basic)    & \num{1.24e-05} & \num{1.93e-04} & \num{5.77e-05} & \num{3.14e-05} \\
$10^{-1}$ & PG--DPO (distill.) & \num{1.20e-05} & \num{1.90e-04} & \num{7.40e-05} & \num{2.47e-05} \\
\bottomrule
\end{tabular}

\captionof{table}{Myopic-component RMSE at $t=0$ for the projected (Stage~2 / PG--DPO) rules (tail medians).}
\label{tab:ou_myo_rmse_s0_sweep_landscape}
\end{center}

\vspace{10pt} 

\sisetup{scientific-notation=false}

\begin{center}
\begin{tabular}{c l cccc}
\toprule
$s_0$ & Method & $d=5$ & $10$ & $50$ & $100$ \\
\midrule
\multicolumn{6}{l}{\textbf{Aligned $P_0$}}\\
\midrule
$10^{-3}$ & PG--DPO (basic)    & \num{0.994} & \num{0.988} & \num{0.991} & \num{0.990} \\
$10^{-3}$ & PG--DPO (distill.) & \num{0.995} & \num{0.986} & \num{0.990} & \num{0.987} \\
\midrule
$10^{-2}$ & PG--DPO (basic)    & \num{0.993} & \num{0.989} & \num{0.992} & \num{0.988} \\
$10^{-2}$ & PG--DPO (distill.) & \num{0.992} & \num{0.994} & \num{0.990} & \num{0.987} \\
\midrule
$10^{-1}$ & PG--DPO (basic)    & \num{0.988} & \num{0.990} & \num{0.936} & \num{0.932} \\
$10^{-1}$ & PG--DPO (distill.) & \num{0.996} & \num{0.990} & \num{0.949} & \num{0.922} \\
\midrule
\multicolumn{6}{l}{\textbf{Misaligned $P_0$}}\\
\midrule
$10^{-3}$ & PG--DPO (basic)    & \num{0.988} & \num{0.988} & \num{0.993} & \num{0.990} \\
$10^{-3}$ & PG--DPO (distill.) & \num{0.994} & \num{0.995} & \num{0.990} & \num{0.987} \\
\midrule
$10^{-2}$ & PG--DPO (basic)    & \num{0.994} & \num{0.976} & \num{0.992} & \num{0.988} \\
$10^{-2}$ & PG--DPO (distill.) & \num{0.994} & \num{0.980} & \num{0.988} & \num{0.989} \\
\midrule
$10^{-1}$ & PG--DPO (basic)    & \num{0.990} & \num{0.005} & \num{0.668} & \num{0.851} \\
$10^{-1}$ & PG--DPO (distill.) & \num{0.992} & \num{-0.009} & \num{0.642} & \num{0.871} \\
\bottomrule
\end{tabular}

\captionof{table}{Hedging-direction cosine similarity at $t=0$ for the projected (Stage~2 / PG--DPO) rules (tail medians). Higher is better; negative indicates direction reversal.}
\label{tab:ou_hedge_cos_s0_sweep_landscape}
\end{center}

\endgroup
\clearpage

\end{document}